\documentclass[%
 jmp,
Conference Proceedings,
notitlepage
]{revtex4-1}

\usepackage{graphicx}
\usepackage{dcolumn}
\usepackage{bm}

\usepackage[percent]{overpic}
\usepackage{pgffor}

\usepackage[utf8]{inputenc}
\usepackage[T1]{fontenc}
\usepackage{mathptmx}

\usepackage{array}
\usepackage{ifthen}
\usepackage{scalerel}
\usepackage{tabularx}
\usepackage{multirow}
\usepackage{amsmath}
\usepackage{mathtools}
\usepackage{stmaryrd}
\usepackage{amsthm}
\usepackage{graphicx} 
\usepackage{amssymb}
\usepackage{mathrsfs}
\usepackage{subeqnarray}
\usepackage[all]{xy}
\usepackage{accents}
\usepackage{subfigure}

\usepackage[pdftex, bookmarks, colorlinks=true, plainpages = false, citecolor = red, urlcolor = black, filecolor = blue]{hyperref} 
\usepackage[usenames,dvipsnames]{color}
\usepackage[shortlabels]{enumitem}

\newcommand{\red}{\textcolor{red}}

\usepackage[titletoc]{appendix}

\newcommand{\itemcolor}[1]{
  \renewcommand{\makelabel}[1]{\color{#1}\hfil ##1}}

\usepackage{stackengine}
\stackMath

\usepackage{color}

\usepackage{thm-restate}
\DeclareMathAlphabet{\mathcal}{OMS}{cmsy}{m}{n}

\theoremstyle{plain}
\newtheorem*{theorem*}{Theorem}
\newtheorem{theorem}{Theorem}
\numberwithin{theorem}{section} 
\newtheorem{prop}[theorem]{Proposition}
\numberwithin{prop}{section} 

\numberwithin{cor}{section} 
\newtheorem{lem}[theorem]{Lemma}
\numberwithin{lem}{section} 
\newtheorem{conj}[theorem]{Conjecture}
\numberwithin{conj}{section} 

\newtheorem{quest}[theorem]{Question}
\numberwithin{quest}{section} 
\newtheorem{problem}[theorem]{Problem}
\numberwithin{problem}{section} 

\numberwithin{InductAssump}{section} 

\numberwithin{comment}{section}

\theoremstyle{definition}
\newtheorem{remark}[theorem]{Remark}
\numberwithin{remark}{section} 

\numberwithin{recipe}{section} 
\newtheorem{defn}[theorem]{Definition}
\numberwithin{defn}{section} 

\newenvironment{proofIdea}{\paragraph*{Proof (idea).}}{\hfill$\square$ \bigskip}


\newcounter{parentnumber}

\def\clap#1{\hbox to 0pt{\hss#1\hss}}

\makeatletter
\DeclareRobustCommand{\cev}[1]{%
  \mathpalette\do@cev{#1}%
}
\newcommand{\do@cev}[2]{%
  \fix@cev{#1}{+}%
  \reflectbox{$\m@th#1\vec{\reflectbox{$\fix@cev{#1}{-}\m@th#1#2\fix@cev{#1}{+}$}}$}%
  \fix@cev{#1}{-}%
}
\newcommand{\fix@cev}[2]{%
  \ifx#1\displaystyle
    \mkern#23mu
  \else
    \ifx#1\textstyle
      \mkern#23mu
    \else
      \ifx#1\scriptstyle
        \mkern#22mu
      \else
        \mkern#22mu
      \fi
    \fi
  \fi
}

\numberwithin{equation}{section}

\numberwithin{figure}{section}

\renewcommand{\thesection}{\arabic{section}} 

\makeatother

\allowdisplaybreaks

\global\long\def\SLE{\mathrm{SLE}}
\global\long\def\SLEk{\mathrm{SLE}_{\kappa}}
\global\long\def\hSLEk{\mathrm{hSLE}_{\kappa}}

\global\long\def\CLE{\mathrm{CLE}}

\global\long\def\sF{\mathcal{F}}
\global\long\def\sZ{\mathcal{Z}}

\global\long\def\sL{\mathcal{L}}

\global\long\def\sS{\mathcal{S}}

\global\long\def\bR{\mathbb{R}}
\global\long\def\bRpos{\mathbb{R}_{> 0}}

\global\long\def\bZ{\mathbb{Z}}

\global\long\def\bZpos{\mathbb{Z}_{> 0}}
\global\long\def\bZnn{\mathbb{Z}_{\geq 0}}
\global\long\def\bQ{\mathbb{Q}}
\global\long\def\bC{\mathbb{C}}

\global\long\def\bH{\mathbb{H}}

\global\long\def\ii{\mathfrak{i}}

\global\long\def\cl#1{\overline{#1}}
\global\long\def\Mob{f}

\global\long\def\primaryRep{\mathsf{V}}

\global\long\def\one{\scalebox{1.1}{\textnormal{1}} \hspace*{-.75mm} \raisebox{.025em}{|} \,}

\global\long\def\sIndex{\upsilon}

\global\long\def\OO{\mathcal{O}}

\global\long\def\ud{\mathrm{d}}

\global\long\def\pder#1{\frac{\partial}{\partial#1}}
\global\long\def\pdder#1{\frac{\partial^{2}}{\partial#1^{2}}}
\global\long\def\pddder#1{\frac{\partial^{3}}{\partial#1^{3}}}

\global\long\def\set#1{\left\{  #1\right\}  }

\global\long\def\Uqsltwo{\mathcal{U}_{q}(\mathfrak{sl}_{2})}

\global\long\def\qnum#1{\left[#1\right] }
\global\long\def\qfact#1{\left[#1\right]! }

\global\long\def\swaltime{\tau}

\global\long\def\LP{\mathrm{LP}}

\global\long\def\Catalan{\mathrm{C}}

\global\long\def\constantfromdiagram#1#2#3{A_{#1}^{#2,#3}}

\global\long\def\linkpatt{\omega}

\global\long\def\hF{{}_2F_1}
\global\long\def\graph{G}

\global\long\def\link#1#2{\{#1,#2\}}

\global\long\def\removeLink{/}

\global\long\def\SymmGrp{\mathfrak{S}}

\newcommand{\edgeof}[2]{{\langle #1 , #2 \rangle}}

\global\long\def\dmn{\mathrm{dim}\,}

\global\long\def\id{\mathrm{id}}

\global\long\def\SymmGrp{\mathfrak{S}}

\global\long\def\Vir{\mathfrak{Vir}}

\global\long\def\chamber{\mathfrak{X}}
\global\long\def\extendedChamber{\mathfrak{W}}

\global\long\def\PartF{\sZ}
\global\long\def\Sol{\sS}

\global\long\def\multii{\vartheta}
\global\long\def\multidim{\overline{s}}

\global\long\def\projdmn{s}

\global\long\def\BasisF{\mathscr{F}}

\global\long\def\indexSet{I}
\global\long\def\index{\iota}

\begin{document}


\title[Towards a conformal field theory for Schramm-Loewner evolutions]{\vspace*{1.5cm}Towards a conformal field theory for Schramm-Loewner evolutions \vspace*{.5cm}}

\author{\bf Eveliina Peltola}
\affiliation{{\color{blue}{\tt{\small eveliina.peltola@hcm.uni-bonn.de}}} \\ 
Institute for Applied Mathematics, University of Bonn, \\
Endenicher Allee 60, D-53115 Bonn, Germany
\vspace*{.5cm}}


%



\begin{abstract}
We discuss the partition function point of view for chordal Schramm-Loewner evolutions 
and their relationship with correlation functions in conformal field theory. 
Both are closely related to crossing probabilities and interfaces in critical models in two-dimensional statistical mechanics.
We gather and supplement previous results with different perspectives,
point out remaining difficulties, and suggest directions for future studies.
\end{abstract}

\maketitle

\renewcommand{\tocname}{}
{\hypersetup{linkcolor=black}
\tableofcontents
}


%

%

\section{\label{sec:intro}Introduction}

The general aim of this article is to illustrate some features of the connection of 
critical models of statistical mechanics with conformal field theory, 
i.e., conformally invariant quantum field theory. 
One way to mathematically formulate such a connection is in terms of random geometry, 
where topological or geometric properties of the models are associated to conformally invariant objects.
Recently, this approach has been very successful for two-dimensional systems:
examples include the conformal invariance of crossing probabilities in critical 
models~\cite{Cardy:Critical_percolation_in_finite_geometries, 
LPS:Conformal_invariance_in_2d_percolation, 
LLS:Conformal_invariance_in_Ising_model,
Smirnov:Critical_percolation_in_the_plane}, 
their relationship with correlation functions in conformal field theory
(see~\cite{BBK:Multiple_SLEs_and_statistical_mechanics_martingales,
Izyurov:Smirnovs_observable_for_free_boundary_conditions_interfaces_and_crossing_probabilities,
FSKZ-A_formula_for_crossing_probabilities_of_critical_systems_inside_polygons, 
Peltola-Wu:Crossing_probabilities_of_multiple_Ising_interfaces}, 
and references therein), 
the description of critical planar interfaces 
in terms of conformally invariant random curves (Schramm-Loewner evolutions) 
\cite{Schramm:Scaling_limits_of_LERW_and_UST, 
Smirnov:Critical_percolation_in_the_plane,
LSW:Conformal_invariance_of_planar_LERW_and_UST,
Smirnov:Towards_conformal_invariance_of_2D_lattice_models,
Schramm-Sheffield:Contour_lines_of_2D_discrete_GFF,
Schramm-Sheffield:A_contour_line_of_the_continuum_GFF,
CDHKS:Convergence_of_Ising_interfaces_to_SLE},
and a random geometry formulation of 2D quantum 
gravity~\cite{Polyakov:Quantum_geometry_of_bosonic_strings,
Duplantier:Conformal_fractal_geometry_and_boundary_quantum_gravity, 
LeGall:Uniqueness_and_universality_of_the_Brownian_map,
Miermont:Brownian_map_is_scaling_limit_of_uniform_random_plane_quadrangulations,
Sheffield-Miller:Imaginary_geometry1,
DMS:Liouville_quantum_gravity_as_mating_of_trees}.
In~this article, we focus on the relationship of Schramm-Loewner evolutions 
with correlation functions in conformal field theory. 

\bigskip

Quantum field theory is manifest in particle physics and condensed matter physics:
it describes, for instance, interactions in electromagnetic theory, the standard model, 
and many-body systems. 
The basic objects, ``fields'', have infinitely many degrees of freedom and they might not admit a mathematically precise meaning.
The observable quantities are ``averages'' (expectation values) of the fields, usually 
termed correlation functions~\cite{DMS:CFT, Schottenloher:Mathematical_introduction_to_CFT, Mussardo:Statistical_field_theory}. 
Quantum field theory is believed to also describe scaling limits of many lattice models of statistical mechanics
(lattice models are formulated on discretizations of the space, lattices, 
and the limit when the mesh of the lattice tends to zero is called the scaling limit).
In general, statistical mechanics 
concerns systems with a large number of degrees of freedom, such as gases, liquids, and solids.
The key objective is to derive a macroscopic description of the system 
(which could perhaps be concretely observed) 
via a suitable probability distribution for the microscopic states, 
which due to the enormous number of variables 
cannot be deterministically analyzed, 
see~\cite{Mussardo:Statistical_field_theory, Friedli-Velenik:Statistical_mechanics_of_lattice_systems}.

Of particular interest to us are statistical models which exhibit continuous (second order) phase transitions
--- abrupt changes of macroscopic properties when, e.g., the temperature of the system is varied continuously. 
An example of such a phenomenon is the loss of magnetization in a ferromagnet when it is heated above the Curie temperature
(in dimension at least two); see Figure~\ref{fig: Ising phase transition}.
The value of the temperature at which the phase transition occurs is called critical.
A common feature of critical phenomena in continuous phase transitions is that the characteristic length scale of the system, 
the correlation length $\xi(T)$, diverges as the temperature $T$ approaches its critical value $T_c$.
For instance, in the ferromagnet, the characteristic length scale is described by the decay of correlations
$C_T(x,y):= \mathbb{E}[\sigma_x \sigma_y] - \mathbb{E}[\sigma_x] \mathbb{E}[\sigma_y]$ 
of two atomic spins $\sigma_x$ and $\sigma_y$ at positions $x$ and $y$ far apart:
at very high temperatures, thermal fluctuations 
overcome the spins' interactions and the correlations decay exponentially fast: 
$C_T(x,y) \sim e^{-|x-y| / \xi(T)}$ as $|x-y| \to \infty$.
On the other hand, when $T \searrow T_c$, we have $\xi(T) \to \infty$. 
At criticality $T = T_c$, the correlations decay according to a power law:  
$C_{T_c}(x,y) \sim |x-y|^{-2 \Delta}$ as $|x-y| \to \infty$,
where $\Delta$ is a critical exponent for the model.

Scaling limits of the above type of models 
at criticality should be scale-invariant, as the divergence of the correlation length indicates 
(more formally, the scaling limit is described by a fixed point of the renormalization group 
flow, see, e.g.,~\cite{Cardy:Scaling_and_renormalization_in_statistical_physics}).
A.~Polyakov~\cite{Polyakov:Conformal_symmetry_of_critical_fluctuations} 
conjectured in the 1970s that these models  
should even enjoy a much stronger symmetry, conformal invariance.
In the 1980s, convincing physical arguments for the conformal invariance
were indeed given for two-dimensional systems by Polyakov with A.~Belavin  
and A.~Zamolodchikov~\cite{BPZ:Infinite_conformal_symmetry_in_2D_QFT, 
BPZ:Infinite_conformal_symmetry_of_critical_fluctuations_in_2D}, 
and later by J. Polchinski~\cite{Polchinski:Scale_and_conformal_invariance_in_quantum_field_theory}.
Specifically, in the scaling limit, a critical lattice model with continuous phase transition 
should converge to some conformal field theory (CFT), 
regardless of the precise microscopic details of the model 
(e.g., choice of lattice, see~\cite{Chelkak-Smirnov:Universality_in_2D_Ising_and_conformal_invariance_of_fermionic_observables},
or exact interaction range, see~\cite{GGM:The_scaling_limit_of_the_energy_correlations_in_non_integrable_Ising_models}). 
Also, the models should form universality classes, within which macroscopic properties,
such as decay of correlations and critical exponents, are similar.

In two dimensions, supplementing the global conformal symmetry,
Belavin, Polyakov, and Zamolodchikov~\cite{BPZ:Infinite_conformal_symmetry_in_2D_QFT}
also postulated invariance under ``infinitesimal''   conformal transformations, 
yielding infinitely many conserved quantities (instead of fixing only finitely many degrees of freedom, as the global conformal symmetry does). 
This idea had striking implications: the universality classes are classified by one 
parameter $c$, the central charge of the CFT; the CFTs form representations of the Virasoro algebra, 
the conformal symmetry algebra of the plane; and
the representations of this algebra were completely classified by 
B.~Fe{\u\i}gin and D.~Fuchs~\cite{Feigin-Fuchs:Representations_of_Virasoro}.
Thus, the two-dimensional CFTs could be analyzed in great detail. 
Further developments include J.~Cardy's introduction of CFTs with boundary~\cite{Cardy:Conformal_invariance_and_surface_critical_behavior, 
Cardy:Boundary_conditions_fusion_rules_and_Verlinde_formula, 
Cardy:Critical_percolation_in_finite_geometries}
to understand surface critical phenomena and the effect of boundary conditions,
as well as B.~Nienhuis's Coulomb gas formulation for phase 
transitions~\cite{Nienhuis:Exact_critical_point_and_exponents_of_the_On_model_in_two_dimensions,
Nienhuis:Critical_behavior_of_two-dimensional_spin_models_and_charge_asymmetry_in_the_Coulomb_gas,
Nienhuis:Coulomb_gas_formulation_of_2D_phase_transitions}, 
giving new predictions for, e.g., the values of critical exponents, many of which still remain extremely challenging for mathematicians.

\begin{figure}
\centering
\includegraphics[width=.25\textwidth]{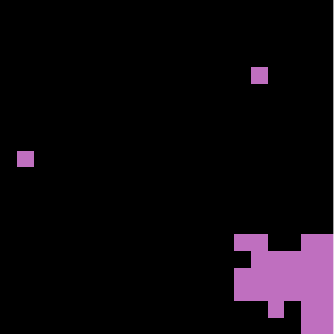}
\qquad\qquad
\includegraphics[width=.25\textwidth]{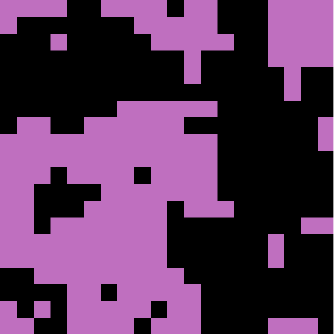}
\qquad\qquad
\includegraphics[width=.25\textwidth]{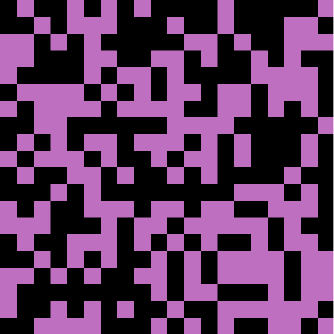}
\caption{\label{fig: Ising phase transition}
The phase transition in the ferromagnetic Ising model. 
In high temperatures (right), the system is disordered (paramagnetic) and spins at far away points almost independent
(i.e., correlations of spins decay exponentially fast in the distance).
In low temperatures (left), typical configurations are ordered (ferromagnetic) and the system is strongly correlated even at long distances.
At the unique critical temperature $T_c$, macroscopic clusters of both spins appear,
the system does not have a typical length scale, and correlations decay polynomially in the distance.
As the lattice mesh tends to zero, this critical system should be described by 
a conformally invariant quantum field theory.}
\end{figure}

A major breakthrough in mathematics relating conformal invariance and critical phenomena was
the introduction of  stochastic Loewner evolutions, now known as Schramm-Loewner evolutions (SLE),
in the seminal work~\cite{Schramm:Scaling_limits_of_LERW_and_UST} of O.~Schramm. 
The $\SLE_\kappa$ is a one-parameter family 
of random planar curves indexed by $\kappa \geq 0$ 
(the speed of the curve when viewed as a growth process driven by Brownian motion),
which is uniquely characterized by its conformal invariance and a Markovian property.
Schramm's idea led to remarkable success:
with G.~Lawler and W.~Werner, Schramm calculated critical exponents for planar Brownian 
motion~\cite{LSW:Brownian_intersection_exponents1, LSW:Brownian_intersection_exponents2}
and proved one of the first results towards conformal invariance of critical models 
in statistical mechanics~\cite{LSW:Conformal_invariance_of_planar_LERW_and_UST}:
SLE curves indeed describe scaling limits of interfaces for certain polymer models 
(loop-erased walks and uniform spanning trees).
Also, critical exponents for percolation were rigorously derived using 
SLE~\cite{Smirnov-Werner:Critical_exponents_for_two-dimensional_percolation, LSW:One-arm_exponent_for_2D_critical_percolation}.

Around that time, S.~Smirnov and R.~Kenyon
independently and ingeniously 
used discrete complex analysis to establish more results on conformal invariance of scaling limits of critical planar models:
convergence of the dimer model height function to the Gaussian free field (``free boson'') by 
Kenyon~\cite{Kenyon:The_asymptotic_determinant_of_the_discrete_Laplacian, Kenyon:Conformal_invariance_of_domino_tiling, 
Kenyon:Dominos_and_the_Gaussian_free_field},
conformal invariance for the exploration process and crossing probabilities in critical percolation 
by Smirnov~\cite{Smirnov:Critical_percolation_in_the_plane}, 
extended  by F.~Camia and 
C.~Newman~\cite{Camia-Newman:2D_percolation_full_scaling_limit,
Camia-Newman:Critical_percolation_exploration_path_and_SLE_proof_of_convergence} 
to include the collection of loops (cluster boundaries inside the domain),
and later, conformal invariance for the critical Ising and FK-Ising models
by Smirnov et.~al~\cite{Smirnov:Towards_conformal_invariance_of_2D_lattice_models, Smirnov:Conformal_invariance_in_random_cluster_models1,
Chelkak-Smirnov:Universality_in_2D_Ising_and_conformal_invariance_of_fermionic_observables,
Hongler-Smirnov:Energy_density_in_planar_Ising_model, 
CDHKS:Convergence_of_Ising_interfaces_to_SLE, CHI:Conformal_invariance_of_spin_correlations_in_planar_Ising_model},
in terms of correlations and interfaces.

Physicists also became very interested in SLEs.
Indeed, Schramm's ideas were novel, providing a different approach to
understanding critical phenomena in relation with quantum field theory, especially CFT. 
After the introduction of SLEs, J.~Cardy  soon  predicted a relationship between SLE curves and certain ``boundary condition changing operators'' 
in critical models~\cite{Cardy:SLE_and_Dyson_circular_ensembles, Cardy:SLE_for_theoretical_physicists}. 
This was formalized by M.~Bauer and D.~Bernard~\cite{Bauer-Bernard:Conformal_field_theories_of_SLEs, 
Bauer-Bernard:SLE_martingales_and_Virasoro_algebra,
Bauer-Bernard:Conformal_transformations_and_SLE_partition_function_martingale}, 
who argued in particular that certain CFT correlation functions are related to 
martingales for the SLE curves, and 
there must be a specific relationship between the $\SLEk$
and the central charge $c(\kappa)$ of the CFT.
Thus, conjecturally, certain CFT fields  
should correspond to the growth of SLE curves.

Since then, many variants of SLEs have been rigorously related to critical models,
thus verifying their conformal invariance in the scaling 
limit~\cite{Smirnov:Critical_percolation_in_the_plane, 
LSW:Conformal_invariance_of_planar_LERW_and_UST, Camia-Newman:2D_percolation_full_scaling_limit, 
Zhan:Scaling_limits_of_planar_LERW_in_finitely_connected_domains,
Schramm-Sheffield:Contour_lines_of_2D_discrete_GFF, Hongler-Kytola:Ising_interfaces_and_free_boundary_conditions,
Schramm-Sheffield:A_contour_line_of_the_continuum_GFF,
CDHKS:Convergence_of_Ising_interfaces_to_SLE, 
Izyurov:Smirnovs_observable_for_free_boundary_conditions_interfaces_and_crossing_probabilities, 
BPW:On_the_uniqueness_of_global_multiple_SLEs}. 
However, these limits as conformal (quantum) field theories are still not mathematically well understood.
From the SLE point of view, so-called partition 
functions~\cite{BBK:Multiple_SLEs_and_statistical_mechanics_martingales, 
Dubedat:Commutation_relations_for_SLE,
Lawler:Partition_functions_loop_measure_and_versions_of_SLE,
Dubedat:SLE_and_Virasoro_representations_fusionB}
can be abstractly viewed as CFT correlation functions. 
We will see how such a connection also makes mathematical sense, even though the ``SLE generating fields'' themselves might not.

\bigskip

{\bf Role of this article.}
The main goal is to shed light on the  
connection of SLE curves with certain CFT correlation functions, probabilistically known as SLE partition functions 
(i.e., ``total masses'' for the measures on curves).
We also discuss the role of  the ``SLE generating fields'' which should be associated to these correlation functions,
but cautiously note that the mathematical meaning of such fields is not clear,
whereas the correlation functions are both well-defined and quite well understood.

The SLE partition functions can be studied in terms of a hypoelliptic PDE system.
Such PDEs are well  known in the CFT literature for correlation functions of so-called degenerate conformal fields. 
Notably, exactly the same PDEs also follow by purely probabilistic arguments from SLE martingales, or viewing the SLEs as 
hypoelliptic diffusion processes~\cite{Kontsevich:CFT_SLE_and_phase_boundaries, 
Kontsevich-Suhov:On_Malliavin_measures_SLE_and_CFT,
Dubedat:SLE_and_Virasoro_representations_localizationA}.
In particular, strong classification results for these functions can be 
established~\cite{Flores-Kleban:Solution_space_for_system_of_null-state_PDE1, Flores-Kleban:Solution_space_for_system_of_null-state_PDE2,
Flores-Kleban:Solution_space_for_system_of_null-state_PDE3, Flores-Kleban:Solution_space_for_system_of_null-state_PDE4}.

In fact, such a classification is not only interesting from the field theoretical point of view, 
but also regarding the SLE processes themselves, 
and especially their relation with interfaces and crossing probabilities in critical statistical mechanics models.
Indeed, different connectivity patterns of multiple SLEs can be encoded in so-called pure partition functions,
which form a distinguished basis in the space of SLE partition 
functions~\cite{Kytola-Peltola:Pure_partition_functions_of_multiple_SLEs,
Peltola-Wu:Global_and_local_multiple_SLEs_and_connection_probabilities_for_level_lines_of_GFF}.
These basis functions, in turn, 
are also naturally related to probabilities of non-local crossing events, e.g., for the critical Ising 
model~\cite{Peltola-Wu:Crossing_probabilities_of_multiple_Ising_interfaces, KKP:Conformal_blocks_pure_partition_functions_and_KW_binary_relation}.

Finally, these 
functions admit a beautiful hierarchy of fusion rules,
which can be thought of as a rigorous operator product expansion, 
one of the cornerstones of conformal field theory.
In fact, in some cases the fusion can also be related to actual observables in critical models~\cite{Gamsa-Cardy:The_scaling_limit_of_two_cluster_boundaries_in_critical_lattice_models, 
KKP:Conformal_blocks_pure_partition_functions_and_KW_binary_relation},
$\SLE$ observables~\cite{BJV:Some_remarks_on_SLE_bubbles_and_Schramms_2point_observable,
Lenells-Viklund:Coulomb_gas_integrals_for_commuting_SLEs},
or generalizations of multiple $\SLE$ 
measures~\cite{Friedrich-Werner:Conformal_restriction_highest_weight_representations_and_SLE,
Kontsevich:CFT_SLE_and_phase_boundaries,
Friedrich-Kalkkinen:On_CFT_and_SLE,
Kontsevich-Suhov:On_Malliavin_measures_SLE_and_CFT,
Dubedat:SLE_and_Virasoro_representations_fusionB}. 
See also~\cite{BPZ:Infinite_conformal_symmetry_of_critical_fluctuations_in_2D, 
Cardy:Critical_percolation_in_finite_geometries,  
Watts:A_crossing_probability_for_critical_percolation_in_two_dimensions,
Bauer-Bernard:Conformal_field_theories_of_SLEs,
Bauer-Bernard:SLE_CFT_and_zigzag_probabilities, 
Dubedat:Euler_integrals_for_commuting_SLEs,
Dubedat:Excursion_decomposition_for_SLE_and_Watts_crossing_formula,
Sheffield-Wilson:Schramms_proof_of_Watts_formula,
Flores-Kleban:Solution_space_for_system_of_null-state_PDE4,
FSK:Multiple_SLE_connectivity_weights_for_rectangles_hexagons_and_octagons,
JJK:SLE_boundary_visits, Peltola-Wu:Global_and_local_multiple_SLEs_and_connection_probabilities_for_level_lines_of_GFF} for further examples.

\bigskip

{\bf Organization of this article.}
We discuss the SLE, its relation to critical lattice models, and basics of CFT in Section~\ref{sec: preli}.
Specifically, Section~\ref{subsec:SLE} contains the definition and basic properties of the SLE.
In Section~\ref{subsec: Ising}, we introduce  the Ising model as an example of a critical lattice model in statistical mechanics.
In Section~\ref{subsec:CFT}, we review some basic features of two-dimensional CFT, and in Section~\ref{subsec:mgles},
we explain how lattice interfaces and SLEs could be related to CFT correlation functions via martingale observables.
Sections~\ref{subsec:CFT}--\ref{subsec:mgles} are not intended to be mathematically precise, 
but rather to serve as motivation and illustration. As supporting material,
Appendix~\ref{app:Vir} contains some representation theory of the Virasoro algebra.

In Section~\ref{sec: Multiple SLE partition functions}, we introduce the SLE partition functions.
First, in Section~\ref{subsec:multiple SLEs} we briefly discuss multiple SLEs and the notion of an SLE partition function.
Then, in Sections~\ref{subsec:ppfdef} and~\ref{subsec:ppfprop}  we give a PDE theoretic 
definition for the multiple SLE partition functions and discuss their most important properties.
Last, in Section~\ref{subsec:scaling_limimt_results_etc}, 
we  briefly discuss applications 
to the theory of SLEs as well as to the conformal invariance for critical models.


Section~\ref{sec:OPE} concerns an operator product expansion
(OPE) 
for the SLE partition functions --- a fusion procedure to generate other CFT correlation functions 
from the functions of Section~\ref{sec: Multiple SLE partition functions}. 
We begin in Section~\ref{subsec:FusionCFT} with a brief  and heuristic summary of the role of the OPE in CFT,
following the physics literature. In Sections~\ref{subsec: fusion dub} and~\ref{subsec: fusion SCCG},
we discuss two possible approaches to make the OPE structure for the SLE partition functions
mathematically well-defined. In Section~\ref{subsec:OPE for ppf}, we state a rather general result to this end.

Section~\ref{sec:Malek} is devoted to some speculations on how the OPE structure from  Section~\ref{sec:OPE}
could be useful 
for constructive field theory, based on ideas presented recently in~\cite{Abdesselam:Second-quantized_Kolmogorov-Chentsov_theorem}. 
In Section~\ref{subsec: Random tempered distributions}, we briefly discuss one way to make mathematical sense of the ``fields'' in quantum field theory
as random distributions. In Sections~\ref{subsec: ASPWC} and~\ref{subsec: Bootstrap for SLEs}, 
we outline how the OPE structure from Section~\ref{sec:OPE}
could perhaps be used to try and understand the ``SLE generating fields'' mathematically. 
The goal of this last section is to open some perspectives and to rise questions for future developments in the field.

\bigskip
\bigskip

\begin{center}
\bf Acknowledgements
\end{center}

The purpose of this paper is to summarize several joint works from my personal perspective.
The presentation is intended to give an overview 
of work distributed in many articles, as well as to gather and clarify 
known results and remaining problems.
The occasionally heuristic motivations and speculations are supposed 
to serve as motivation rather than exposition.
Also, I have omitted many important related works, to which the given citations should guide the reader. 

I have benefited enormously from discussions with numerous people, 
and trying to list all of them would only result in forgetting to mention many important names.
Special thanks belong to Julien Dub\'edat and Greg Lawler for inspiration and interesting discussions, 
as well as to Vincent Beffara, Steven Flores, Alex Karrila, Kalle Kyt\"ol\"a, and Hao Wu 
for fruitful collaboration and discussions.
I have enjoyed very much my collaboration with Hao Wu
on the probabilistic approach to the multiple SLE partition functions and their relation with critical lattice models,
discussed in Section~\ref{sec: Multiple SLE partition functions} and Appendices~\ref{app:Hao} \&~\ref{app}.
Of special importance to Section~\ref{sec:OPE} is my joint work with Kalle Kyt\"ol\"a,
whom I would like to thank also for introducing me to the subject in the first place.
%
Finally, I wish to thank Abdelmalek Abdesselam for pointing out a connection to his work,
that I briefly discuss in Section~\ref{sec:Malek}.

I cordially acknowledge the financial support of
the ERC AG COMPASP, the NCCR SwissMAP, and the Swiss~NSF.


\section{\label{sec: preli} Schramm-Loewner evolution in statistical mechanics and conformal field theory}

In this section, we introduce Schramm-Loewner evolutions (SLE) and  describe how they are connected to 
lattice models of statistical mechanics and conformal field theory (CFT).
We focus on the case of planar domains with boundary and consider chordal interfaces, 
neglecting many (also interesting) phenomena in the bulk.
We omit altogether, for example, the conformal loop ensembles 
(CLE)~\cite{Sheffield:Exploration_trees_and_CLEs, Sheffield-Werner:CLEs},
Brownian excursions and loops and conformal restriction measures~\cite{LSW:Conformal_restriction_the_chordal_case} 
related to the stress-energy 
tensor~\cite{Friedrich-Werner:Conformal_restriction_highest_weight_representations_and_SLE,
Friedrich-Kalkkinen:On_CFT_and_SLE,
DRC:Identification_of_the_stress-energy_tensor_through_conformal_restriction_in_SLE_and_related_processes,
Doyon-Random_loops_and_conformal_field_theory},
as well as the case of  general Riemann 
surfaces~\cite{Kontsevich:CFT_SLE_and_phase_boundaries, 
Kontsevich-Suhov:On_Malliavin_measures_SLE_and_CFT, Dubedat:SLE_and_Virasoro_representations_localizationA}.
One of the first celebrated applications of SLE was the rigorous calculation of critical 
exponents~\cite{LSW:Brownian_intersection_exponents1, LSW:Brownian_intersection_exponents2, 
Smirnov-Werner:Critical_exponents_for_two-dimensional_percolation, LSW:One-arm_exponent_for_2D_critical_percolation},
in agreement with the earlier predictions in the physics 
literature~\cite{denNijs:Extended_scaling_relations_for_the_magnetic_critical_exponents_of_the_Potts_model,
BPZ:Infinite_conformal_symmetry_in_2D_QFT,
BPZ:Infinite_conformal_symmetry_of_critical_fluctuations_in_2D,
Cardy:Conformal_invariance_and_surface_critical_behavior,
Dotsenko-Fateev:Conformal_algebra_and_multipoint_correlation_functions_in_2D_statistical_models,
Duplantier-Saleur:Exact_determination_of_the_percolation_hull_exponent_in_two_dimensions,
Nienhuis:Coulomb_gas_formulation_of_2D_phase_transitions}.
There is also an interesting connection of SLEs with Liouville  theory of
gravity~\cite{Duplantier:Conformal_fractal_geometry_and_boundary_quantum_gravity,  
Sheffield-Miller:Imaginary_geometry1,
DMS:Liouville_quantum_gravity_as_mating_of_trees}.
For these developments, we invite the reader to consult the aforementioned papers and references therein.

The obvious relation of SLE curves with lattice models is rather geometric ---
SLEs describe interfaces, or domain walls, of critical planar lattice models in the scaling limit (i.e, as the lattice mesh tends to zero).
In general, these models are believed to be described by conformally invariant quantum field theories, CFTs, in the continuum.
However, mathematical understanding of such a statement remains unclear
and is one of the major challenges in modern mathematical physics.
On the other hand, martingale observables for SLE curves are closely related to certain correlation functions in 
CFT, which can be mathematically defined as real or complex analytic functions.
One of the goals of the present article is to shed light on this latter connection.

\bigskip

We begin in Section~\ref{subsec:SLE} with the introduction of the chordal SLE and discuss some of its main features.
Then, in Section~\ref{subsec: Ising} we make connection with lattice models, taking as an example the critical planar Ising model, 
for which many important results have been rigorously obtained. 
Analogous results have also been proved or conjectured for many other critical models~\cite{Schramm:ICM}:
percolation, self-avoiding and loop-erased walks, Potts model, $O(n)$-model, random-cluster model, Gaussian free field, etc.

Section~\ref{subsec:CFT} contains a very brief and incomplete introduction to some aspects of conformal field theory,
important for the purposes of the present article.
Then, in Section~\ref{subsec:mgles} we discuss martingale observables and describe how 
the two fundamental properties of SLE, conformal invariance and the domain Markov property, 
give rise to a prediction that certain conformal fields, denoted ``$\Phi_{1,2}$'', 
should be associated to the growth of SLE curves from the boundary.
The discussion in these two subsections is not intended to be rigorous, but rather to serve as motivation and illustration,
and to provide ideas and background from physics.

To keep the discussion brief and intuitive, 
most of this section is presented in a rather informal manner, 
one reason being that the mathematical content of some of the statements is not yet fully understood,
and another that we wish to avoid the technical (although important) points.
There already exists an extensive literature,  
to which we give references along the way.

\subsection{\label{subsec:SLE}Schramm-Loewner evolution (SLE)}

The Schramm-Loewner evolutions, originally called ``stochastic'' Loewner evolutions, were introduced 
at the turn of the millennium
by O.~Schramm~\cite{Schramm:Scaling_limits_of_LERW_and_UST}, 
who argued that they are the only 
possible random curves that could describe scaling limits of critical lattice interfaces in  two-dimensional systems.
Schramm's definition was inspired by the classical theory of C.~Loewner~\cite{Loewner:Untersuchungen_uber_schlichte_konforme_Abbildungen_des_Einheitskreises} 
for dynamical description of the growth of hulls, encoded in conformal maps. Schramm's revolutionary input was that such maps
could also be random.
Aiming at the construction of scaling limits of critical lattice interfaces, 
the law of the SLE curve should be manifestly conformally invariant. 
Schramm observed in~\cite{Schramm:Scaling_limits_of_LERW_and_UST} 
that when requiring in addition a Markovian property for the growth of the curve, 
there is only a one-parameter family of such random curves, that he labeled by $\kappa \geq 0$
and called the $\SLE_\kappa$.
Physically, the parameter $\kappa$ describes the universality class
of the corresponding critical model, or equivalently, the central charge of the corresponding conformal field 
theory~\cite{Cardy:Scaling_and_renormalization_in_statistical_physics, Cardy:SLE_for_theoretical_physicists}. 
Mathematically, $\kappa$ is the ``speed'' of the Brownian motion associated to the growth of the $\SLE_\kappa$ curve; 
see Figure~\ref{fig: SLE} and the construction below Definition~\ref{defn:SLE}.

\bigskip

\begin{defn} \label{defn:SLE}
For $\kappa \geq 0$, the (chordal) Schramm-Loewner evolution $\SLE_{\kappa}$ 
is a family of probability measures $\mathbb{P}_{\Omega; x,y}$ on curves, indexed by simply connected domains $\Omega \subsetneq \bC$
with two distinct boundary points $x,y \in \partial \Omega$.
Each measure $\mathbb{P}_{\Omega; x,y}$ is supported on continuous unparameterized curves in $\overline{\Omega}$ from $x$ to $y$.
This family is uniquely determined by the following two properties: 
\begin{itemize}

\item {\bf \textit{Conformal invariance}}: 
Fix two simply connected domains $\Omega, \Omega' \subsetneq \bC$
and boundary points $x,y \in \partial \Omega$ and $x',y' \in \partial \Omega'$,
with $x \neq y$ and $x' \neq y'$. 
According to the Riemann mapping theorem,
there exists a conformal bijection $\Mob \colon \Omega \to \Omega'$ such that $\Mob(x)=x'$ and $\Mob(y)=y'$.
With any choice of such a map, we have $\Mob(\eta) \sim \mathbb{P}_{\Omega'; x',y'}$ if $\eta \sim \mathbb{P}_{\Omega; x,y}$.

\item {\bf \textit{Domain Markov property}}: 
Given an initial segment $\eta[0,\tau]$ of the $\SLE_\kappa$ curve $\eta \sim \mathbb{P}_{\Omega; x,y}$
up to a stopping time $\tau$ (parameterizing $\eta$ by $[0,\infty)$, say), 
the conditional law of the remaining piece $\eta[\tau,\infty)$ is the law $\mathbb{P}_{\Omega_\tau; \eta(\tau),y}$ of 
the $\SLE_\kappa$ from the tip $\eta(\tau)$ to $y$ in the 
component $\Omega_\tau$ of the complement $\Omega \setminus \eta[0,\tau]$ of the initial segment containing the target point $y$ on its boundary.
\end{itemize}
\end{defn}

\begin{figure}
\centering
\includegraphics[width=\textwidth]{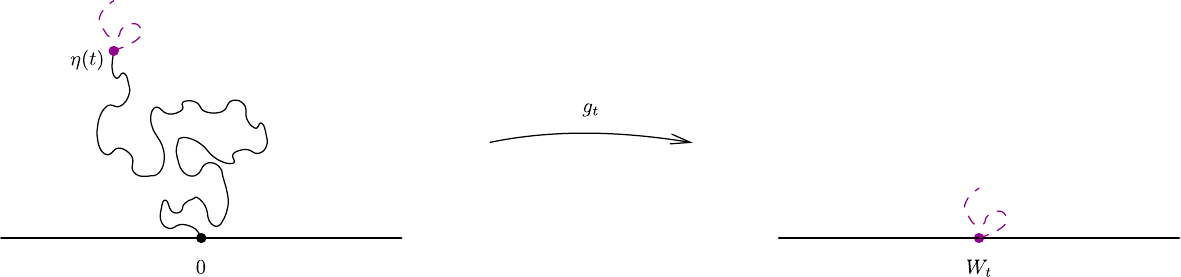}
\caption{\label{fig: SLE}
Illustration of the Loewner maps 
$g_t \colon \mathsf{H}_t \to \bH$ for the $\SLEk$ curve $\eta$,
where $\mathsf{H}_t$ is the unbounded  
component of  the curve's complement $\bH \setminus \eta[0,t]$ at time $t$.
The image of the tip $\eta(t)$ of the $\SLEk$ curve is the driving process $W_t = \sqrt{\kappa} B_t$.
}
\end{figure}

Explicitly, $\SLE_{\kappa}$ curves can be generated using random Loewner evolutions.
Thanks to its conformal invariance, 
it suffices to construct the $\SLE_{\kappa}$ curve $\eta \sim \mathbb{P}_{\bH; 0,\infty}$
in the upper half-plane $\bH := \{z \in \bC \; | \; \mathrm{Im}(z) > 0\}$ from $0$ to $\infty$.
In its construction as a growth process, the time evolution of $\eta$
is encoded in a solution of the Loewner differential equation: a collection $(g_{t})_{t\ge 0}$ of conformal maps $z \mapsto g_t(z)$.
Such maps were first considered by C.~Loewner in the 1920s while studying the Bieberbach 
conjecture~\cite{Loewner:Untersuchungen_uber_schlichte_konforme_Abbildungen_des_Einheitskreises}.
He managed to describe certain growth processes by a single ordinary differential equation, now known as 
the Loewner equation. In the upper half-plane $\bH \ni z$, it has the form 
\begin{align}\label{eq: Loewner equation}
\frac{\ud}{\ud t} g_t(z) = \frac{2}{g_t(z) - W_t},\qquad g_0(z)=z,
\end{align}
where $t \mapsto W_t$
is a real-valued continuous function, called the driving function. 
Note that, for each $z \in \bH$,
 this equation is only well-defined  up to a blow-up time, called the swallowing time of $z$,
\begin{align*}
\swaltime_z := \sup\Big\{ t > 0 \; \big| \; \inf_{ s \in[0,t]} |g_s(z)-W_s| > 0 \Big\}  .
\end{align*}
The hulls $K_{t} := \overline{\{z\in\bH \; | \; \swaltime_z \le t\}}$, for $t \geq 0$, define a growth process, called a Loewner chain.
For each $t \in [0,\swaltime_z)$, the map $z \mapsto g_t(z)$ is 
the unique conformal bijection from $\mathsf{H}_t:=\bH\setminus K_{t}$ onto $\bH$
with normalization chosen as $\smash{\underset{z\to\infty}{\lim}|g_t(z)-z|=0}$. 
Figure~\ref{fig: SLE} illustrates the Loewner chain associated to the $\SLEk$ process.

Originally, Loewner considered continuous, deterministic driving functions
(continuity of $W_t$ ensures that $K_t$ grow only locally). 
Schramm's groundbreaking idea in~\cite{Schramm:Scaling_limits_of_LERW_and_UST} 
was to take $W_t$ to be a random driving process.
In order for the process to describe scaling limits of critical interfaces, he required the 
resulting curve to satisfy the two properties in Definition~\ref{defn:SLE}.
The domain Markov property is particularly natural for discrete exploration processes, 
as we shall see in Section~\ref{subsec:mgles}. 
With conformal invariance, it guarantees that the driving process $(W_t)_{t\ge 0}$ 
has independent and stationary increments, and moreover that
$W_t=\sqrt{\kappa}B_t$, where $(B_{t})_{t\ge 0}$ is the standard Brownian motion.
Schramm proved with S.~Rohde in~\cite{Rohde-Schramm:Basic_properties_of_SLE} that
this growth process $(K_{t})_{t\ge 0}$ is almost surely generated by a continuous transient curve $(\eta_{t})_{t\ge 0}$, 
in the sense that $\mathsf{H}_t$ is the unbounded 
component of $\bH \setminus \eta[0,t]$ for each $t\ge 0$, and 
$|\eta(t)|\to\infty$ as $t\to\infty$.
The curve $\eta$ is (a parametrization of) the chordal $\SLE_{\kappa}$ in $(\bH;0,\infty)$ and $K_t$ is its hull. 
In~\cite{Rohde-Schramm:Basic_properties_of_SLE}, it was also shown that 
the $\SLE_\kappa$ curve exhibits 
phase transitions at $\kappa=4$ and $\kappa = 8$: almost surely,
\begin{itemize}
\item when $\kappa\in [0,4]$, the $\SLE_\kappa$ are simple curves, which only touch the boundary of the domain at their endpoints,

\item when $\kappa\in (4,8)$, the $\SLE_\kappa$ curves have self-touchings, are non-self-crossing, and
touch the boundary of the domain in a fractal set (with dimension $2-8/\kappa$~\cite{Alberts-Sheffield:Hausdorff_dim_of_SLE_intersected_with_real_line}),

\item when $\kappa\ge 8$, the $\SLE_\kappa$ curves are space-filling.
\end{itemize}

For more background on SLEs and related topics,
see, e.g., the books~\cite{Lawler:Conformally_invariant_processes_in_the_plane, Kemppainen:SLE_book}
and the original papers~\cite{Schramm:Scaling_limits_of_LERW_and_UST, Rohde-Schramm:Basic_properties_of_SLE}.

\subsection{\label{subsec: Ising} SLE in critical models -- the Ising model}

Next, we discuss how SLEs are related to scaling limits of critical statistical mechanics models.
We recall that many models are formulated on discretizations of the space, lattices, 
and the limit when the mesh of the lattice tends to zero is called the scaling limit.
For definiteness, we consider the 
Ising model, which describes a magnet with a paramagnetic (disordered)
and a ferromagnetic (ordered) phase --- see Figure~\ref{fig: Ising phase transition} for an illustration.
It was postulated in the seminal articles~\cite{Polyakov:Conformal_symmetry_of_critical_fluctuations, BPZ:Infinite_conformal_symmetry_of_critical_fluctuations_in_2D}
of A.~Belavin, A.~Polyakov, and A.~Zamolodchikov  that in the scaling limit, 
the critical planar Ising model is conformally invariant.
Indeed, this has been recently verified to a large 
extent~\cite{Chelkak-Smirnov:Universality_in_2D_Ising_and_conformal_invariance_of_fermionic_observables, 
Hongler-Smirnov:Energy_density_in_planar_Ising_model, CDHKS:Convergence_of_Ising_interfaces_to_SLE, CGN:Planar_Ising_magnetization_field1, CHI:Conformal_invariance_of_spin_correlations_in_planar_Ising_model, 
Izyurov:Smirnovs_observable_for_free_boundary_conditions_interfaces_and_crossing_probabilities, 
Izyurov:Critical_Ising_interfaces_in_multiply_connected_domains,
Benoist-Hongler:The_scaling_limit_of_critical_Ising_interfaces_is_CLE3,
BPW:On_the_uniqueness_of_global_multiple_SLEs}, 
and the Ising model can be claimed to be the best understood model from this point of view.
In the present article, we concentrate on a geometric description of conformal invariance,
phrased in terms of chordal interfaces~\cite{CDHKS:Convergence_of_Ising_interfaces_to_SLE, 
Izyurov:Smirnovs_observable_for_free_boundary_conditions_interfaces_and_crossing_probabilities, 
BPW:On_the_uniqueness_of_global_multiple_SLEs}, 
and their description in terms of certain CFT correlation functions, known as ``partition functions'' for the interfaces
(see Section~\ref{sec: Multiple SLE partition functions}).

\bigskip

In the Ising model, the magnet is described as a collection of atoms lying on a lattice, each with spin $\ominus$ or $\oplus$.
The configurations on a finite (planar) graph $\graph = (V, E)$ are random assignments
$\sigma = (\sigma_v)_{v \in V} \in \{\ominus, \oplus\}^{V}$ of spins at each vertex $v \in V$, with nearest-neigbor interaction
at inverse-temperature $\beta = \frac{1}{T} > 0$ sampled according to the Boltzmann measure 
\begin{align*}
  \mu_{\beta,\graph}(\sigma)
  := \frac{1}{Z_{\beta, \graph}} \exp\bigg(\beta \sum_{\edgeof{v}{w} \in E} \sigma_v \sigma_w \bigg), \qquad \qquad
  \textnormal{with partition function} \quad Z_{\beta, \graph} := \sum_{\sigma} \exp \bigg(\beta \sum_{\edgeof{v}{w} \in E}\sigma_v \sigma_w \bigg) .
\end{align*}
(We consider constant interaction strength 
at all edges  without an external magnetic field).

A more geometrical way to view the Ising model is its domain-wall representation. 
The spin configuration $\sigma$ 
results in a collection of contours, called domain walls, that separate the two different spin values from each other on the dual graph 
$\graph^* = (V^*, E^*)$ of $\graph$.
Conversely, each contour collection corresponds to two spin configurations, $\sigma$ and $-\sigma$ (related by a global spin-flip 
$\oplus \leftrightarrow \ominus$).
In other words, the Ising spin configurations $\sigma \in  \{\ominus, \oplus\}^{V}$ are in two-to-one correspondence with
subsets $\Gamma_{|\sigma|} \subset E^*$ of edges of the dual graph 
that consist of loops in the interior
and paths connecting some boundary points  
--- see also Figure~\ref{fig: Ising} with two colors representing the two spins. 
The Boltzmann weight of $\sigma$ can be written as
\begin{align*}
\exp\bigg(\beta \sum_{\edgeof{v}{w} \in E} \sigma_v \sigma_w \bigg)
=  \exp\bigg( \beta \; \# E + \sum_{e \in \Gamma_{|\sigma|}}(-2\beta) \bigg) 
=  \exp\big( \beta \; \# E -2\beta \; \# \Gamma_{|\sigma|} \big) ,
\end{align*}
where ``$\#$'' denotes the number of edges in $E$ or $\Gamma_{|\sigma|}$. Therefore, we have
\begin{align*}
  \mu_{\beta,\graph}(\sigma) = 
\frac{\exp\big( -2\beta \; \# \Gamma_{|\sigma|} \big)}{2 \tilde{Z}_{\beta, \graph}} , \qquad \qquad
  \textnormal{where} \quad \tilde{Z}_{\beta, \graph} = 
\sum_{\substack{\Gamma \subset E^* \\ \textnormal{in-even subgraphs}}} \exp\big( -2\beta \; \# \Gamma \big) ,
\end{align*}
and ``in-even'' subgraphs mean subsets $\Gamma$ of $E^*$ such that each vertex in $\Gamma$ 
which lies in the interior of $\graph^*$ has an even number of neighbors in $\Gamma$
(with no restriction for vertices on the boundary). 
The factor $2$ in the denominator is due to the symmetry $\oplus \leftrightarrow \ominus$.

In low temperatures, the factor $e^{-2\beta}$ is very small, so most likely are the configurations 
where there are only a few, if any, disagreeing nearest-neigbor spins;
see Figure~\ref{fig: Ising phase transition} (left). This is the ordered phase. 
On the other hand, in very high temperatures, $e^{-2\beta}$ is close to one
and all configurations seem equally likely. 
In a typical configuration, there are many small loops; see Figure~\ref{fig: Ising phase transition} (right).
This is the disordered phase. 
The existence of two phases indicates that a phase transition would occur as the temperature is 
varied. Indeed, R.~Peierls proved in 1936  the existence of a unique critical temperature $T_c$
where the phase transition occurs; see Figure~\ref{fig: Ising phase transition} (middle).
The value of $T_c$ was (non-rigorously) identified in the 1940s by H.~Kramers and G.~Wannier by a duality 
argument, and rigorously derived by C.~Yang in the 1950s.
We refer to, e.g.,~\cite{McCoy-Wu:2D_Ising_model, DCS:Conformal_invariance_of_lattice_models,
Mussardo:Statistical_field_theory, Friedli-Velenik:Statistical_mechanics_of_lattice_systems} for more details. 
(The critical temperature $T_c$ is also a critical fixed point of the renormalization group 
flow, see~\cite{Polyakov:Conformal_symmetry_of_critical_fluctuations,
BPZ:Infinite_conformal_symmetry_of_critical_fluctuations_in_2D, 
Cardy:Scaling_and_renormalization_in_statistical_physics}.)

So far, we had no restrictions for the spins on the boundary of $\graph$ or $\graph^*$ --- the model had free boundary conditions.
In general, one can impose various boundary conditions for the Ising model, such as free, wired ($\oplus$ or $\ominus$), 
or different on different segments of the boundary. 
For instance, in wired $\oplus$ boundary conditions, the spins at all boundary vertices are set to equal $\oplus$.
In this case, the domain-wall representation 
is particularly simple: all domain walls are collections of loops, and we have
\begin{align*}
  \mu_{\beta,\graph}^\oplus(\sigma) = 
\frac{\exp\big( -2\beta \; \# \Gamma_{|\sigma|} \big)}{\tilde{Z}_{\beta, \graph}^\oplus} , \qquad \qquad
  \textnormal{where} \quad \tilde{Z}_{\beta, \graph}^\oplus = 
\sum_{\substack{\Gamma \subset E^* \\ \textnormal{even subgraphs}}} \exp\big( -2\beta \; \# \Gamma \big) ,
\end{align*}
and even subgraphs mean subsets $\Gamma$ of $E^*$ whose every vertex 
has an even number of neighbors in $\Gamma$ (so $\Gamma$ consists of loops).

Of particular interest to us are the Dobrushin boundary conditions (domain-wall boundary conditions), where we choose
$\oplus$ along a given boundary arc $(x \; y)$ and $\ominus$ along the complementary boundary arc $(y \; x)$; see Figure~\ref{fig: Ising} (left). 
Then,  the domain walls consist of collections of loops together with
one chordal path (interface) connecting $x$ and $y$. Therefore, we have
\begin{align*}
  \mu_{\beta,\graph}^{\textrm{Dob}}(\sigma) = 
\frac{\exp\big( -2\beta \; \# \Gamma_{|\sigma|} \big)}{\tilde{Z}_{\beta, \graph}^{\textrm{Dob}}} , \qquad \qquad
  \textnormal{where} \quad \tilde{Z}_{\beta, \graph}^{\textrm{Dob}} = 
\sum_{\substack{\Gamma = \gamma \cup L , \\ L \subset E^* \textnormal{ even subgraph,} \\
\gamma \textnormal{ path $x \, \leftrightarrow \,y$}}}  \exp\big( -2\beta \; \# \Gamma \big) .
\end{align*}
By the celebrated results of D.~Chelkak, S.~Smirnov, 
et.~al.~\cite{Smirnov:Towards_conformal_invariance_of_2D_lattice_models, Smirnov:Conformal_invariance_in_random_cluster_models1,
Chelkak-Smirnov:Universality_in_2D_Ising_and_conformal_invariance_of_fermionic_observables, 
CDHKS:Convergence_of_Ising_interfaces_to_SLE}, 
at the critical temperature $T = T_c$,
the random interface $\gamma$ converges in the scaling limit 
weakly to the chordal $\SLEk$ process with $\kappa = 3$
(for suitable approximations, see~\cite{CDHKS:Convergence_of_Ising_interfaces_to_SLE} 
and Section~\ref{subsec:scaling_limimt_results_etc} for more details).
More generally, under alternating boundary conditions, 
$\oplus$ along given boundary arcs and  $\ominus$ along the complementary boundary arcs (see Figure~\ref{fig: Ising} (right)),
several macroscopic interfaces occur, and they converge in the scaling limit (at criticality) to multiple $\SLE_3$ 
processes~\cite{Izyurov:Smirnovs_observable_for_free_boundary_conditions_interfaces_and_crossing_probabilities,
BPW:On_the_uniqueness_of_global_multiple_SLEs}.
It has also been proven recently that the interior domain walls 
converge in the scaling limit (at criticality) to the so-called conformal loop ensemble
$\CLE_3$~\cite{Benoist-Hongler:The_scaling_limit_of_critical_Ising_interfaces_is_CLE3},
and critical interfaces with other variants of $\oplus$/$\ominus$/free boundary conditions
to variants of the $\SLE_3$~\cite{Hongler-Kytola:Ising_interfaces_and_free_boundary_conditions,
Izyurov:Smirnovs_observable_for_free_boundary_conditions_interfaces_and_crossing_probabilities}.

\begin{figure}
\begin{minipage}[b]{0.4\textwidth}
\centering
\begin{overpic}[width=.8\textwidth]{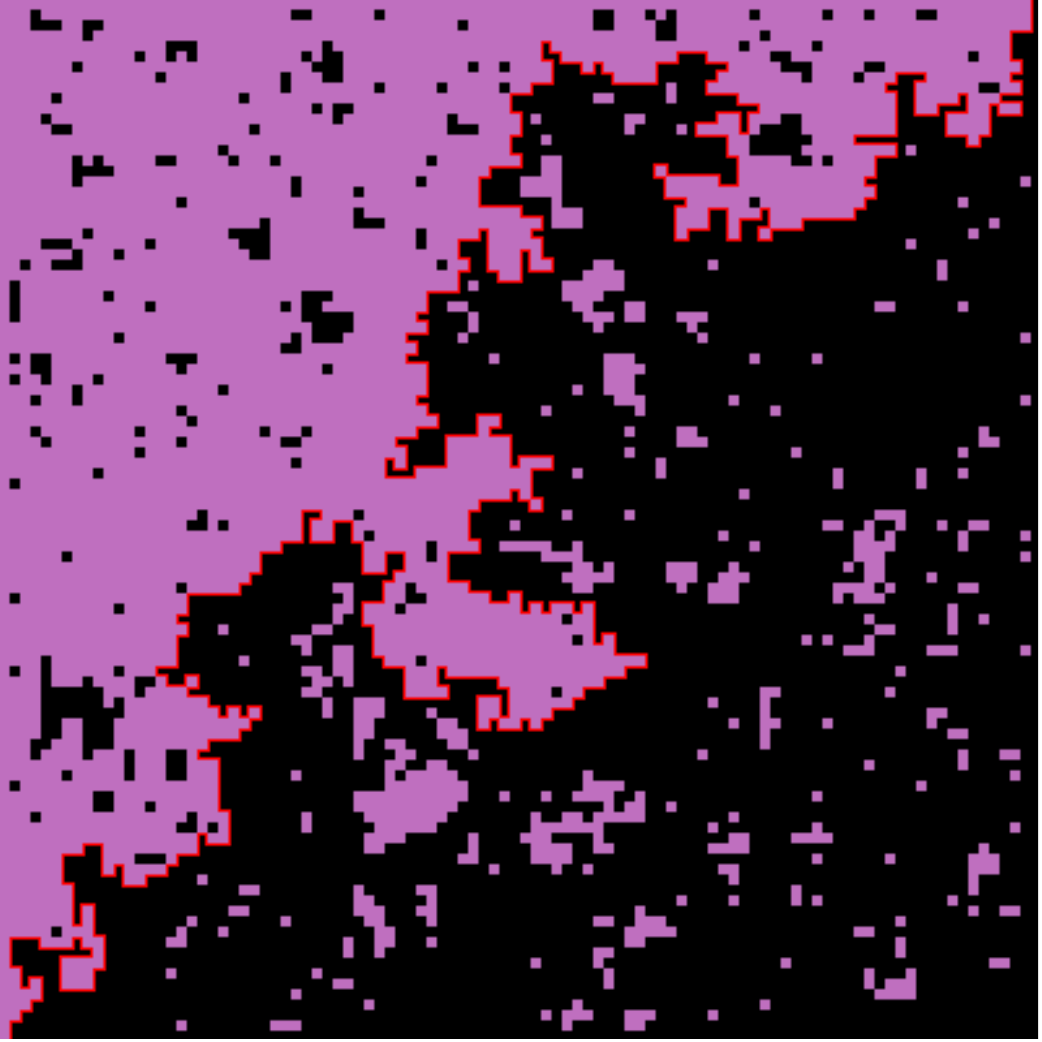}
 \put (-7,-2) {\Large$x$}
 \put (103,99) {\Large$y$}
\end{overpic}
\end{minipage}
\qquad \qquad
\begin{minipage}[b]{0.4\textwidth}
\centering
\begin{overpic}[width=.8\textwidth]{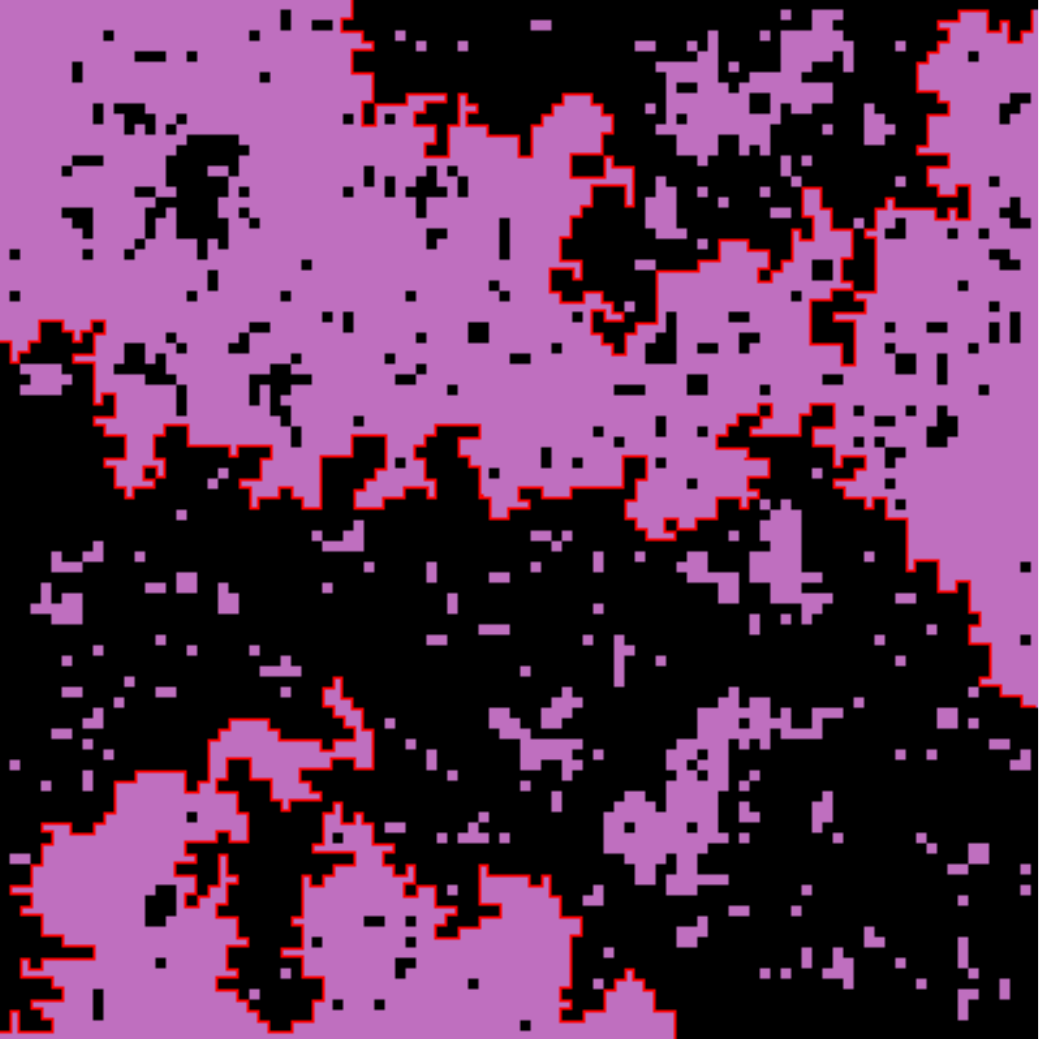}
 \put (-10,-2) {\Large$x_1$}
 \put (-10,65) {\Large$x_6$}
 \put (103,30) {\Large$x_3$}
 \put (63,-7) {\Large$x_2$}
 \put (103,99) {\Large$x_4$}
 \put (30,104) {\Large$x_5$}
\end{overpic}
\end{minipage}
\caption{\label{fig: Ising}
Critical Ising model configurations on a square 
lattice with Dobrushin (left) and alternating (right) boundary conditions. 
The points $x$ and $y$ (resp. $x_1, \ldots, x_6$) should be understood, e.g., 
as midpoints of edges connecting two boundary vertices where the boundary conditions change.
In the figure (and in Figure~\ref{fig: Ising phase transition}), the two colors represent the two spins $\oplus$ and $\ominus$.
}
\end{figure}
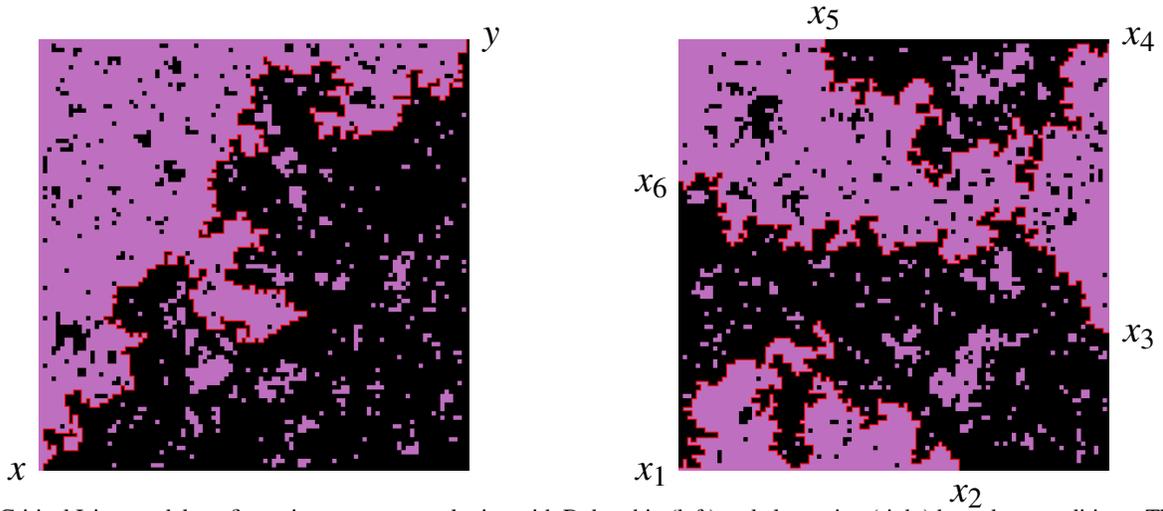

\subsection{\label{subsec:CFT}Conformal field theory (CFT)}

Next, we briefly describe some aspects of 2D conformal field theory (CFT). 
There are many textbooks on CFT from different viewpoints, 
see, e.g,~\cite{DMS:CFT, Schottenloher:Mathematical_introduction_to_CFT, Mussardo:Statistical_field_theory}. 
Here, we aim to only give some rough ideas,
in order to motivate the connection of SLEs with CFT and 
to illustrate how it could be understood. 
We emphasize that in CFT, the \emph{fields} themselves might not be analytically well-defined  objects,
but nevertheless, their \emph{correlation functions} are well-defined  functions of several complex variables.
Moreover, some correlation functions have been rigorously related to lattice model 
correlations (see, e.g.,~\cite{Hongler-Smirnov:Energy_density_in_planar_Ising_model,
CHI:Conformal_invariance_of_spin_correlations_in_planar_Ising_model, CHI:inprep} for the Ising model)
and SLE curves (see, e.g.,~\cite{Kytola-Peltola:Pure_partition_functions_of_multiple_SLEs,
Peltola-Wu:Crossing_probabilities_of_multiple_Ising_interfaces,
KKP:Conformal_blocks_pure_partition_functions_and_KW_binary_relation}, and Section~\ref{sec: Multiple SLE partition functions}).

\bigskip

Scaling limits of critical lattice models are expected to enjoy conformal invariance.
The conformal maps on the extended complex plane $\hat{\bC} := \bC \cup \{\infty\}$ form a group of finite dimension, 
the M\"obius group $\mathrm{PSL}(2,\bC)$, 
acting as M\"obius transformations
$\Mob(z) = \frac{az+b}{cz+d}$ with $a,b,c,d \in \bC$ and $ad-bc = 1$.
In particular, global conformal invariance only results in finitely many (three)
constraints for the physical system. However, A.~Belavin, A.~Polyakov, and A.~Zamolodchikov
observed in the 1980s that, in two dimensions, imposing \emph{local} conformal invariance 
yields infinitely many independent symmetries
\cite{BPZ:Infinite_conformal_symmetry_in_2D_QFT,
BPZ:Infinite_conformal_symmetry_of_critical_fluctuations_in_2D}.
On $\hat{\bC}$, the local conformal transformations 
are just the locally invertible holomorphic and anti-holomorphic maps ---
see, e.g.,~\cite[Chapters~\red{1},\red{2},\red{5}]{Schottenloher:Mathematical_introduction_to_CFT} for details.
In CFT \`a la Belavin, Polyakov \& Zamolodchikov, one 
regards the local conformal invariance as  invariance
under infinitesimal transformations (or vector fields which generate the local 
conformal mappings): for instance, the infinitesimal holomorphic transformations are written as Laurent 
series, $z \mapsto z + \sum_{n \in \bZ} a_n z^n$,
which can be seen to be generated by the vector fields $\ell_n := - z^{n+1} \pder{z}$, for $n \in \bZ$,
constituting a Lie algebra isomorphic to the Witt algebra
with commutation relations $[\ell_n,\ell_m] = (n-m)\ell_{n+m}$. 
(In this section, we will not take into account the anti-holomorphic sector, 
see~\cite{DMS:CFT, Schottenloher:Mathematical_introduction_to_CFT, Mussardo:Statistical_field_theory}.)

In quantized systems, the symmetry groups and algebras often are central 
extensions of their classical counterparts.
In particular, in conformally invariant quantum field theory (i.e., CFT),
the conformal symmetry algebra is the unique central extension of the Witt algebra by the one-dimensional abelian Lie algebra $\bC$, 
namely the Virasoro algebra $\Vir$.
The central part represents a ``conformal anomaly'', giving rise to a projective representation of the Witt algebra
--- see, e.g.,~\cite[Chapters~\red{3},\red{4},\red{5}]{Schottenloher:Mathematical_introduction_to_CFT} 
for the algebraic side 
and~\cite{Cardy:Scaling_and_renormalization_in_statistical_physics, DMS:CFT} 
for a geometric interpretation of the conformal anomaly.
Precisely, $\Vir$ is the infinite-dimensional Lie algebra generated by $\mathrm{L}_n$, for $n \in \bZ$, 
together with a central element $\mathrm{C}$, with commutation relations 
\begin{align*}
[\mathrm{L}_n,\mathrm{C}] = 0 \qquad \quad \textnormal{and} \quad \qquad 
[\mathrm{L}_n,\mathrm{L}_m] = (n-m) \mathrm{L}_{n+m} 
+ \frac{1}{12} n(n^2-1) \delta_{n,-m} \mathrm{C} , \qquad \textnormal{for } n,m \in \bZ.
\end{align*}

Algebraically, the basic objects in a CFT, the conformal fields, can be regarded as elements in
representations of the symmetry algebra $\Vir$, where the central element acts as 
a constant multiple of the identity, $\mathrm{C} = c \; \id$. The number $c \in \bC$ is 
called the central charge of the CFT. For relation to SLEs and statistical physics,
real central charges $c \leq 1$ are relevant (using the parameterization 
$c(\kappa) = \frac{(3\kappa-8)(6-\kappa)}{2\kappa}$, this corresponds to $\kappa > 0$).
We briefly review some representation theory of $\Vir$ in Appendix~\ref{app:Vir}.

There are many attempts to understand conformal fields analytically
--- e.g., as operator-valued distributions~\cite{Schottenloher:Mathematical_introduction_to_CFT},
vertex operators~\cite{Huang:2D_Conformal_geometry_and_VOAs},
or formal objects in a bosonic Fock space~\cite{Kang-Makarov:Gaussian_free_field_and_conformal_field_theory}.
In the present article, we focus on correlation functions.
They are analytic (multi-valued) functions $F \colon \extendedChamber_{n} \to \bC$ 
(also called $n$-point functions)
defined on the configuration space 
\begin{align} \label{eq: chamberComplex}
\extendedChamber_n :=\; & 
 \{ (z_{1},\ldots,z_n) \in \bC^{n} \; | \; z_i \neq z_j \textnormal{ if } i \neq j \} .
\end{align}
Physicists speak of correlation functions as ``vacuum expectation values'' of  
fields $\Phi_{\index_i}(z_i)$ and denote them by
\begin{align} \label{eq: corr fction notation} 
F_{\index_1, \ldots, \index_n}(z_1,\ldots,z_n) = \big\langle \Phi_{\index_1}(z_1) \cdots \Phi_{\index_n}(z_n) \big\rangle .
\end{align}
Because of the conformal symmetry, the correlation functions are assumed to be covariant under (global)
conformal transformations. In a CFT on the full $\hat{\bC}$, this means that 
under all M\"obius transformations $\Mob \in \mathrm{PSL}(2,\bC)$, we have
\begin{align} \label{eq: correlation function Mobius covariance}
F_{\index_1, \ldots, \index_n}(z_1,\ldots,z_n) = 
\prod_{i=1}^n | \Mob'(z_i) |^{\Delta_{\index_i}}  \times  F_{\index_1, \ldots, \index_n} (\Mob(z_1),\ldots,\Mob(z_n)),
\end{align}
with some conformal weights $\Delta_{\index_i} \in \bR$ associated to the fields $\Phi_{\index_i}$.
Of specific interest to us is CFT in the domain $\bH$ with boundary $\partial \bH = \bR$, 
where the global conformal transformations are also M\"obius maps, $\Mob \in \mathrm{PSL}(2,\bR)$.
For example, the multiple $\SLEk$ partition functions discussed in Section~\ref{sec: Multiple SLE partition functions}
satisfy covariance property~\eqref{eq: correlation function Mobius covariance},
where $\Delta_{\index_i} = h_{1,2} = \frac{6-\kappa}{2\kappa}$, for all $1 \leq i \leq  n$;
see~\eqref{eq: multiple SLE Mobius covariance}.

\bigskip

In this article, we are concerned with so-called primary fields. 
They are fields whose correlation functions also have
a covariance property under local conformal transformations, 
in an infinitesimal sense, see~\cite[Chapter~\red{9}]{Schottenloher:Mathematical_introduction_to_CFT}. 
Other fields in the CFT are called descendant fields,  obtained from the primary fields by action of the Virasoro algebra.
A primary field $\Phi(z)$ of conformal weight $\Delta$
generates a highest-weight module $\primaryRep_{c,\Delta}$ of the Virasoro algebra of 
weight $\Delta$ and central charge $c$ (see Appendix~\ref{app:Vir}). 
In physics, it is called the ``conformal family'' of $\Phi(z)$, 
consisting of linear combinations of the ``descendant fields'' of $\Phi(z)$.
In general, the descendants have the form $\mathrm{L}_{-n_1} \cdots \mathrm{L}_{-n_k} \Phi(z)$,
where $n_1 \geq \cdots \geq n_k > 0$ and $k \geq 1$. Their correlation functions 
are formally determined from the correlation functions of $\Phi(z)$ using linear differential 
operators which arise from the generators of the Virasoro algebra 
(see, e.g.,~\cite[Chapter~\red{10}]{Mussardo:Statistical_field_theory}): 
for any primary fields $\{ \Phi_{\index_i}(z_i) \; | \; 1 \leq i \leq n\}$, 
we have 
\begin{align}
& \; \big\langle \Phi_{\index_1}(z_1) \cdots \Phi_{\index_n}(z_n) 
\; \mathrm{L}_{-k} \Phi(z) \big\rangle \; 
= \; \mathcal{L}_{-k}^{(z)} \;
\big\langle \Phi_{\index_1}(z_1) \cdots \Phi_{\index_n}(z_n) \Phi(z) \big\rangle,
\qquad \textnormal{where} \nonumber \\
& \; \mathcal{L}_{-k}^{(z)} \; := \; \sum_{i=1}^n 
\left( \frac{(k-1) \Delta_{\index_i}}{(z_i - z)^k}
\; - \; \frac{1}{(z_i - z)^{k-1}} \pder{z_i} \right) , \qquad \textnormal{for } k \in \bZpos . 
\label{eq: Virasoro partial differential operator}
\end{align}

Now, consider the $\Vir$-module $\primaryRep_{c,\Delta}$ generated by the primary 
field $\Phi(z)$. It is necessarily a quotient of a Verma module,
$\primaryRep_{c,\Delta} \cong \mathrm{M}_{c,\Delta}/\mathrm{J}$, by some submodule $\mathrm{J}$ (see Appendix~\ref{app:Vir}). 
Suppose that the conformal weight $\Delta = h_{r,s}$ belongs to the special class~\eqref{eq: Kac weights GEN} 
discussed in Appendix~\ref{app:Vir}, and denote $\Phi = \Phi_{r,s}$ accordingly.
Then, by Theorem~\ref{thm:FF}, the Verma module $\mathrm{M}_{c,h_{r,s}}$ contains a singular vector 
$v = P(\mathrm{L}_{-1},\mathrm{L}_{-2},\ldots) v_{c,h_{r,s}}$
at level $rs$, where $P$ is a polynomial in the generators of the Virasoro algebra.
If this vector is contained in $\mathrm{J}$
(which is the case, e.g., when $\primaryRep_{c,h_{r,s}}$ is irreducible), then the descendant field 
$P(\mathrm{L}_{-1},\mathrm{L}_{-2},\ldots) \Phi_{r,s}(z)$ in $\primaryRep_{c,h_{r,s}}$
corresponding to the singular vector $v$ is zero, a null field.  
In this case, we say that $\Phi_{r,s}(z)$ has a degeneracy at level $rs$.
In particular, correlation functions containing the field $\Phi_{r,s}(z)$ then  satisfy partial differential equations
(known as null-field equations)
given by the polynomial $P(\mathcal{L}_{-1}^{(z)},\mathcal{L}_{-2}^{(z)},\ldots)$ and the differential operators~\eqref{eq: Virasoro partial differential operator},
\begin{align*}
0 \; = \; \big\langle \Phi_{\index_1}(z_1)\cdots \Phi_{\index_n}(z_n) P(\mathrm{L}_{-1},\mathrm{L}_{-2},\ldots) \Phi_{r,s}(z) \big\rangle 
\; = \; P(\mathcal{L}_{-1}^{(z)},\mathcal{L}_{-2}^{(z)},\ldots) \;
\big\langle \Phi_{\index_1}(z_1)\cdots \Phi_{\index_n}(z_n) \Phi_{r,s}(z) \big\rangle .
\end{align*}
In other words, for the correlation function~\eqref{eq: corr fction notation} with $\Phi_ \index(z) =  \Phi_{r,s}(z)$,
we have the following (perfectly well-defined) PDE:
\begin{align} \label{eq: PDE for correlation functions}
F_{\index_1, \ldots, \index_n, \index}  \colon \extendedChamber_{n+1} \to \bC ,
\qquad\qquad
P(\mathcal{L}_{-1}^{(z)},\mathcal{L}_{-2}^{(z)},\ldots) \; F_{\index_1, \ldots, \index_n, \index} (z_1,\ldots,z_n, z) = 0 .
\end{align}

An example of such a PDE is the second order equation~\eqref{eq: singular equation level two} generated by the singular 
vector~\eqref{eq: singular vector level two} at level two, associated to the primary field 
$\Phi_{1,2}(z)$ of conformal weight $h_{1,2}$ (or $\Phi_{2,1}(z)$, $h_{2,1}$, see Appendix~\ref{app:Vir}). 
Combining with translation invariance, this PDE gives rise to a PDE in the system of equations~\eqref{eq: multiple SLE PDEs} 
for the multiple $\SLE_\kappa$ partition functions,  
discussed in Section~\ref{sec: Multiple SLE partition functions}.

\begin{remark}
Primary fields with degeneracy at level two can be associated, 
e.g., with the spin 
and the energy density in the scaling limit of the critical Ising 
model~\cite{BPZ:Infinite_conformal_symmetry_in_2D_QFT, 
BPZ:Infinite_conformal_symmetry_of_critical_fluctuations_in_2D}
(with $c=1/2$, $h_{2,1} = 1/16$, $h_{1,2} = 1/2$, and $\kappa = 3$). 
Furthermore, it was argued in~\cite{Cardy:Effect_of_boundary_conditions_on_the_operator_content_of_two-dimensional_conformally_invariant_theories,
Cardy:Boundary_conditions_fusion_rules_and_Verlinde_formula, 
Burkhardt-Xue:Conformal_invariance_and_critical_systems_with_mixed_boundary_conditions,
Burkhardt-Guim:Conformal_theory_of_2D_Ising_model_with_homogeneous_boundary_conditions_and_with_disordered_boundary_fields}
that the field $\Phi_{1,2}(x)$ implements a boundary condition change from $\oplus$ to $\ominus$ at the boundary point $x$, 
see also Figure~\ref{fig: Ising}. 
Thus, $\Phi_{1,2}$ could be thought of as an ``interface generating field'' for the spin Ising model.

One could also modify the boundary conditions of a critical lattice model
by inserting other types  of boundary condition changes at given boundary points.
For instance, fields of type $\Phi_{1,s}(x)$ or $\Phi_{r,1}(x)$ with higher level degeneracies
could perhaps generate arm events on the 
boundary~\cite{Duplantier-Saleur:Exact_surface_and_wedge_exponents_for_polymers_in_two_dimensions,
Bauer-Saleur:On_some_relations_between_local_height_probabilities_and_conformal_invariance}. 
Correlation functions of these fields satisfy PDEs of higher order, that we will discuss in Section~\ref{sec:OPE}.
One can construct solutions to these PDEs from limits of solutions of the second order PDEs~\eqref{eq: multiple SLE PDEs}. 
\end{remark}

\subsection{\label{subsec:mgles}Martingale observables for interfaces}

In this section, we describe heuristically how certain martingales 
associated to critical interfaces can be related to correlation functions of the CFT fields $\Phi_{1,2}$
appearing in Section~\ref{subsec:CFT}. Our presentation is not intended to be rigorous, but we rather wish to
give the intuitive idea of why such a connection might exist. 
Even though the nature of the objects ``$\Phi_{1,2}$'' is unclear,
their correlation functions~\eqref{eq: corr fction notation} can be 
well understood and studied, e.g., as 
multiple SLE partition functions (discussed in Section~\ref{sec: Multiple SLE partition functions}).

\bigskip

Consider the 
Ising model with some boundary conditions (b.c.).
The expected value of a random variable $\mathcal{O}$ (``observable''), such as a product 
$\sigma_{v_1} \cdots \sigma_{v_n}$ of spins at given vertices $v_1, \ldots, v_n \in V$,
or the energy $\varepsilon_{\edgeof{v}{w}} = \sigma_v \sigma_w$
at an edge $\edgeof{v}{w} \in E$, is 
\begin{align*}
\mathbb{E}_{\beta, \graph}^{\textrm{b.c.}}[\mathcal{O}] := \frac{1}{Z_{\beta, \graph}^{\textrm{b.c.}}} 
\sum_{\sigma} 
\mathcal{O}(\sigma) \; \exp \bigg(\beta \sum_{\edgeof{v}{w} \in E}\sigma_v \sigma_w \bigg) .
\end{align*}
Conjecturally, the expectation of the discrete observable $\mathcal{O}$ 
should converge in the scaling limit to a correlation function of some ``continuum observable'' (or quantum field) $\Phi$.
In particular, for the planar Ising model at its critical temperature $T_c = \frac{1}{\beta_c}$,
the object ``$\Phi$'' should be a conformally invariant field in a CFT.
Thus, if $\graph = \graph^\delta \subset \delta \bZ^2$ approximate some planar (simply connected) domain  
$\Omega\subset \bC$ as $\delta \searrow 0$ (e.g., in the Carath\'eodory topology), we expect the following convergence to take place
(of course, a lot of work has to be done in order to make such a statement mathematically precise ---
for the critical Ising model, this can actually be established to a large extent, see~\cite{Chelkak-Smirnov:Universality_in_2D_Ising_and_conformal_invariance_of_fermionic_observables, 
Hongler-Smirnov:Energy_density_in_planar_Ising_model,
CGN:Planar_Ising_magnetization_field1,
CHI:Conformal_invariance_of_spin_correlations_in_planar_Ising_model, 
CHI:inprep}):
\begin{align*}
\delta^{-D} \,
\mathbb{E}_{\beta_c, \graph^\delta}^{\textrm{b.c.}}[\mathcal{O}^\delta] 
\qquad
\overset{\delta \to 0}{\longrightarrow} 
\qquad
\frac{\big\langle \Phi \big\rangle_{\Omega}^{\textrm{b.c.}}}{\big\langle \one \big\rangle_{\Omega}^{\textrm{b.c.}}} 
= \frac{\big\langle \Phi \; \Psi^{\textrm{b.c.}} \big\rangle_{\Omega}}{\big\langle \Psi^{\textrm{b.c.}} \big\rangle_{\Omega}}  ,
\end{align*}
where $D \in \bR$ is the scaling dimension of $\Phi$, 
and $\Psi^{\textrm{b.c.}}$ is a ``field''  implementing the 
boundary conditions on $\partial \Omega$. 
(In general, the scaling dimension $D = \Delta + \tilde{\Delta}$ 
is the sum of the conformal ($\Delta$) and anti-conformal ($\tilde{\Delta}$) weights of 
$\Phi$~\cite{DMS:CFT, Schottenloher:Mathematical_introduction_to_CFT, Mussardo:Statistical_field_theory}.)

For instance, 
according to  predictions in the physics 
literature~\cite{Cardy:Boundary_conditions_fusion_rules_and_Verlinde_formula, 
Cardy:Effect_of_boundary_conditions_on_the_operator_content_of_two-dimensional_conformally_invariant_theories,
Burkhardt-Xue:Conformal_invariance_and_critical_systems_with_mixed_boundary_conditions,
Burkhardt-Guim:Conformal_theory_of_2D_Ising_model_with_homogeneous_boundary_conditions_and_with_disordered_boundary_fields},
when imposing Dobrushin boundary conditions
$\oplus$ on the boundary arc $(x^\delta \; y^\delta)$ and $\ominus$ on the complementary arc $(y^\delta \; x^\delta)$,
as in Figure~\ref{fig: Ising} (left), 
we expect that
\begin{align*}
\delta^{-D} \,
\mathbb{E}_{\beta_c, \graph^\delta}^{\textrm{Dob}}[\mathcal{O}^\delta] 
\qquad
\overset{\delta \to 0}{\longrightarrow} 
\qquad
\frac{\big\langle \Phi \big\rangle_{\Omega}^{\textrm{Dob}}}{\big\langle \one \big\rangle_{\Omega}^{\textrm{Dob}}} 
= \frac{\big\langle \Phi \; \Psi^{\textrm{Dob}} \big\rangle_{\Omega}}{\big\langle \Psi^{\textrm{Dob}} \big\rangle_{\Omega}}  ,
\end{align*}
where the boundary condition changing operator has the form
$\Psi^{\textrm{Dob}}(x,y) = \Phi_{1,2}(x) \Phi_{1,2}(y)$,
with $x = \smash{\underset{\delta\to0}{\lim} \, x^\delta}$ and $y = \smash{\underset{\delta\to0}{\lim} \, y^\delta}$. 
In general, for alternating boundary conditions with $2N$ marked boundary points 
$x_1^{\delta}, \ldots x_{2N}^{\delta}$ converging to $x_1, \ldots, x_{2N}$,
\begin{align} \label{eq::alternating}
\oplus \textnormal{ on }(x_{2j-1}^{\delta} \, x_{2j}^{\delta}), \qquad \textnormal{for } 1 \leq j \leq N ,
\qquad\qquad \textnormal{and} \qquad \qquad
\ominus \textnormal{ on }(x_{2j}^{\delta} \,  x_{2j+1}^{\delta}),\qquad \textnormal{for }  0 \leq j \leq N ,
\end{align}
as in Figure~\ref{fig: Ising} (right), the boundary condition changing operator should 
have the form~\cite{Burkhardt-Guim:Conformal_theory_of_2D_Ising_model_with_homogeneous_boundary_conditions_and_with_disordered_boundary_fields}
\begin{align} \label{eq: Ising pf general domain}
\Psi^{\textrm{alt}}(x_1,x_2,\ldots,x_{2N}) = \Phi_{1,2}(x_1) \Phi_{1,2}(x_2) \cdots  \Phi_{1,2}(x_{2N}) ,
\qquad 
\big\langle \Psi^{\textrm{alt}}(x_1,\ldots,x_{2N}) \big\rangle_{\Omega}
= \mathrm{pf} \Big( 
\big\langle \Phi_{1,2}(x_i) \Phi_{1,2}(x_j) \big\rangle_{\Omega} \Big)_{i,j=1}^{2N} ,
\end{align}
where $\mathrm{pf}(\cdot)$ is the Pfaffian of the $(2N \times 2N)$-matrix of two-point functions with zeros on the diagonal.
We remark that the Pfaffian structure on the right side is specific for the spin-Ising model 
(with $\kappa=3$, $h_{1,2} = 1/2$, and $c=1/2$), 
whereas the normalization factors $\big\langle \Psi^{\textrm{Dob}} \big\rangle_{\Omega} = \big\langle \Phi_{1,2}(x) \Phi_{1,2}(y)\big\rangle_{\Omega}$
and
$\big\langle \Psi^{\textrm{alt}}\big\rangle_{\Omega} = \big\langle \Phi_{1,2}(x_1) \cdots  \Phi_{1,2}(x_{2N})\big\rangle_{\Omega}$
could also be defined for other models for which alternating boundary conditions can be made sense of
(see also Section~\ref{sec: Multiple SLE partition functions} for general classification and relation to the $\SLE_\kappa$).

\bigskip

Consider now the planar Ising model on $\graph$ 
with Dobrushin boundary conditions as in Figure~\ref{fig: Ising} (left). 
We define a (discrete time) exploration process $(\gamma(t))_{t \geq 0}$ 
by following the chordal interface on the dual graph starting from $x = \gamma(0)$ 
in such a way that immediately to the left (resp.~right) of $\gamma$ we have spins $\ominus$ (resp.~$\oplus$),
and in case of ambiguity, we always turn left.
(In what follows, we will abuse notation for the time $t \geq 0$, discrete for the lattice exploration process, continuous for the $\SLE$ process).

The exploration process naturally has the following  domain Markov property: if we have observed $\gamma{[0,t]}$ up to a time $t$
(i.e., after a certain number of steps), then the remaining part of $\gamma$ is distributed as the exploration process
for the Ising model on the graph  $\graph \setminus \gamma{[0,t]}$ with Dobrushin boundary conditions
$\oplus$ on the boundary arc $(\gamma(t) \;  y)$ and $\ominus$ on the complementary arc $(y \; \gamma(t))$. 
We recall from Definition~\ref{defn:SLE} that such a Markovian property is manifest also for the growth of the chordal SLE process.

Using the exploration process, we can define its natural filtration $(\mathcal{F}_t)_{t \geq 0}$ and consider martingale observables.
The conditional expectation of an observable $\mathcal{O}$ given $\mathcal{F}_t$
is trivially a local martingale, and thanks to the domain Markov property, 
we can rewrite the conditional expectation as the usual expectation on the graph $\graph\setminus \gamma{[0,t]}$:
\begin{align*}
\mathbb{E}_{\beta, \graph}^{\textrm{Dob}}[\mathcal{O} \; | \; \mathcal{F}_t]
= \mathbb{E}_{\beta, \graph\setminus \gamma{[0,t]}}^{\textrm{Dob}}[\mathcal{O}] .
\end{align*}
Again, we expect that at criticality, this quantity converges in the scaling limit to a ratio of CFT correlation functions:
\begin{align} \label{eq:mgle obs}
\delta^{-D} \, \mathbb{E}_{\beta_c, \graph^\delta}^{\textrm{Dob}}[\mathcal{O}^\delta \; | \; \mathcal{F}_t]
= \delta^{-D} \,\mathbb{E}_{\beta_c, \graph^\delta \setminus \gamma^\delta{[0,t]}}^{\textrm{Dob}}[\mathcal{O}^\delta]
\qquad
\overset{\delta \to 0}{\longrightarrow} 
\qquad
\frac{\big\langle \Phi \; \Psi^{\textrm{Dob}} \big\rangle_{\Omega_t}}{\big\langle \Psi^{\textrm{Dob}} \big\rangle_{\Omega_t}} ,
\end{align}
where the domain $\Omega_t \subset \bC$ is approximated by $\graph^\delta \setminus \gamma^\delta{[0,t]}$ as $\delta \searrow 0$.
Of course, the domain $\Omega_t = \Omega \setminus \gamma{[0,t]}$ should be
given by the complement of the scaling limit curve $\gamma$
of the discrete exploration interface $\gamma^\delta$, 
namely, the chordal $\SLE_3$ curve~\cite{CDHKS:Convergence_of_Ising_interfaces_to_SLE}.
In particular, the limiting expression on the right side of~\eqref{eq:mgle obs} should be a local martingale for 
the chordal $\SLE_3$ curve $\gamma$.

To see what the martingale property gives us, suppose that our 
observable depends on some variables $z_1, \ldots, z_n \in \cl\Omega$ and
$\Phi$ has the form of a product of some CFT primary fields,
$\Phi(z_1, \ldots, z_n) = \Phi_{\index_1}(z_1) \cdots \Phi_{\index_n}(z_n)$,
with conformal weights $\Delta_1, \ldots, \Delta_n \in \bR$.
For example, $\Phi$ could be a product of spins (with $\Phi_{\index_i}(z_i) = \sigma_{z_i}$ and $\Delta_{\index_i} = 1/16$, for all $i$).
Write also the boundary condition changing operator as $\Psi^{\textrm{Dob}}(x,y) = \Phi^{\ominus \oplus}(x) \Phi^{\oplus \ominus}(y)$, 
a product of some primary fields of some weights $\Delta^{\ominus \oplus}$ and $\Delta^{\oplus \ominus}$,
where $x$ and $y$ are (the scaling limits of) the boundary points where $\oplus$ changes to $\ominus$.
Then, using  conformal covariance postulate~\eqref{eq: correlation function Mobius covariance} for CFT correlation functions,
we can write the local martingale~\eqref{eq:mgle obs} in the form
\begin{align*}
M_{\Omega_t} (\gamma(t),y; z_1, \ldots,z_n) := \; &
\frac{\big\langle \Phi(z_1, \ldots, z_n) \; \Phi^{\ominus \oplus}(\gamma(t)) \Phi^{\oplus \ominus}(y)
\big\rangle_{\Omega_t}}{\big\langle \Phi^{\ominus \oplus}(\gamma(t)) \Phi^{\oplus \ominus}(y)
\big\rangle_{\Omega_t}} 
=  \prod_{i=1}^n | \Mob'(z_i) |^{\Delta_{\index_i}}  \times
M_{\Omega} (\Mob(\gamma(t)),\Mob(y); \Mob(z_1), \ldots, \Mob(z_n) )
,
\end{align*}
where $\Mob \colon \Omega_t \to \Omega$ is a conformal map (and we assume that it extends to the boundary of $\Omega_t$).
In particular, taking $\Omega = \bH$ to be the upper half-plane, $x = 0$, $y = \infty$, 
and $\Mob = g_t \colon \mathsf{H}_t \to \bH$ 
the solution to the Loewner equation~\eqref{eq: Loewner equation}
for the $\SLE_\kappa$ curve $\gamma$ 
with driving function $W_t = \sqrt{\kappa} B_t$
(for the Ising model, $\kappa = 3$, but let us keep it symbolic here), 
and dropping $g_t(y) = y = \infty$, we have
\begin{align} \label{eq:mgle obs ex}
M_{\mathsf{H}_t} (\gamma(t); z_1, \ldots,z_n) =
\prod_{i=1}^n g_t'(z_i)^{\Delta_{\index_i}}  \times
M_{\bH}(W_t; g_t(z_1), \ldots, g_t(z_n) ),
\end{align}
where $W_t = g_t(\gamma(t))$. 
Now, it is straightforward to formally calculate the It\^o differential of 
the local martingale~\eqref{eq:mgle obs ex} using
It\^o's formula, the observation $g_t'(z) > 0$, and the relations 
\begin{align*}
 \ud g_t(z) = \frac{2}{g_t(z) - W_t} \; \ud t 
\qquad \qquad \textnormal{and} \qquad \qquad
\ud g_t'(z) = -  \frac{2g_t'(z)}{(g_t(z) - W_t)^2} \; \ud t ,
\end{align*}
which follow from the Loewner equation~\eqref{eq: Loewner equation}.
(We cautiously note that it is not clear that $M$ is smooth enough to apply It\^o's formula.)
By the martingale property, the drift term in the result should equal zero, which gives the following second order PDE:
\begin{align} \label{eqn:CFTmglePDE}
\left[ \frac{\kappa}{2}\pdder{x} +
\sum_{i=1}^n \left( \frac{2}{z_i-x} \pder{z_i} -  \frac{2\Delta_{\index_i}}{(z_i-x)^2}\right) \right] M_{\bH} (x; z_1, \ldots,z_n) = 0 .
\end{align} 
We invite the reader to compare the PDE~\eqref{eqn:CFTmglePDE} to the PDEs in~\eqref{eq: multiple SLE PDEs} 
in Section~\ref{sec: Multiple SLE partition functions}.
In CFT language, the boundary condition changing operator  
$\Phi^{\ominus \oplus}(x) = \Phi_{1,2}(x)$
has a degeneracy at level two, with
conformal weight of special type: $\Delta^{\ominus \oplus} = h_{1,2}$ 
--- recall Section~\ref{subsec:CFT} and Appendix~\ref{app:Vir} for the degeneracies and PDEs in CFT.
Similarly, we have $\Phi^{\oplus \ominus}(y) = \Phi_{1,2}(y)$ and $\Delta^{\oplus \ominus} = h_{1,2}$.

\begin{remark}
In~\cite{Bauer-Bernard:Conformal_field_theories_of_SLEs, 
Bauer-Bernard:Conformal_transformations_and_SLE_partition_function_martingale}, 
M.~Bauer and D.~Bernard considered the effect of the local conformal symmetry realized by the Virasoro algebra
to the Loewner chains that generate the SLE curves. 
They observed, in particular, that there must be an explicit relation with the $\SLE_\kappa$ parameter $\kappa > 0$
and the conformal weight $h_{1,2} = \frac{6-\kappa}{2\kappa}$, and the central charge 
of the theory should be parametrized by $c = \frac{(3\kappa-8)(6-\kappa)}{2\kappa}$.
The above martingale ideas also appear in~\cite{Bauer-Bernard:Conformal_field_theories_of_SLEs, 
Bauer-Bernard:Conformal_transformations_and_SLE_partition_function_martingale}.
\end{remark}

Explicitly, the normalization factor (``partition function'') 
with Dobrushin boundary conditions reads
\begin{align} \label{eq: 2ptf}
\big\langle \one \big\rangle_{\Omega}^{\textrm{Dob}} 
= \big\langle \Phi_{1,2}(x) \Phi_{1,2}(y) \big\rangle_{\Omega} = H_{\Omega}(x,y)^{h_{1,2}} ,
\end{align}
where $H_{\Omega}$ is the boundary Poisson kernel in 
$\Omega$~\cite{BBK:Multiple_SLEs_and_statistical_mechanics_martingales, Dubedat:Commutation_relations_for_SLE,
Kozdron-Lawler:Configurational_measure_on_mutually_avoiding_SLEs, Dubedat:SLE_and_free_field}. 
This function is well-defined for all $\kappa > 0$, 
although it might not always have an interpretation in a discrete model. 
$H_{\Omega}(x,y)^{h_{1,2}}$ can also be understood as the partition function (or ``total mass'') for the chordal $\SLE_\kappa$ curve from $x$ to 
$y$, introduced in~\cite{Lawler:SLE, Lawler:Partition_functions_loop_measure_and_versions_of_SLE} ---
see also~\cite{Dubedat:SLE_and_free_field, Kang-Makarov:Gaussian_free_field_and_conformal_field_theory}.
In the upper half-plane $\bH$, we have the simple formula $H_\bH(x,y) =  |x-y|^{-2}$, so 
\begin{align} \label{eq: 2ptf uhp}
\big\langle \Phi_{1,2}(x) \Phi_{1,2}(y) \big\rangle_{\bH} = |x-y|^{-2h_{1,2}}.
\end{align}
For general $\kappa> 0$, it is not obvious at all what the $2N$-point function 
$\big\langle \Phi_{1,2}(x_1) \cdots  \Phi_{1,2}(x_{2N}) \big\rangle_{\Omega}$ should be.
For some lattice models (e.g., the critical Ising model, the Gaussian free field), 
it can be understood in terms of specific solutions of the PDEs~(\ref{eqn:CFTmglePDE},~\ref{eq: multiple SLE PDEs}),
known as (symmetric)  partition functions for the multiple $\SLEk$ processes,
see~\cite[Section~\red{4.4}]{Peltola-Wu:Global_and_local_multiple_SLEs_and_connection_probabilities_for_level_lines_of_GFF}.
Morally, we expect that
\begin{align*}
\big\langle \Psi^{\textrm{alt}}(x_1,\ldots,x_{2N}) \big\rangle_{\Omega}
= \big\langle \Phi_{1,2}(x_1)  \cdots  \Phi_{1,2}(x_{2N}) \big\rangle_{\Omega}
= \sum_{\alpha \in \LP_N} \PartF_\alpha (\Omega; x_1,\ldots,x_{2N}) ,
\end{align*}
where $\alpha \in \LP_N$ are planar pair partitions 
and $\PartF_\alpha$ are (possibly constant multiples of) the pure
partition functions for the multiple $\SLEk$,
both discussed in Section~\ref{sec: Multiple SLE partition functions}.
See item~\ref{item::ising_crossing_proba} of Theorem~\ref{thm: summary}
in Section~\ref{subsec:scaling_limimt_results_etc}  for a rigorous statement of this sort.

Unfortunately, the mathematical meaning of the ``fields'' $\Phi_{1,2}$ is not really understood even for the Ising model.
The Pfaffian correlation functions~\eqref{eq: Ising pf general domain} 
do coincide with those of the energy density, or the free fermion, on the boundary,
but for neither field a well-defined scaling limit has been established.

\section{\label{sec: Multiple SLE partition functions}Multiple SLE partition functions and applications}

In this section, we discuss one way to make sense of correlation functions of type
$\big\langle \Phi_{1,2}(x_1)  \cdots \Phi_{1,2}(x_{2N}) \big\rangle$. 
Even if the nature of the ``fields'' $\Phi_{1,2}$ is not mathematically clear, 
functions $\PartF (x_1, \ldots, x_{2N}) = \big\langle \Phi_{1,2}(x_1) \cdots \Phi_{1,2}(x_{2N}) \big\rangle$
of $2N$ complex or real variables can still be defined.
Furthermore, these functions do satisfy properties predicted by CFT. 
On the other hand, they are also associated to (commuting) multiple $\SLEk$ processes
growing from  the boundary points $x_1, \ldots, x_{2N} \in \bR = \partial \bH$. 
The multiple $\SLEk$ processes have been shown to describe scaling limits of multiple interfaces in, e.g., the critical Ising
model~\cite{Izyurov:Smirnovs_observable_for_free_boundary_conditions_interfaces_and_crossing_probabilities,
BPW:On_the_uniqueness_of_global_multiple_SLEs}.
(To other variants of the $\SLEk$, certain other CFT correlation functions could be associated, see, 
e.g.,~\cite{BBK:Multiple_SLEs_and_statistical_mechanics_martingales, Kytola:On_CFT_of_SLE_kappa_rho,
Hongler-Kytola:Ising_interfaces_and_free_boundary_conditions}.)

We begin by briefly discussing the multiple $\SLEk$ processes in Section~\ref{subsec:multiple SLEs}.  
Then, in Section~\ref{subsec:ppfdef} we define the multiple SLE partition functions, to be interpreted as 
correlation functions of type $\big\langle \Phi_{1,2}(x_1)  \cdots \Phi_{1,2}(x_{2N}) \big\rangle$.
Section~\ref{subsec:ppfprop} is devoted to properties of these functions, and Section~\ref{subsec:scaling_limimt_results_etc}
to a brief overview of applications to critical lattice models and classification of SLEs.

\subsection{\label{subsec:multiple SLEs} Multiple SLEs and their partition functions}

One curve in a multiple $\SLEk$ (sampled from its marginal law) can be described via a Loewner chain 
similar to the usual chordal case~\eqref{eq: Loewner equation}, but where 
the Loewner driving function $W_t$ has a drift given by the interaction with the other marked boundary points.
On the upper half-plane $\bH$ with marked points $x_{1} < \cdots < x_{2N}$,
for the curve starting from $x_j$, with $j \in \{1,\ldots,2N\}$, we have
\begin{align}\label{eqn::marginalj}
\begin{split}
\begin{cases}
\ud W_t 
= \sqrt{\kappa} \; \ud B_t 
+ \kappa\partial_j \log\PartF \big(V_t^1, \ldots, V_t^{j-1}, W_t, V_t^{j+1},\ldots, V_t^{2N}\big) \; \ud t,  \\
\ud V_t^i = \frac{2 \; \ud t}{V_t^i-W_t}, \qquad \textnormal{for } i\neq j, 
\end{cases}
\qquad \qquad 
\begin{cases}
W_0 = x_j  , \\
V_0^i = x_i, \qquad \textnormal{for } i\neq j ,
\end{cases}
\end{split}
\end{align}
where $\PartF$ is a so-called multiple $\SLEk$ partition function, and $V_t^i$ are the time evolutions of the other marked 
points~\cite{Dubedat:Commutation_relations_for_SLE}.

\begin{remark} \label{eqn:SDE_well_def}
The system~\eqref{eqn::marginalj} of stochastic differential equations (SDE)
only makes sense locally, i.e., up to a certain stopping time.
However, with strong enough control of the partition function $\PartF$, the Loewner chain~\eqref{eqn::marginalj} 
can be shown to be well-defined including the time when the curve 
swallows some of the marked points  --- see Proposition~\ref{prop:cor}
and~\cite[Proposition~\red{4.9}]{Peltola-Wu:Global_and_local_multiple_SLEs_and_connection_probabilities_for_level_lines_of_GFF}
(with $\kappa \in (0,6]$ and $\PartF$ a pure partition function) 
and~\cite[Theorem~\red{5.8}]{Karrila:Multiple_SLE_local_to_global} (for examples arising from critical lattice interfaces).
For the single $\SLEk$ from $x_1$ to $x_2$,  
such a property
for the curve was proven in~\cite[Section~\red{7}]{Rohde-Schramm:Basic_properties_of_SLE}
(and in~\cite{LSW:Conformal_invariance_of_planar_LERW_and_UST} for the exceptional case $\kappa=8$).
It was also shown that almost surely, the curve $(\gamma(t))_{t > 0}$ hits the boundary $\partial \bH = \bR$
only at its endpoint $x_2$ if $\kappa \in (0,4]$, whereas if $\kappa > 4$, then 
the curve almost surely hits the boundary already before hitting $x_2$.
(See also Figure~\ref{fig: iterate}.)
\end{remark}

In fact, the $\SLEk$ type curve $\gamma$ driven by $W_t$, a solution to~\eqref{eqn::marginalj},
is a Girsanov transform of the chordal $\SLEk$ driven by $\sqrt{\kappa} B_t + x_j$ 
by a (local) martingale $M_t$ obtained from the partition function $\PartF$,
\begin{align} \label{eq:magle}
M_t = \prod_{\substack{1 \leq i \leq 2N \\ i\neq j}}g_t'(x_{i})^{h} \times
    \PartF \big(g_t(x_1),\ldots,g_t(x_{j-1}),\sqrt{\kappa}B_t+x_j,g_t(x_{j+1}),\ldots,g_t(x_{2N}) \big) ,
\end{align}
where $h = h_{1,2} = \frac{6-\kappa}{2\kappa}$,
and $g_t$ is the solution to the Loewner equation~\eqref{eq: Loewner equation}
with driving function $\sqrt{\kappa}B_t+x_j$.
In other words, the Radon-Nikodym derivative of the law of 
$\gamma$ with respect to the chordal $\SLEk$ is given by $M_t / M_0$, at least up to a stopping time.

\bigskip

For the SDEs~\eqref{eqn::marginalj}, 
it suffices to define the partition function $\PartF$ up to a multiplicative constant, which disappears in the logarithmic derivative.
By the requirement that $M_t$ is a local martingale, $\PartF$ should satisfy a certain
second order partial differential equation, stated in~\eqref{eq: multiple SLE PDEs} below (c.f. also~\eqref{eqn:CFTmglePDE}).
Such an equation holds symmetrically for all $j \in \{1,\ldots,2N\}$~\cite{Dubedat:Commutation_relations_for_SLE}.
These PDEs appear in the CFT literature as the null-field equations for 
correlations of the field $\Phi_{1,2}$,
which is exactly the field that should generate SLE type curves emerging from the 
boundary~\cite{Cardy:Conformal_invariance_and_surface_critical_behavior, Cardy:Boundary_conditions_fusion_rules_and_Verlinde_formula, 
Bauer-Bernard:Conformal_field_theories_of_SLEs, 
Bauer-Bernard:SLE_martingales_and_Virasoro_algebra,
Bauer-Bernard:Conformal_transformations_and_SLE_partition_function_martingale} ---
recall Sections~\ref{subsec:CFT}--\ref{subsec:mgles}.

As an easy example, let us consider two points $x_1 < x_2$. In this case, there is only one partition function
(up to a multiplicative constant), already appearing in Section~\ref{sec: preli}, 
Equations~\eqref{eq: 2ptf} and~\eqref{eq: 2ptf uhp}: 
$\big\langle \Phi_{1,2}(x_1) \Phi_{1,2}(x_2) \big\rangle_{\bH} = \PartF(x_1,x_2)$,
\begin{align}
\PartF(x_1,x_2)
= \PartF_{\vcenter{\hbox{\includegraphics[scale=0.2]{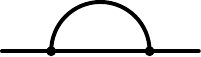}}}}(x_1,x_2) 
 = (x_2-x_1)^{(\kappa-6)/\kappa} .
\end{align}
The driving function $W_t$ of the curve starting from $x_1$ and the time evolution $V_t := V_t^2$
of the other point $x_2$ in~\eqref{eqn::marginalj} satisfy  
\begin{align}\label{eqn::marginalj1curve}
\begin{cases}
\ud W_t 
= \sqrt{\kappa} \ud B_t  + \frac{\kappa-6}{W_t-V_t} \; \ud t , \\
\ud V_t =  \frac{2 \; \ud t}{V_t-W_t}, 
\end{cases}
\qquad \qquad 
\begin{cases}
W_0 = x_1 , \\
V_0 = x_2 .
\end{cases}
\end{align}
This process is the chordal $\SLEk$ in $\bH$ from $x_1$ to $x_2$
--- in particular, by~\cite{LSW:Conformal_invariance_of_planar_LERW_and_UST, Rohde-Schramm:Basic_properties_of_SLE},
it defines a continuous curve that terminates at the point $x_2$.
Similarly, we can grow the curve starting from $x_2$.
In fact, the law of the chordal $\SLEk$ curve is reversible: 
if $\gamma \sim \mathbb{P}_{\bH;x_1,x_2}$, then the time-reversal $\bar{\gamma}$ of $\gamma$
has the law $\mathbb{P}_{\bH;x_2,x_1}$~\cite{Zhan:Reversibility_of_chordal_SLE, Sheffield-Miller:Imaginary_geometry3}.

As a slightly less trivial example, consider four points $x_1 < x_2 < x_3 < x_4$ and assume that $\kappa \in (0,8)$.
Then, the multiple $\SLE_\kappa$ partition functions are given by hypergeometric functions.
Specifically, any linear combination of the two functions
\begin{align} 
\label{eq: hg formulas for 4p fctions1}
\PartF_{\vcenter{\hbox{\includegraphics[scale=0.2]{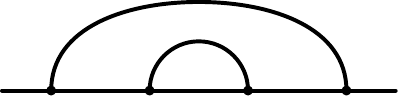}}}} (x_1,x_2,x_3,x_4)
:= \; & (x_4-x_1)^{-2h}(x_3-x_2)^{-2h} \left(\frac{(x_2-x_1)(x_4-x_3)}{(x_4-x_2)(x_3-x_1)}\right)^{2/\kappa} 
\; \frac{\hF\left(\frac{4}{\kappa}, 1-\frac{4}{\kappa}, \frac{8}{\kappa}; \frac{(x_2-x_1)(x_4-x_3)}{(x_4-x_2)(x_3-x_1)}\right)}{\hF\left(\frac{4}{\kappa}, 1-\frac{4}{\kappa}, \frac{8}{\kappa}; 1\right)} , \\ 
\label{eq: hg formulas for 4p fctions2}
\PartF_{\vcenter{\hbox{\includegraphics[scale=0.2]{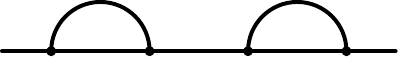}}}} (x_1,x_2,x_3,x_4)
:= \; &  (x_2-x_1)^{-2h}(x_4-x_3)^{-2h} \left(\frac{(x_4-x_1)(x_3-x_2)}{(x_4-x_2)(x_3-x_1)}\right)^{2/\kappa} 
\; \frac{\hF\left(\frac{4}{\kappa}, 1-\frac{4}{\kappa}, \frac{8}{\kappa}; \frac{(x_4-x_1)(x_3-x_2)}{(x_4-x_2)(x_3-x_1)}\right)}{\hF\left(\frac{4}{\kappa}, 1-\frac{4}{\kappa}, \frac{8}{\kappa}; 1\right)} ,
\end{align}%
is a partition function for a $2$-$\SLEk$ process.
(As a side remark, note that $\hF\left(\frac{4}{\kappa}, 1-\frac{4}{\kappa}, \frac{8}{\kappa}; z\right)$ is bounded for $z \in [0,1]$ 
when $\kappa\in (0,8)$, but infinite at $z = 1$ when $\kappa = 8$.)
Let us consider the curve starting from $x_1$, with driving function $W_t$ satisfying  
the SDEs~\eqref{eqn::marginalj} for $j=1$ and $N=2$. 
It can be shown~\cite[Theorem~\red{1.1}]{Wu:Convergence_of_the_critical_planar_ising_interfaces_to_hypergeometric_SLE} 
that with partition function $\PartF = \PartF_{\vcenter{\hbox{\includegraphics[scale=0.2]{figures/link-2.pdf}}}}$,
this curve terminates almost surely at $x_4$.
Similarly, taking the partition function $\PartF = \PartF_{\vcenter{\hbox{\includegraphics[scale=0.2]{figures/link-1.pdf}}}}$,
the curve terminates almost surely at $x_2$.

In general, it follows from J.~Dub\'edat's work~\cite[Theorem~\red{7}]{Dubedat:Commutation_relations_for_SLE} 
(see also~\cite[Theorem~\red{A.4}]{Kytola-Peltola:Pure_partition_functions_of_multiple_SLEs})
that so-called local multiple $\SLEk$ processes, generated via the Loewner chain~\eqref{eqn::marginalj},
are in one-to-one correspondence with the multiple $\SLEk$ partition functions $\mathcal{Z}(x_1, \ldots, x_{2N})$,
defined as positive functions that satisfy the PDE system~\eqref{eq: multiple SLE PDEs} and 
a specific conformal transformation property~\eqref{eq: multiple SLE Mobius covariance} 
stated in Section~\ref{subsec:ppfdef}. 
For the detailed definition of the local multiple $\SLEk$ processes, we refer 
to~\cite{Dubedat:Commutation_relations_for_SLE},~\cite[Appendix~\red{A}]{Kytola-Peltola:Pure_partition_functions_of_multiple_SLEs}, and~\cite[Section~\red{4.2}]{Peltola-Wu:Global_and_local_multiple_SLEs_and_connection_probabilities_for_level_lines_of_GFF}.

The existence of multiple $\SLEk$ partition functions is not clear for general $N \geq 3$.
When $\kappa \in (0,4]$,  
they can be constructed using the Brownian loop 
measure~\cite{Kozdron-Lawler:Configurational_measure_on_mutually_avoiding_SLEs,
Lawler:Partition_functions_loop_measure_and_versions_of_SLE,
Peltola-Wu:Global_and_local_multiple_SLEs_and_connection_probabilities_for_level_lines_of_GFF},
and the curves weighted by such partition functions (in the sense of Girsanov)
are absolutely continuous with respect to the chordal $\SLEk$ --- \eqref{eq:magle} is a true martingale. 
Unfortunately, when $\kappa > 4$, the Brownian loop measures 
appearing in the construction become infinite, so this approach does not work as such. 
There is another construction avoiding the Brownian loop measure, which is 
currently rigorously performed for $\kappa \in (0,6]$~\cite[Section~\red{6}]{Wu:Convergence_of_the_critical_planar_ising_interfaces_to_hypergeometric_SLE}.
We will discuss this approach in Appendix~\ref{app:Hao}.
For the range $\kappa > 6$, no construction is known to date. In Appendix~\ref{app:Hao}, we also discuss 
how the case of $\kappa \in (6,8)$ could be treated, if certain technical estimates could be established.

\subsection{\label{subsec:ppfdef} Definition of the multiple SLE partition functions}

Now we give a PDE theoretic definition and classification of the multiple $\SLE_\kappa$ partition functions 
(relaxing the positivity assumption --- see Remark~\ref{rem:positiovity}).
Our definition is motivated by J.~Dub\'edat's work~\cite{Dubedat:Commutation_relations_for_SLE}, 
where he derived properties that the partition functions must satisfy. 
These properties were further investigated in many works, 
e.g.,~\cite{Graham:Multiple_SLEs,
Kozdron-Lawler:Configurational_measure_on_mutually_avoiding_SLEs,
Lawler:Partition_functions_loop_measure_and_versions_of_SLE, Kytola-Peltola:Pure_partition_functions_of_multiple_SLEs,
Peltola-Wu:Global_and_local_multiple_SLEs_and_connection_probabilities_for_level_lines_of_GFF}. 
A physical derivation with CFT interpretations appears in~\cite{BBK:Multiple_SLEs_and_statistical_mechanics_martingales}
--- see also~\cite{Flores-Kleban:Solution_space_for_system_of_null-state_PDE1, Flores-Kleban:Solution_space_for_system_of_null-state_PDE2, 
Flores-Kleban:Solution_space_for_system_of_null-state_PDE3, Flores-Kleban:Solution_space_for_system_of_null-state_PDE4},
and recall the discussion in Section~\ref{subsec:mgles} for statistical physics motivation.

\bigskip

Fix a parameter $\kappa \in (0,8)$.
For each $N \geq 1$, consider functions
$\PartF \colon \chamber_{2N} \to \bC$ defined 
on the configuration space
\begin{align} \label{eq: chamber}
\chamber_{2N} :=\; & \{ (x_{1},\ldots,x_{2N}) \in \bR^{2N} \; | \; x_{1} < \cdots < x_{2N} \} .
\end{align}
We assume that $\PartF$ satisfy the following three properties:
\begin{itemize}
\item[\red{$\mathrm{(COV)}$}] {\bf \textit{M\"obius covariance}}: 
With conformal weight $h = \frac{6-\kappa}{2\kappa} (= h_{1,2})$,
we have the covariance rule 
\begin{align} \label{eq: multiple SLE Mobius covariance}
\PartF(x_{1},\ldots,x_{2N}) = \; &
\prod_{i=1}^{2N} \Mob'(x_{i})^{h} 
\times \PartF(\Mob(x_{1}),\ldots,\Mob(x_{2N})) ,
\end{align}
for all M\"obius maps $\Mob \colon \bH \to \bH$ 
such that $\Mob(x_{1}) < \cdots < \Mob(x_{2N})$.

\item[\red{$\mathrm{(PDE)}$}] {\bf \textit{Partial differential equations of second order}}: 
We have
\begin{align} \label{eq: multiple SLE PDEs}
\left[ \frac{\kappa}{2}\pdder{x_i} + \sum_{\substack{1 \leq j \leq 2N \\ j \neq i}} \left(\frac{2}{x_{j}-x_{i}}\pder{x_j} - 
\frac{2h}{(x_{j}-x_{i})^{2}}\right) \right] 
\PartF(x_1,\ldots,x_{2N}) =  0 , \qquad \textnormal{for all } i \in \{1,\ldots,2N\} .
\end{align}

\item[\red{$\mathrm{(PLB)}$}] {\bf \textit{power law bound}}: 
There exist $C>0$ and $p>0$ such that, for all 
$N \geq 1$ and for all $(x_1,\ldots, x_{2N}) \in \chamber_{2N}$, we have
\begin{align}\label{eqn::powerlawbound}
|\PartF(x_1, \ldots, x_{2N})|  \leq \;  C \prod_{1 \leq i<j \leq 2N}(x_j-x_i)^{\mu_{ij}(p)}, 
\qquad \qquad 
\textnormal{where } \quad
\mu_{ij}(p) :=
\begin{cases}
p, \quad &\textnormal{if } |x_j-x_i| > 1, \\
-p, \quad &\textnormal{if } |x_j-x_i| < 1.
\end{cases}
\end{align}
\end{itemize}

By~\cite[Theorem~\red{8}]{Flores-Kleban:Solution_space_for_system_of_null-state_PDE3} 
(see also~\cite{Dubedat:Euler_integrals_for_commuting_SLEs, Dubedat:Commutation_relations_for_SLE}), 
for each $N \geq 1$, the solution space
\begin{align} \label{eq: solution space}
\Sol_N := \{ \PartF \colon \chamber_{2N} \to \bC \; | \; \PartF \textnormal{ satisfies $\mathrm{(COV)}$, $\mathrm{(PDE)}$~\&~$\mathrm{(PLB)}$} \}
\end{align}
is finite-dimensional and it consists of so-called Coulomb gas integral solutions 
(see also Appendix~\ref{app: Coulomb gas}).

\begin{theorem} \label{thm:Steven}
\textnormal{\cite[Theorem~\red{8}]{Flores-Kleban:Solution_space_for_system_of_null-state_PDE3}} 
\;
For each $N \geq 1$, we have $\dmn \Sol_N = \Catalan_N := \frac{1}{N+1} \binom{2N}{N}$.
\end{theorem}

Key arguments in~\cite{Flores-Kleban:Solution_space_for_system_of_null-state_PDE1, Flores-Kleban:Solution_space_for_system_of_null-state_PDE2, 
Flores-Kleban:Solution_space_for_system_of_null-state_PDE3, Flores-Kleban:Solution_space_for_system_of_null-state_PDE4}  
include explicit analysis of boundary behavior of the solutions in $\Sol_N$.
Note that the PDEs in~\eqref{eq: multiple SLE PDEs} are singular on the diagonals $x_i = x_j$, for $i \neq j$.
Consequently, the usual theory of elliptic and hypoelliptic PDEs can only be applied away from the boundary of $\chamber_{2N}$.
However, in~\cite{Flores-Kleban:Solution_space_for_system_of_null-state_PDE1, Flores-Kleban:Solution_space_for_system_of_null-state_PDE2} 
S.~Flores and P.~Kleban successfully applied Schauder interior estimates
and  elliptic PDE theory to establish the upper bound $\Catalan_N$ for $\dmn \Sol_N$.
To obtain the lower bound $\Catalan_N$ for $\dmn \Sol_N$, one constructs a linearly independent set of solutions
with cardinality $\Catalan_N$, see~\cite{Flores-Kleban:Solution_space_for_system_of_null-state_PDE3,
Kytola-Peltola:Pure_partition_functions_of_multiple_SLEs}.

\begin{remark} \label{rem:positiovity}
In order to generate local multiple $\SLEk$ processes via the Loewner evolution~\eqref{eqn::marginalj},
the multiple $\SLEk$ partition functions $\PartF$ \`a la Dub\'edat~\cite{Dubedat:Commutation_relations_for_SLE} 
are defined as \emph{positive} 
solutions to $\mathrm{(PDE)}$ and $\mathrm{(COV)}$.
The former property $\mathrm{(PDE)}$ implies that~\eqref{eq:magle} is a local martingale.
The latter property $\mathrm{(COV)}$ arises naturally from the conformal invariance and domain Markov property of the $\SLE_\kappa$ curve.
The positivity of the functions is manifest, e.g., in order for~\eqref{eq:magle} to be a positive local martingale.

In conclusion, only positive functions $\PartF \colon \chamber_{2N} \to \bRpos$
in $\Sol_N$ are multiple $\SLEk$ partition functions in the sense of Section~\ref{subsec:multiple SLEs}.
On the other hand, a multiple $\SLEk$ partition function does not have to satisfy the bound~$\mathrm{(PLB)}$,
but in all known examples, this bound is satisfied nevertheless.
In fact, when $\kappa \in (0,6]$, the solution space $\Sol_N$ has a basis consisting of positive solutions,
so nothing is lost by relaxing the positivity in this case.
The same property is believed to hold also when $\kappa \in (6,8)$.
\end{remark}

It is convenient to index basis elements for $\Sol_N$
by planar pair partitions $\alpha$ of the integers $\{1,2,\ldots,2N\}$
--- indeed, for each $N$, there are exactly $\Catalan_N$ such planar pair partitions $\alpha$.
We denote the set of them by $\LP_N$, and we call elements $\alpha$ in this set ``link patterns''. 
We also denote the collection of link patterns with any number of links (including zero) by
\begin{align}
\LP := \bigsqcup_{N\geq0} \LP_N .
\end{align}

\begin{figure}[h!]
\centering
\includegraphics[scale=1]{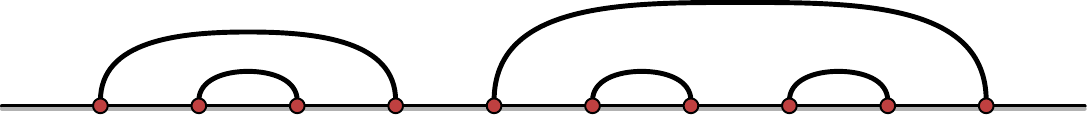}
\caption{\label{fig: nonvallp}
Graphical illustration of a link pattern $\alpha \in \LP_5$
(i.e., a planar pair partition of $\{1,2,\ldots,10\}$).
}
\end{figure}

$\Sol_N$ has certain physically important bases.
One of them, proposed earlier by J.~Dub\'edat~\cite{Dubedat:Euler_integrals_for_commuting_SLEs}, 
was investigated by
S.~Flores and P.~Kleban in~\cite{Flores-Kleban:Solution_space_for_system_of_null-state_PDE1, Flores-Kleban:Solution_space_for_system_of_null-state_PDE2,
Flores-Kleban:Solution_space_for_system_of_null-state_PDE3, Flores-Kleban:Solution_space_for_system_of_null-state_PDE4}. 
Using this explicit basis $\{ \sF_{\alpha} \; | \; \alpha \in \LP_N \}$, they (non-rigorously) argued 
that its dual basis with respect to a certain bilinear form is closely related to crossing probabilities 
in critical models in statistical physics. Elements in this dual basis were called ``connectivity weights''
and denoted by $\Pi_\alpha$. 
Instead of $\Pi_\alpha$, we denote this dual basis by $\{ \PartF_\alpha \; | \; \alpha \in \LP_N \}$, 
following the notation in the author's work~\cite{Kytola-Peltola:Pure_partition_functions_of_multiple_SLEs, 
Peltola-Wu:Global_and_local_multiple_SLEs_and_connection_probabilities_for_level_lines_of_GFF} with K.~Kyt\"ol\"a and H.~Wu.

\begin{defn}
The functions $\PartF_\alpha$ are defined in terms of properties that uniquely characterize them:
the collection (if it exists)
\begin{align} \label{eq: collectionPartF}
\{ \PartF_\alpha \; | \; \alpha \in \LP \} , \qquad \textnormal{with } \kappa\in (0,8) ,
\end{align}
is uniquely determined by the normalization convention 
$\PartF_{\emptyset} \equiv 1$,  
for the empty link pattern $\emptyset \in \LP_0$,
and the requirements that, first, we have $\PartF_\alpha \in \Sol_N$, for all $\alpha \in \LP_N$ and $N \geq 1$,
and second, the following recursive asymptotics properties $\mathrm{(ASY)}$ hold:
\begin{itemize}
\item[\red{$\mathrm{(ASY)}$}] {\bf \textit{Asymptotics}}: 
For all $N \geq 1$, for all $\alpha \in \LP_N$, and for all $j \in \{1, \ldots, 2N-1 \}$ and $\xi \in (x_{j-1}, x_{j+2})$, we have
\begin{align} \label{eq: multiple SLE asymptotics}
\lim_{x_j , x_{j+1} \to \xi} 
\frac{\PartF_\alpha(x_1 , \ldots , x_{2N})}{(x_{j+1} - x_j)^{-2h}} 
= \; & \begin{cases}
0 , \quad &
    \textnormal{if } \link{j}{j+1} \notin \alpha , \\
\PartF_{\hat{\alpha}}(x_{1},\ldots,x_{j-1},x_{j+2},\ldots,x_{2N}) , &
    \textnormal{if } \link{j}{j+1} \in \alpha ,
\end{cases} 
\end{align}
where $\hat{\alpha} = \alpha \removeLink \link{j}{j+1} \in \LP_{N-1}$ denotes
the link pattern obtained from $\alpha$ by removing the link $\link{j}{j+1}$ 
(and relabeling the remaining indices by $1,2,\ldots,2N-2$).
\end{itemize}
Asymptotics properties $\mathrm{(ASY)}$ can be regarded as
boundary conditions for PDE system~\eqref{eq: multiple SLE PDEs},
or as a specified operator product expansion (OPE) if the functions $\PartF_\alpha$
are viewed as correlation functions of 
some ``conformal fields'' --- see Sections~\ref{sec:OPE}--\ref{sec:Malek}.
These asymptotics properties~\eqref{eq: multiple SLE asymptotics} were proposed 
in the work~\cite{BBK:Multiple_SLEs_and_statistical_mechanics_martingales} of M.~Bauer, D.~Bernard, and K.~Kyt\"ol\"a. 
\end{defn}

The power law bound $\mathrm{(PLB)}$ stated in~\eqref{eqn::powerlawbound} might not be a necessary property
but instead a consequence of the other properties of $\PartF_\alpha$. However, the current proof of uniqueness of 
these functions strongly relies on this technical property. The uniqueness is established by virtue of the following lemma,
whose proof constitutes the whole article~\cite{Flores-Kleban:Solution_space_for_system_of_null-state_PDE2}:

\begin{prop}\label{prop::purepartition_unique}
\textnormal{\cite[Lemma~\red{1}, paraphrased]{Flores-Kleban:Solution_space_for_system_of_null-state_PDE2}}
\;
Let $\kappa\in (0,8)$. 
\begin{enumerate}
\itemcolor{red}
\item If $F \in \Sol_N$ satisfies the asymptotics
\begin{align*}
\lim_{x_j , x_{j+1} \to \xi} 
\frac{F(x_1 , \ldots , x_{2N})}{(x_{j+1} - x_j)^{-2h}} = 0 , 
\end{align*}
for all $ j \in \{ 2, 3, \ldots , 2N-1 \}$ and for all  $\xi \in (x_{j-1}, x_{j+2})$,
then $F \equiv 0$. 

\item In particular, if $\{\PartF_\alpha \; | \; \alpha \in \LP\}$ is a collection of functions  $\PartF_\alpha \in \mathcal{S}_N$, 
for $\alpha \in \LP_N$ and $N \geq 1$, satisfying the recursive asymptotics properties~\eqref{eq: multiple SLE asymptotics}
in $\mathrm{(ASY)}$ and the normalization $\PartF_\emptyset = 1$,
then this collection $\{\PartF_\alpha \; | \; \alpha \in \LP\}$ is unique.
\end{enumerate}
\end{prop}

The above proposition says nothing about the existence of the functions $\PartF_\alpha$. 
In~\cite{Flores-Kleban:Solution_space_for_system_of_null-state_PDE3}, $\PartF_\alpha$ were implicitly defined 
in terms of a dual space of certain allowable sequences of limits.
An explicit construction for $\PartF_\alpha$ in Coulomb gas integral form
(via integrals similar to, but yet slightly different than in~\cite{Flores-Kleban:Solution_space_for_system_of_null-state_PDE3})
was given in~\cite{Kytola-Peltola:Pure_partition_functions_of_multiple_SLEs} for all $\kappa \in (0,8) \setminus \bQ$
(see  Appendix~\ref{app: Coulomb gas}).
On the other hand, in~\cite{Peltola-Wu:Global_and_local_multiple_SLEs_and_connection_probabilities_for_level_lines_of_GFF}
an explicit probabilistic construction of the functions $\PartF_\alpha$ was given for all $\kappa \in (0,4]$,
following the ideas of M.~Kozdron and G.~Lawler~\cite{Kozdron-Lawler:Configurational_measure_on_mutually_avoiding_SLEs} 
and relating these functions to multiple SLEs.
This construction uses the Brownian loop measure and fails when $\kappa > 4$.
Another construction, somewhat similar in spirit but more suitable for SLE curves with self-touchings,
was given in~\cite[Section~\red{6}]{Wu:Convergence_of_the_critical_planar_ising_interfaces_to_hypergeometric_SLE}. 
Currently, this construction works for $\kappa \in (0,6]$, as discussed in Appendix~\ref{app:Hao}.

\bigskip

It follows from either probabilistic 
construction~\cite{Kozdron-Lawler:Configurational_measure_on_mutually_avoiding_SLEs,
Peltola-Wu:Global_and_local_multiple_SLEs_and_connection_probabilities_for_level_lines_of_GFF,
Wu:Convergence_of_the_critical_planar_ising_interfaces_to_hypergeometric_SLE} 
that each function $\PartF_\alpha$ in fact satisfies a bound significantly stronger than~\eqref{eqn::powerlawbound}:
\begin{itemize}
\item[\red{$\mathrm{(B)}$}] {\bf \textit{``Strong'' power law bound}}: 
Let $\kappa \in (0,6]$. Then, for all $N \geq 1$ and $\alpha \in \LP_N$, and for all $(x_1,\ldots, x_{2N}) \in \chamber_{2N}$, we have
\begin{align} \label{eqn::partitionfunction_positive000}
0<\PartF_\alpha(x_1, \ldots, x_{2N}) 
\le \prod_{\link{a}{b} \in \alpha} |x_{b}-x_{a}|^{-2h} .
\end{align}
\end{itemize}

The upper bound in~\eqref{eqn::partitionfunction_positive000} depends on $\alpha \in \LP_N$.
It is very useful for establishing fine properties of the functions $\PartF_\alpha$.
The lower bound shows that all functions in the collection~\eqref{eq: collectionPartF}
with $\kappa \in (0,6]$
are not only real-valued but also positive,
which is crucial for relating them to multiple $\SLEk$ processes 
via the SDEs~\eqref{eqn::marginalj} (as discussed in Remark~\ref{rem:positiovity}), 
as well as to crossing probabilities 
of critical models in statistical physics~\cite{Peltola-Wu:Crossing_probabilities_of_multiple_Ising_interfaces,
Peltola-Wu:Global_and_local_multiple_SLEs_and_connection_probabilities_for_level_lines_of_GFF}.
Indeed, it has now been proven for $\kappa \in (0,4]$ that
the functions $\PartF_\alpha$ give rise to multiple $\SLEk$ processes
with prescribed connectivity of the curves according to the pairing $\alpha$~\cite{Peltola-Wu:Global_and_local_multiple_SLEs_and_connection_probabilities_for_level_lines_of_GFF}. 
In light of the construction of the functions $\PartF_\alpha$ for $\kappa \in (4,6]$, discussed in Appendix~\ref{app:Hao},
similar arguments should extend to this range --- see Proposition~\ref{prop:cor} and the discussion after it in Section~\ref{subsec:scaling_limimt_results_etc}.
Furthermore, rigorous connections with  crossing probabilities in critical models
(the Ising model, Gaussian free field, and loop-erased random 
walks) have been established~\cite{Peltola-Wu:Crossing_probabilities_of_multiple_Ising_interfaces,
Peltola-Wu:Global_and_local_multiple_SLEs_and_connection_probabilities_for_level_lines_of_GFF,
KKP:Conformal_blocks_pure_partition_functions_and_KW_binary_relation}
--- see Theorem~\ref{thm: summary} in Section~\ref{subsec:scaling_limimt_results_etc} for an example.

\bigskip

So far, we have discussed the pure partition functions as functions of real variables $x_1 < \cdots < x_{2N}$.
However, they can also be defined in other simply connected domains $\Omega \subsetneq \bC$ 
via their conformal covariance property. Namely, if $x_1, \ldots, x_{2N} \in \partial \Omega$ are $2N$ distinct boundary points 
appearing in counterclockwise order on 
sufficiently regular boundary segments (from the point of view of derivatives 
of conformal maps existing in their vicinity), then we set
\begin{align} \label{eq: ppf def in polygon}
\PartF_\alpha(\Omega; x_1, \ldots, x_{2N}) 
:= \; & \prod_{i=1}^{2N} |\Mob'(x_i)|^{h} \times 
\PartF_\alpha(\Mob(x_1), \ldots, \Mob(x_{2N})),
\end{align}
where $\Mob$ is any conformal map from $\Omega$ onto $\bH$ such that 
$\Mob(x_1) < \cdots < \Mob(x_{2N})$.
It is worthwhile to note that when considering ratios of partition functions,
the regularity assumptions for the boundary of $\Omega$ can be relaxed,
and it then suffices to require that conformal maps (but not necessarily their derivatives) exist
in the vicinity of the marked points $x_1, \ldots, x_{2N} \in \partial \Omega$.

For the case of $\Omega = \bH$ and $x_1 < \cdots < x_{2N}$, we still use the shorter notation
\begin{align*}
\PartF_\alpha(x_1, \ldots, x_{2N})  = \PartF_\alpha(\bH; x_1, \ldots, x_{2N}) .
\end{align*}
We also remark that the asymptotics~\eqref{eq: multiple SLE asymptotics} in property $\mathrm{(ASY)}$
holds for $j=2N$ as well, with 
$x_1 \to -\infty$ and $x_{2N} \to +\infty$.
(However, this property is not necessary for the definition of $\PartF_\alpha$.)
In general, given a ``polygon'' 
$(\Omega; x_1, \ldots, x_{2N})$ and $\alpha \in \LP_N$,
the asymptotics property $\mathrm{(ASY)}$ can be written in the form
\begin{align} \label{eq: multiple SLE asymptotics GEN}
\lim_{x_j , x_{j+1} \to \xi} 
\frac{\PartF_\alpha(\Omega; x_1 , \ldots , x_{2N})}{H_\Omega(x_j,x_{j+1})^{h}} 
= \; & \begin{cases}
0 , \quad &
    \textnormal{if } \link{j}{j+1} \notin \alpha , \\
\PartF_{\hat{\alpha}}(\Omega; x_{1},\ldots,x_{j-1},x_{j+2},\ldots,x_{2N}) , &
    \textnormal{if } \link{j}{j+1} \in \alpha ,
\end{cases} 
\end{align}
where $H_\Omega$ is the boundary Poisson kernel in $\Omega$. 
We may also allow ``$\link{j}{j+1} = \link{2N}{1}$'' in this formula.
Similarly, the strong bound~\eqref{eqn::partitionfunction_positive000} can be written in the form
\begin{align} \label{eqn::partitionfunction_positive_in_polygon}
0 < \PartF_\alpha(\Omega;x_1, \ldots, x_{2N}) 
\le \prod_{\link{a}{b} \in \alpha} H_\Omega(x_{a},x_{b})^{h} , \qquad \textnormal{when } \kappa\in (0,6] .
\end{align}

\subsection{\label{subsec:ppfprop} Properties of the multiple SLE partition functions}

The main purpose of this section is to collect known results for the functions $\PartF$ in the solution space $\Sol_N$
and to discuss open problems related to them.
From the CFT point of view, the  defining properties $\mathrm{(COV)}$ and $\mathrm{(PDE)}$ of $\PartF \in \Sol_N$
are manifest for correlation functions of the primary fields $\Phi_{1,2}$ ---
recall from Section~\ref{sec: preli} and Appendix~\ref{app:Vir}
the conformal covariance postulate~\eqref{eq: correlation function Mobius covariance} 
and  PDEs~(\ref{eq: PDE for correlation functions},~\ref{eq: singular equation level two simplified}) for fields with degeneracy at level two. 
Furthermore, the defining asymptotics properties $\mathrm{(ASY)}$ for the basis functions $\PartF_\alpha$,
stated in~\eqref{eq: multiple SLE asymptotics}, 
reflect a ``fusion structure'' (operator product expansion), which we shall discuss in detail in Section~\ref{sec:OPE} (see Remark~\ref{rem:fusion}).
Asymptotics properties $\mathrm{(ASY)}$ are also natural for the identification of $\PartF_\alpha$ as those multiple $\SLEk$ pure partition functions 
which generate curves with prescribed planar connectivity $\alpha$, see Proposition~\ref{prop:extr}.
Finally, these functions also describe crossing probabilities in critical planar models 
(as detailed in item~\ref{item::ising_crossing_proba} of Theorem~\ref{thm: summary} for the Ising model),
and their asymptotics properties are also crucial from this point of view.

\bigskip

Theorem~\ref{thm::purepartition_existence_forallK} below
supplements \cite[Theorem~\red{15}]{Dubedat:SLE_and_Virasoro_representations_fusionB},
\cite[Theorem~\red{8}]{Flores-Kleban:Solution_space_for_system_of_null-state_PDE3}, 
\cite[Theorem~\red{4.1}]{Kytola-Peltola:Pure_partition_functions_of_multiple_SLEs}, 
\cite[Proposition~\red{6.1}]{Wu:Convergence_of_the_critical_planar_ising_interfaces_to_hypergeometric_SLE},
and~\cite[Theorem~\red{1.1}]{Peltola-Wu:Global_and_local_multiple_SLEs_and_connection_probabilities_for_level_lines_of_GFF}.
It states the existence and uniqueness of the pure partition functions $\PartF_\alpha$
and some additional properties for them: linear independence (property~\ref{item: linearly independent}),
a strong growth bound (property~\ref{item: positive}), 
another natural bound (property~\ref{item: Malek}), whose role will be discussed in Section~\ref{sec:Malek},
as well as fusion properties~\ref{item: baby fusion2} and~\ref{item: OPE} 
related to the operator product expansion hierarchy for $\PartF_\alpha$,  discussed in detail in Section~\ref{sec:OPE}.
After Theorem~\ref{thm::purepartition_existence_forallK},
we include a short proof mainly indicating the relevant literature. 
We then discuss limitations of these results and further questions and problems.

\begin{theorem}
\label{thm::purepartition_existence_forallK}
Let $\kappa \in (0,8)$. 
There exists a unique collection $\{\PartF_\alpha \; | \; \alpha\in \LP\}$ of smooth
functions $\PartF_\alpha \in \Sol_N$, for $\alpha \in \LP_N$, such that $\PartF_\emptyset = 1$ 
and the recursive asymptotics properties~\eqref{eq: multiple SLE asymptotics} in $\mathrm{(ASY)}$ hold. 
These functions have the following further properties:
\begin{enumerate}
\itemcolor{red}
\item  \label{item: linearly independent}
For each $N \geq 0$, the functions in $\{\PartF_\alpha \; | \; \alpha\in \LP_N\}$ are linearly independent. 

\item  \label{item: positive} 
If $\kappa\in (0,6]$, then, for each $\alpha \in \LP$, 
the function $\PartF_\alpha$ is positive and satisfies the ``strong'' power law bound $\mathrm{(B):}$
\begin{align} \label{eqn::partitionfunction_positive}
0<\PartF_\alpha(x_1, \ldots, x_{2N}) 
\le \prod_{\link{a}{b} \in \alpha} |x_{b}-x_{a}|^{-2h} .
\end{align}

\item  \label{item: Malek}
 If $\kappa\in (0,6]$, then, for each $\alpha \in \LP_N$, the function $\PartF_\alpha$ satisfies the power law bound
\begin{align} \label{eq:Maleks bound0}
0 < \PartF_\alpha(x_1, \ldots, x_{2N})  \leq \; \prod_{i = 1}^{2N} \big( \min_{j \neq i} |x_i-x_j| \big)^{-h} .
\end{align}

\item \label{item: baby fusion2} 

Let $\kappa \in (0,8) \setminus \bQ$. 
Let $\alpha \in \LP_N$ and suppose that $\link{1}{2} \notin \alpha$. 
Then, for any $\xi < x_3$, the limit 
\begin{align} \label{eq: fusion limit example}
\hat{\PartF}_\alpha(\xi, x_3, \ldots, x_{2N}) := 
\lim_{x_1,x_2  \to \xi} \frac{\PartF_\alpha(x_1 , \ldots , x_{2N})}{(x_2 - x_1)^{2/\kappa}} 
\end{align}
exists and defines a solution to a system of $2N-1$ PDEs
given in Equation~\eqref{eq: PDEs after one fusion}. 
The limit function is M\"obius covariant: 
\begin{align*} 
\hat{\PartF}_\alpha(\xi, x_3, \ldots, x_{2N}) = 
\Mob'(\xi)^{(8-\kappa)/\kappa} \prod_{i = 3}^{2N} \Mob'(x_{i})^{h} 
\times 
\hat{\PartF}_\alpha \left( \Mob(\xi), \Mob(x_3), \ldots, \Mob(x_{2N}) \right) ,
\end{align*}
for all M\"obius maps $\Mob \colon \bH \to \bH$ 
such that $\Mob(\xi) < \Mob(x_{3}) < \cdots < \Mob(x_{2N})$.
Furthermore, such a limiting procedure can be iterated to construct solutions to higher order
PDEs of type~\eqref{eq: BSA differential equations}, as discussed in Section~\ref{sec:OPE}.

\item \label{item: OPE}
Let $\kappa \in (0,8) \setminus \bQ$. 
The collection $\{\PartF_\alpha \; | \; \alpha\in \LP\}$ satisfies an operator product expansion 
detailed in Proposition~\ref{prop: OPE}.
\end{enumerate}
\end{theorem}
\begin{proofIdea}
Uniqueness follows from Proposition~\ref{prop::purepartition_unique}, and
existence was implicitly argued in~\cite[Theorem~\red{8}]{Flores-Kleban:Solution_space_for_system_of_null-state_PDE3}. 
Property~\ref{item: linearly independent} follows from the results 
in~\cite{Flores-Kleban:Solution_space_for_system_of_null-state_PDE3, Kytola-Peltola:Pure_partition_functions_of_multiple_SLEs}
--- a short proof using the ideas from the previous literature 
is given in~\cite[Proposition~\red{4.5}]{Peltola-Wu:Global_and_local_multiple_SLEs_and_connection_probabilities_for_level_lines_of_GFF} 
for the case of $\kappa \in (0,4]$, and exactly the same proof works for $\kappa \in (4,8)$ as well. 
Property~\ref{item: positive} was proved in~\cite[Lemma~\red{4.1}]{Peltola-Wu:Global_and_local_multiple_SLEs_and_connection_probabilities_for_level_lines_of_GFF} 
for the case of $\kappa \in (0,4]$ 
and extended in~\cite[Proposition~\red{6.1}]{Wu:Convergence_of_the_critical_planar_ising_interfaces_to_hypergeometric_SLE} 
to the range $\kappa \in (0,6]$. 
Property~\ref{item: Malek} is a direct consequence of the bound in property~\ref{item: positive},
see Proposition~\ref{prop:Maleks bound} for the calculation.
Finally, property~\ref{item: baby fusion2} follows from Theorem~\ref{thm: big prop} 
appearing in Section~\ref{subsec: fusion SCCG},
and property~\ref{item: OPE} from Proposition~\ref{prop: OPE} appearing in Section~\ref{subsec:OPE for ppf}. 
\end{proofIdea}

Theorem~\ref{thm::purepartition_existence_forallK}  does not 
give a complete understanding of the pure partition functions $\{\PartF_\alpha \; | \; \alpha\in \LP\}$.
Indeed, properties~\ref{item: positive} and~\ref{item: Malek} have only been proven for $\kappa\in (0,6]$,
and properties~\ref{item: baby fusion2} and~\ref{item: OPE} for $\kappa \in (0,8) \setminus \bQ$. 
We list some unanswered questions below.

\begin{problem}
Construct the functions $\{\PartF_\alpha \; | \; \alpha\in \LP\}$ explicitly for $\kappa \in (6,8) \cap \bQ$.
\end{problem}
In~\cite[Theorem~\red{4.1}]{Kytola-Peltola:Pure_partition_functions_of_multiple_SLEs},
the functions $\{\PartF_\alpha \; | \; \alpha\in \LP\}$ were  
explicitly constructed for all $\kappa \in (0,8) \setminus \bQ$, using a quantum group method. 
The restriction that $\kappa$ is irrational is needed because 
the representation theory of the quantum group is required to be semisimple.
In principle, the functions thus obtained could be analytically continued to include all $\kappa \in (0,8)$,
but the explicit continuation is not obvious, due to delicate cancellations of infinities and zeroes.
(See also Appendix~\ref{app: Coulomb gas}.)

In Appendix~\ref{app:Hao}, we discuss another, probabilistic construction
from~\cite{Wu:Convergence_of_the_critical_planar_ising_interfaces_to_hypergeometric_SLE}, for $\kappa \in (0,6]$.
This construction might also work for the remaining range $\kappa \in (6,8) \cap \bQ$.
In Appendix~\ref{app:Hao},
we will discuss the technical difficulties for establishing this case.

\begin{problem}
Prove property~\ref{item: positive} for $\kappa \in (6,8)$.
\end{problem}
Property~\ref{item: positive} holds for all $\kappa \in (0,8)$ in the case of $N=2$, as can be seen by inspection of the explicit
formulas~\eqref{eq: hg formulas for 4p fctions1}--\eqref{eq: hg formulas for 4p fctions2} for
the two functions 
$\PartF_{\vcenter{\hbox{\includegraphics[scale=0.2]{figures/link-2.pdf}}}}$
and
$\PartF_{\vcenter{\hbox{\includegraphics[scale=0.2]{figures/link-1.pdf}}}}$.
The main trouble for the case of $\kappa \in (6,8)$ and $N \geq 3$
 is that the scaling exponent $h$ in~\eqref{eqn::partitionfunction_positive} is negative,
which results in technical difficulties in the probabilistic approach (discussed in Appendix~\ref{app:Hao}, see Lemma~\ref{lem: bound}).
On the other hand, the Coulomb gas integral approach 
of~\cite{Flores-Kleban:Solution_space_for_system_of_null-state_PDE3, Kytola-Peltola:Pure_partition_functions_of_multiple_SLEs}
does not seem to easily give a bound as strong as~\eqref{eqn::partitionfunction_positive}.

\begin{quest}
Does property~\ref{item: Malek} hold for $\kappa \in (6,8)$?
\end{quest}
The conformal weight $h = h_{1,2}$ is negative when $\kappa > 6$, 
whereas it is positive for $\kappa \in (0,6)$ and zero for $\kappa=6$.
The negative conformal weight spoils unitarity of the corresponding CFT, but from the $\SLEk$ point of view, nothing should really change. 
However, if the bound~\eqref{eq:Maleks bound0} fails for $\kappa \in (6,8)$,
this might indicate something interesting for the $\SLE_6$.
(See also Conjecture~\ref{conj:Malek}.)

\begin{quest}
Is there a hidden phase transition for the $\SLEk$ at $\kappa = 6$?
\end{quest}

The fusion procedure in properties~\ref{item: baby fusion2} and~\ref{item: OPE}
should imply that the functions obtained as limits of the pure partition functions 
satisfy strong bounds analogous to property~\ref{item: positive}, with appropriate conformal weights. For example:

\begin{problem} 
Prove that the function $\hat{\PartF}_\alpha$ in property~\ref{item: baby fusion2} satisfies a bound of type
\begin{align} \label{eq: bound example}
\hat{\PartF}_\alpha (\xi, x_3, \ldots, x_{2N}) 
\le \; C(\kappa) \; |\xi - x_{\alpha(1)}|^{-h_{1,3}} |\xi - x_{\alpha(2)}|^{-h_{1,3}}
|x_{\alpha(2)} - x_{\alpha(1)}|^{h_{1,3}-2h_{1,2}}
\prod_{\substack{\link{a}{b} \in \alpha \\ a,b \neq 1,2,\alpha(1),\alpha(2)}} |x_{b}-x_{a}|^{-2h_{1,2}} ,
\end{align}
for some constant $C(\kappa) > 0$ depending on $\kappa \in (0,8)$, 
where $\alpha(1)$ and $\alpha(2)$ denote the pairs of $1$ and $2$ in $\alpha$, i.e., 
$\link{1}{\alpha(1)} \in \alpha$ and $\link{2}{\alpha(2)} \in \alpha$, 
and the exponents are
\begin{align*}
h_{1,2} =  \frac{6-\kappa}{2\kappa} = h
 \qquad 
\textnormal{and} \qquad
h_{1,3}  = \frac{8-\kappa}{\kappa} .
\end{align*}
\end{problem}

When $N=2$, using the explicit formula~\eqref{eq: hg formulas for 4p fctions1} for 
$\PartF_{\vcenter{\hbox{\includegraphics[scale=0.2]{figures/link-2.pdf}}}}$, 
one can check by hand that a bound of type~\eqref{eq: bound example} holds true:
\begin{align*}
\PartF_{\vcenter{\hbox{\includegraphics[scale=0.2]{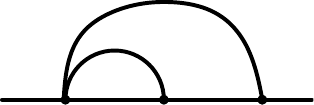}}}}(\xi, x_3, x_4) =  \; & 
\hat{\PartF}_{\vcenter{\hbox{\includegraphics[scale=0.2]{figures/link-2.pdf}}}}(\xi, x_3, x_4) :=  
\lim_{x_1,x_2  \to \xi} \frac{\PartF_{\vcenter{\hbox{\includegraphics[scale=0.2]{figures/link-2.pdf}}}}(x_1 , \ldots , x_{4})}{(x_2 - x_1)^{2/\kappa}}  \\
= \; & 
\frac{\hF\left(\frac{4}{\kappa}, 1-\frac{4}{\kappa}, \frac{8}{\kappa}; 0 \right)}{\hF\left(\frac{4}{\kappa}, 1-\frac{4}{\kappa}, \frac{8}{\kappa}; 1\right)} \;
(x_4-\xi)^{-h_{1,3}} (x_3-\xi)^{-h_{1,3}} (x_4-x_3)^{2/\kappa} ,
\end{align*}
where the prefactor is a constant $C(\kappa)$ depending only on $\kappa$: 
\vspace*{-5mm}
\begin{align} \label{eq: intersting constant}
C(\kappa)
= \frac{\hF\left(\frac{4}{\kappa}, 1-\frac{4}{\kappa}, \frac{8}{\kappa}; 0 \right)}{\hF\left(\frac{4}{\kappa}, 1-\frac{4}{\kappa}, \frac{8}{\kappa}; 1\right)} 
= \frac{1}{\hF\left(\frac{4}{\kappa}, 1-\frac{4}{\kappa}, \frac{8}{\kappa}; 1\right)} 
= \frac{ \Gamma \left( \frac{4}{\kappa} \right) \Gamma \left( \frac{12}{\kappa} - 1 \right)}{\Gamma \left( \frac{8}{\kappa} \right) \Gamma\left( \frac{8}{\kappa} - 1 \right)} 
\in \; & \begin{cases}
(1,\infty), \quad & \kappa \in (0,4) , \\
\{1\}, \quad & \kappa \in \{4\} , \\
(0,1), \quad & \kappa \in (4,8) , \\
\{0\}, \quad & \kappa \in \{8\} .
\end{cases} 
\end{align}

In principle, it should be possible to verify a bound of type~\eqref{eq: bound example} 
for $\kappa \in (0,6]$ using the upper bound from property~\ref{item: positive} and
the explicit construction of $\PartF_\alpha$ discussed in Appendix~\ref{app:Hao}.
This explicit construction should also show that property~\ref{item: baby fusion2} holds for rational values of $\kappa \in (0,6]$.
To verify the PDEs for the limit function, one could use, e.g., 
a result of J.~Dub\'edat~\cite{Dubedat:SLE_and_Virasoro_representations_fusionB},
that we state in Theorem~\ref{thm: Dubedat} in Section~\ref{subsec: fusion dub},
combined with continuity of the function $\PartF_\alpha$ in the parameter $\kappa$.

\begin{problem}
Prove properties~\ref{item: baby fusion2} and~\ref{item: OPE} for $\kappa \in (0,8) \cap \bQ$.
\end{problem}

\begin{quest}
What happens at $\kappa = 8$?
How about when $\kappa > 8$?
\end{quest}
At $\kappa = 8$, the chordal $\SLEk$ describes the scaling limit of the Peano curve between
a uniform spanning tree and its dual tree~\cite{LSW:Conformal_invariance_of_planar_LERW_and_UST}.
Peano curves associated to forests could correspond to multiple $\SLE_8$
processes.  Note that $\SLEk$ type curves with $\kappa \geq 8$ are space-filling, so the situation is drastically different from the range $\kappa \in (0,8)$.
We can also observe this fact from formulas~\eqref{eq: hg formulas for 4p fctions1}--\eqref{eq: hg formulas for 4p fctions2}
for the pure partition functions with $N=2$: both $\PartF_{\vcenter{\hbox{\includegraphics[scale=0.2]{figures/link-2.pdf}}}}$
and $\PartF_{\vcenter{\hbox{\includegraphics[scale=0.2]{figures/link-1.pdf}}}}$ equal zero at $\kappa = 8$,
because their normalization constant, also written explicitly in Equation~\eqref{eq: intersting constant},
tends to zero as $\kappa \to 8$. However, with different normalization, i.e., removing this multiplicative constant,
one obtains a non-zero limit for the renormalized functions
$\PartF_{\vcenter{\hbox{\includegraphics[scale=0.2]{figures/link-2.pdf}}}}$
and $\PartF_{\vcenter{\hbox{\includegraphics[scale=0.2]{figures/link-1.pdf}}}}$ as $\kappa \to 8$.
Finally, let us note that for $\kappa > 8$, the normalization constant~\eqref{eq: intersting constant} becomes negative,
with pole at $\kappa = 12$.

The normalization constant in~\eqref{eq: hg formulas for 4p fctions1}--\eqref{eq: hg formulas for 4p fctions2}
is necessary in order to obtain a clean operator product structure for the multiple SLE pure partition functions $\PartF_\alpha$
(e.g., to establish asymptotics property~\eqref{eq: multiple SLE asymptotics} with no multiplicative constant in front), see Section~\ref{sec:OPE}.
However, from the point of view of multiple $\SLE_\kappa$ processes grown via 
the Loewner evolution~\eqref{eqn::marginalj}, multiplicative constants in the partition functions $\PartF_\alpha$ are irrelevant.

\subsection{\label{subsec:scaling_limimt_results_etc} Relation to Schramm-Loewner evolutions and critical models}

In this section, we briefly illustrate the connection of the partition functions $\PartF$ with,
on the one hand, multiple $\SLE_\kappa$ processes and, on the other hand, 
critical planar lattice models. 
To begin, we discuss the close connection of the pure partition functions $\PartF_\alpha$ to crossing probabilities in critical models.
We give the statement for 
the critical Ising model --- see~\cite{Peltola-Wu:Crossing_probabilities_of_multiple_Ising_interfaces,
Peltola-Wu:Global_and_local_multiple_SLEs_and_connection_probabilities_for_level_lines_of_GFF,
KKP:Conformal_blocks_pure_partition_functions_and_KW_binary_relation} 
for other known results.
We also state convergence results for 
critical Ising interfaces, proved in~\cite{CDHKS:Convergence_of_Ising_interfaces_to_SLE,
Izyurov:Smirnovs_observable_for_free_boundary_conditions_interfaces_and_crossing_probabilities,
BPW:On_the_uniqueness_of_global_multiple_SLEs}.
For other models, analogous statements are expected (and in some cases proven) to hold as well.

\bigskip

Suppose that $\graph^\delta \subset \delta \bZ^2$ approximates a planar simply connected domain  
$\Omega$ as $\delta \searrow 0$ in the Carath\'eodory topology, and  
boundary points $x_1^{\delta}, \ldots, x_{2N}^{\delta}$ of $\graph^\delta$ approximate 
distinct boundary points $x_1, \ldots, x_{2N}$ of $\Omega$ 
(see, e.g.,~\cite{Peltola-Wu:Crossing_probabilities_of_multiple_Ising_interfaces} for the detailed definitions).
Consider the critical Ising model on 
$\graph^{\delta}$ with alternating boundary conditions~\eqref{eq::alternating}. 
Then, each configuration contains $N$ macroscopic interfaces connecting the points
$x_1^{\delta}, \ldots, x_{2N}^{\delta}$ pairwise, illustrated in Figure~\ref{fig: Ising} (right). 
The $\Catalan_N= \frac{1}{N+1} \binom{2N}{N}$ possible planar pairings are labeled by link patterns $\alpha \in \LP_N$.
The basis 
$\{\PartF_\alpha \; | \; \alpha\in \LP_N\}$
of $\Sol_N$ is labeled similarly.

\begin{theorem} \label{thm: summary}
The following hold for the critical Ising model on $(\graph^{\delta};x_1^{\delta}, \ldots, x_{2N}^{\delta})$
with alternating boundary conditions~\eqref{eq::alternating}\textnormal{:}
\begin{enumerate}
\itemcolor{red}
\item \label{item::ising_crossing_proba}
\textnormal{\cite[Theorem~\red{1.1}]{Peltola-Wu:Crossing_probabilities_of_multiple_Ising_interfaces}}
\;
With $\kappa = 3$, we have
\begin{align}\label{eqn::ising_crossing_proba}
\lim_{\delta\to 0} \mathbb{P} [ \textnormal{ the Ising interfaces form the connectivity } \alpha \, ] 
= \frac{\PartF_{\alpha}(\Omega;x_{1},\ldots,x_{2N})}{\PartF_{\textnormal{Ising}}(\Omega;x_{1},\ldots,x_{2N})} , \qquad \textnormal{for all }\alpha\in\LP_N,
\end{align}
where $\{\PartF_{\alpha} \; | \; \alpha \in \LP_N\}$ are the pure partition functions of multiple $\SLE_3$
from Theorem~\ref{thm::purepartition_existence_forallK} and
\begin{align} \label{eqn::ZIsingtotal}
\PartF_{\textnormal{Ising}} (\Omega;x_1, \ldots, x_{2N})
= \sum_{\alpha\in\LP_N}\PartF_{\alpha}  (\Omega;x_1, \ldots, x_{2N}) .
\end{align}
The normalization factor $\PartF_{\textnormal{Ising}}$ also equals the right side of~\eqref{eq: Ising pf general domain},
with~\eqref{eq: 2ptf} plugged in.

\item \label{item::ising_mutltiinterfacealpha}
\textnormal{\cite[Proposition~\red{1.3}]{BPW:On_the_uniqueness_of_global_multiple_SLEs}}
\;
\end{enumerate}
\begin{itemize}
\item[$\bullet$] 
Let $\alpha\in   \LP_N$. Then, as $\delta \to 0$, conditionally on the event that
they form the connectivity $\alpha$,
  the    law     of    the    collection  
of  critical   Ising
  interfaces converges weakly to the (global) $N$-$\SLE_3$ associated to
  $\alpha$, defined in~\textnormal{\cite[Definition~\red{1.1}]{BPW:On_the_uniqueness_of_global_multiple_SLEs}}. 

\item[$\bullet$]
  In particular, as $\delta \to 0$, the law of 
  a single  curve 
  in this  collection connecting two
  points  $x_j$ and  $x_{\alpha(j)}$, where $\link{j}{\alpha(j)} \in \alpha$, 
  converges  weakly to  a conformal  image of  the Loewner  chain given 
  by the SDEs~\eqref{eqn::marginalj} with $\PartF = \PartF_\alpha$ and $\kappa = 3$.

\end{itemize}

\begin{enumerate}
\setcounter{enumi}{2}
\itemcolor{red}
\item \label{item::ising_mutltiinterface} 
\textnormal{\cite[Theorem~\red{3.1}]{Izyurov:Smirnovs_observable_for_free_boundary_conditions_interfaces_and_crossing_probabilities}, 
\cite[Theorem~\red{1.1}]{Izyurov:Critical_Ising_interfaces_in_multiply_connected_domains},
and \cite[Theorem~\red{4.1} \& Proposition~\red{5.1}]{Peltola-Wu:Crossing_probabilities_of_multiple_Ising_interfaces}}
\;

As $\delta \to 0$,  the law of a single  curve 
in the    collection of  critical   Ising  interfaces starting from $x_j$ 
  converges  weakly to  a conformal  image of  the Loewner  chain 
  given by the SDEs~\eqref{eqn::marginalj} with 
  $\PartF = \PartF_{\textnormal{Ising}}$ and $\kappa = 3$.
This curve terminates almost surely at one of the marked points $x_\ell$, where $\ell$ has different parity than $j$.
\end{enumerate}
\end{theorem}

\begin{proofIdea}
The convergence of one critical Ising interface with Dobrushin boundary conditions ($N=1$) was proven
in the celebrated work~\cite{CDHKS:Convergence_of_Ising_interfaces_to_SLE}.
This is established in two steps. 
First, one proves that the sequence $(\gamma^\delta)_{\delta>0}$
of lattice interfaces on $\graph^\delta$ is relatively compact in a certain space of curves.
Thus, one deduces that there exist convergent subsequences as $\delta \to 0$.
For the Ising model, the relative compactness is established using 
topological crossing estimates, see in particular~\cite{Kemppainen-Smirnov:Random_curves_scaling_limits_and_Loewner_evolutions}.
Second, one has to prove that all of the subsequences in fact converge to a unique limit,
identified as the chordal $\SLEk$ with $\kappa=3$. For the identification of the limit, 
Smirnov used a discrete holomorphic martingale observable~\cite{Smirnov:Towards_conformal_invariance_of_2D_lattice_models, Smirnov:Conformal_invariance_in_random_cluster_models1},
that is, a solution to a discrete boundary value problem on $\graph^\delta$, 
converging  as $\delta \to 0$ to the solution of the corresponding boundary value problem on  $\Omega$.
Using the martingale observable, he identified the Loewner driving function
of the scaling limit curve as $\sqrt{3} B_t$.

For multiple curves, the relative compactness follows from the one-curve 
case~\cite{Karrila:Multiple_SLE_local_to_global, Wu:Convergence_of_the_critical_planar_ising_interfaces_to_hypergeometric_SLE}.
For the identification, one can use either a multipoint discrete holomorphic observable,
as for item~\ref{item::ising_mutltiinterface} 
in~\cite{Izyurov:Smirnovs_observable_for_free_boundary_conditions_interfaces_and_crossing_probabilities,
Izyurov:Critical_Ising_interfaces_in_multiply_connected_domains},
or the classification of multiple SLE probability measures, 
as for item~\ref{item::ising_mutltiinterfacealpha} in~\cite{BPW:On_the_uniqueness_of_global_multiple_SLEs}.
See also~\cite{Karrila:Limits_of_conformal_images_and_conformal_images_of_limits_for_planar_random_curves,
Karrila:Multiple_SLE_local_to_global} for discussion on the technical points.

Finally, to prove item~\ref{item::ising_crossing_proba}, we used in~\cite{Peltola-Wu:Crossing_probabilities_of_multiple_Ising_interfaces}
the convergence of the interfaces to multiple $\SLE_3$ processes and a martingale argument:
the ratio $\PartF_\alpha / \PartF_{\textnormal{Ising}}$ defines a bounded martingale for the growth of the curve.
Fine properties of the functions $\PartF_\alpha$ and $\PartF_{\textnormal{Ising}}$
were crucial in the proof. See~\cite{Peltola-Wu:Crossing_probabilities_of_multiple_Ising_interfaces} for details.
\end{proofIdea}

\begin{problem}
Prove results analogous to Theorem~\ref{thm: summary} for other critical lattice models.
\end{problem}

\bigskip

Next, we discuss how the pure partition functions are related to the theory of multiple SLEs.
For background, we refer to~\cite{Dubedat:Commutation_relations_for_SLE,
Peltola-Wu:Global_and_local_multiple_SLEs_and_connection_probabilities_for_level_lines_of_GFF},
and references therein.
In general, a curve in a local multiple $N$-$\SLEk$ (sampled from its marginal law) 
has the Loewner chain description~\eqref{eqn::marginalj}
with some partition function $\PartF$. The word ``local'' refers to the fact that a priori, the Loewner chain is only 
defined up to a blow-up time. Choosing $\PartF = \PartF_\alpha$ in~\eqref{eqn::marginalj} 
gives a process where the curves 
growing from the marked boundary points $x_1, \ldots, x_{2N}$ should connect together according to the pairing $\alpha$.
This was indeed proven for $\kappa \in (0,4]$ in~\cite[Proposition~\red{4.9}]{Peltola-Wu:Global_and_local_multiple_SLEs_and_connection_probabilities_for_level_lines_of_GFF},
see also Proposition~\ref{prop:cor} below.
These multiple $\SLEk$ 
processes are extremal, or pure, in the sense that they generate a convex set of 
probability measures for multiple SLEs:

\begin{prop} 
\label{prop:extr}
\textnormal{\cite[Corollary~\red{1.2}, extended]{Peltola-Wu:Global_and_local_multiple_SLEs_and_connection_probabilities_for_level_lines_of_GFF}} 
\;
Let $\kappa \in (0,6]$. 
\begin{enumerate}
\itemcolor{red}
\item
For any $\alpha \in \LP_N$, the pure partition function $\PartF_\alpha$
defines a local $N$-$\SLE_\kappa$ process via the SDEs~\eqref{eqn::marginalj}. 

\item
For any $N\geq 1$, the convex hull of the local $N$-$\SLE_\kappa$ probability measures 
corresponding to $\{\PartF_\alpha \; | \; \alpha\in \LP_N \}$
has dimension $\Catalan_N-1$.
The $\Catalan_N$ local $N$-$\SLE_\kappa$ 
with pure partition functions $\PartF_\alpha$
are the extremal points of this convex set. 
\end{enumerate}
\end{prop}
\begin{proofIdea}
For $\kappa \in (0,4]$, this statement appears as~\cite[Corollary~\red{1.2}]{Peltola-Wu:Global_and_local_multiple_SLEs_and_connection_probabilities_for_level_lines_of_GFF}.
Its proof works also for $\kappa \in (4,6]$. 
The main idea is to use the classification of local $N$-$\SLE_\kappa$ probability measures in terms of their partition 
functions, proven  
in J.~Dub\'edat's work~\cite{Dubedat:Commutation_relations_for_SLE}
(see also~\cite[Theorem~\red{A.4}]{Kytola-Peltola:Pure_partition_functions_of_multiple_SLEs} and~\cite[Proposition~\red{4.7}]{Peltola-Wu:Global_and_local_multiple_SLEs_and_connection_probabilities_for_level_lines_of_GFF}),
and linear independence of the functions $\PartF_\alpha$,
stated in item~\ref{item: linearly independent} of Theorem~\ref{thm::purepartition_existence_forallK}.
A crucial technical point for the proof is that the partition functions $\PartF_\alpha$ must be positive,
which for $\kappa \in (0,6]$ is guaranteed by item~\ref{item: positive} of Theorem~\ref{thm::purepartition_existence_forallK}. 
Lack of positivity is the only obstacle for extending this result to the remaining range 
$\kappa \in (6,8)$. This
could perhaps be established via the probabilistic construction
presented in Appendix~\ref{app:Hao}.
\end{proofIdea}

The extremal (pure) $\SLE_\kappa$ processes associated to $\PartF_\alpha$ are known to be well-defined  also globally, i.e.,
up to and including the terminal time of the curve (see Proposition~\ref{prop:cor}).
On the other hand, choosing $\PartF = \sum_{\alpha \in \LP_N} \PartF_\alpha$ in~\eqref{eqn::marginalj} 
generates a local multiple $N$-$\SLEk$ for which all planar connectivities of the curves are possible. 
However, it has not been proven in general that such a process is well-defined  up to and including its 
terminal time ---
the case of $\kappa=3$ was treated in~\cite{Peltola-Wu:Crossing_probabilities_of_multiple_Ising_interfaces}
via SLE techniques, and 
``local-to-global'' multiple SLEs arising from scaling limits of interfaces in critical lattice models
were considered in~\cite{Karrila:Multiple_SLE_local_to_global}. 
The main difficulty to generalize Proposition~\ref{prop:cor} for the sum function $\PartF = \sum_{\alpha \in \LP_N} \PartF_\alpha$
is the lack of a bound of type~\eqref{eqn::partitionfunction_positive}. 

\begin{prop} \label{prop:cor}
\textnormal{\cite[Proposition~\red{4.9}, extended]{Peltola-Wu:Global_and_local_multiple_SLEs_and_connection_probabilities_for_level_lines_of_GFF}} 
\;
Let $\kappa\in (0,6]$. Let $\alpha\in\LP_N$ and suppose that $\link{a}{b}\in\alpha$. 
Let $W_t$ be the solution to the SDEs~\eqref{eqn::marginalj} with $j = a$
and $\PartF = \PartF_\alpha$,
and let 
\begin{align*}
\swaltime := \min_{i \neq a} \; \sup\Big\{ t > 0 \; | \; \inf_{ s \in[0,t]} |g_s(x_i)-W_s| > 0 \Big\}
\end{align*} 
be the first swallowing time of one of the points $\{x_1,\ldots, x_{2N}\} \setminus \{x_a\}$.
Then, the Loewner chain driven by $W_t$ is well-defined  up to the swallowing time $\swaltime$.
Moreover, it is almost surely generated by a continuous curve up to and including $\swaltime$. 
\end{prop}

\begin{proofIdea}
For $\kappa \in (0,4]$, this statement appears as the first part of~\cite[Proposition~\red{4.9}]{Peltola-Wu:Global_and_local_multiple_SLEs_and_connection_probabilities_for_level_lines_of_GFF}.
The same proof works also for $\kappa \in (4,6]$.
The key point is that up to and including $\swaltime$, the Loewner chain with driving function $W_t$
is absolutely continuous with respect to the chordal $\SLEk$ curve from $x_a$ to $x_b$,
thanks to the strong upper bound~\eqref{eqn::partitionfunction_positive} 
in item~\ref{item: positive} in Theorem~\ref{thm::purepartition_existence_forallK}.
Again, lack of this bound prohibits extending this result to the remaining range 
$\kappa \in (6,8)$. 
\end{proofIdea}

In principle, the process in Proposition~\ref{prop:cor} could also be continued after the time $\swaltime$,
and we expect that eventually
it gives rise to a continuous transient curve that terminates at the point $x_b$.
For $\kappa \in (0,4]$, this is indeed the case, because $\swaltime$ is equal to the hitting time of $x_b$ ---
namely, the simple chordal $\SLEk$ only hits the boundary of the domain at its endpoints~\cite{Rohde-Schramm:Basic_properties_of_SLE}.
However, for $\kappa > 4$, the chordal $\SLEk$ almost surely hits the boundary elsewhere as well,
so further care is needed.

\begin{figure}
\centering
\includegraphics[width=.75\textwidth]{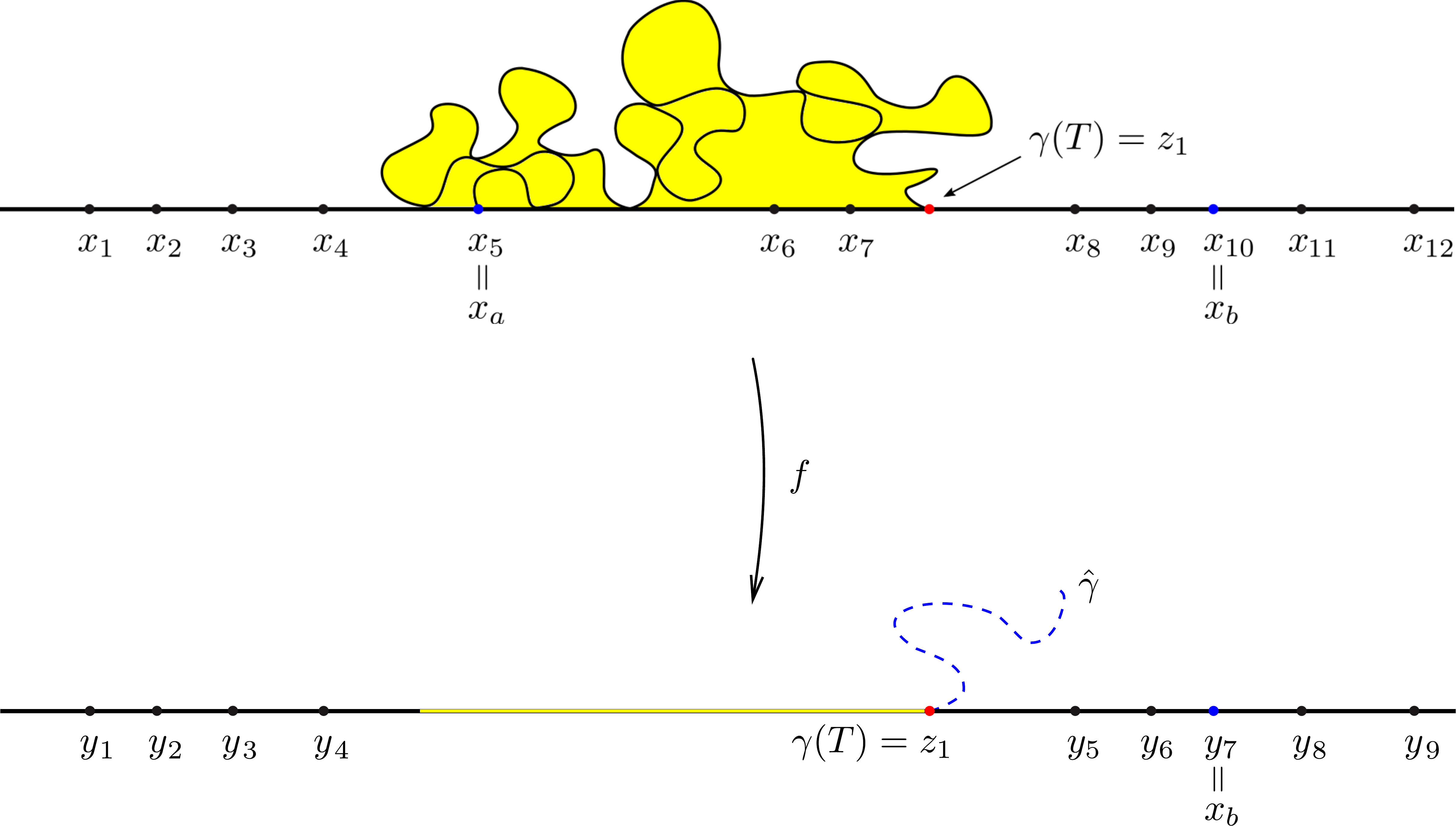}
\caption{\label{fig: iterate}
Illustration of the Loewner chain with partition function $\PartF_\alpha$ when $\kappa > 4$.
}
\end{figure}

Let us briefly sketch how the case of $\kappa > 4$ could be treated; see also Figure~\ref{fig: iterate}.
First, properties of the pure partition function $\PartF_\alpha$, discussed in more detail
in Appendix~\ref{app:Hao} and in~\cite[Section~\red{6}]{Wu:Convergence_of_the_critical_planar_ising_interfaces_to_hypergeometric_SLE},
should guarantee that almost surely, at time $\swaltime$
the Loewner chain $\gamma = (\gamma(t))_{t \leq \swaltime}$ associated to $(W_t)_{t \leq \swaltime}$ does not disconnect 
any two points $x_c$, $x_d$ that correspond to endpoints of a link $\link{c}{d} \in \alpha$ from each other.
In particular, after the swallowing time $\swaltime =: \swaltime_1$, we expect that
the Loewner chain may be continued in the connected component $\hat{\Omega}$ of $\bH \setminus \gamma$ containing $x_b$ on its boundary
as follows. Let 
$\Mob \colon \hat{\Omega} \to \bH$ be a conformal map fixing $\gamma(\swaltime_1) =:z_1 \in \partial \bH$ and $x_b$.
Let $\{y_1, \ldots, y_\ell\}$ be the conformal images under $\Mob$ of those points in $\{x_1,\ldots, x_{2N}\} \setminus \{x_a\}$
that belong to $\partial \hat{\Omega}$ (note that these include $x_b$, and $\ell$ is odd).
Also, let $\hat{\alpha}_0$ be the sub-link pattern of $\alpha$ associated to the points $\{\Mob^{-1}(y_1), \ldots, \Mob^{-1}(y_\ell) \} \cup \{x_a\}$,
let $\hat{\alpha}$ be the link pattern obtained from $\hat{\alpha}_0$ by replacing the endpoint corresponding to $x_a$ with $z_1$
(this possibly results in a cyclic permutation of the endpoints).
Finally, let $\PartF_{\hat{\alpha}}$ be the pure partition function associated to this link pattern
with variables  $\{y_1, \ldots, y_\ell \} \cup \{z_1\}$.
Thus, effectively, $x_a$ gets replaced by the new starting point $z_1$ on $\partial \bH$.

Next, define another Loewner chain $\hat{\gamma}$ driven by $\PartF_{\hat{\alpha}}$ starting from $z_1$ 
(and targeted to $x_b$, which is the pair of $z_1$ determined by the link pattern $\hat{\alpha}$)
by the SDEs~\eqref{eqn::marginalj} with $x_j = z_1$ and $\PartF = \PartF_{\hat{\alpha}}$.
By Proposition~\ref{prop:cor},
this Loewner chain is almost surely generated by a continuous curve up to and including the first swallowing time $\swaltime_2$
of one of the points $\{y_1, \ldots, y_\ell \}$.
Iterating this construction produces a finite sequence of Loewner chains $\gamma_1 = \gamma$, $\gamma_2 = \hat{\gamma}$, $\ldots$,
each of which is almost surely generated by a continuous curve up to and including the first swallowing time ($\swaltime_1, \swaltime_2, \ldots$)
of one of the marked points (excluding its starting point). 
The last Loewner chain $\gamma_m$ in this sequence, defined up to the stopping time $\swaltime_m$,
has the law of the chordal $\SLEk$ from some random point $\gamma_{m-1}(\swaltime_{m-1})$ to $x_b$,
because there are no other marked points left in the same connected component.
Now, we expect that the concatenation of these Loewner chains 
defines a continuous curve from $x_a$ to $x_b$,
regarded as the Loewner chain associated to the original SDE with $\PartF_\alpha$ from $x_a$ to $x_b$.
(See also~\cite{Karrila:Multiple_SLE_local_to_global} for curves arising as scaling limits of lattice interfaces.)

\section{\label{sec:OPE}Fusion and operator product expansion}

In this section, we discuss a fusion hierarchy for the multiple SLE partition functions.
Such ideas are important in quantum field theory, where the  fields are supposed to form an algebra
with multiplication given by their ``operator product expansion'' 
(OPE)~\cite{Wilson:Non_Lagrangian_models_of_current_algebra,
BPZ:Infinite_conformal_symmetry_in_2D_QFT}. 
This postulate results in a hierarchy of correlation functions appearing in each others' Frobenius series.
For the multiple SLE partition functions, there is a particularly nice combinatorial 
structure~\cite{Dubedat:SLE_and_Virasoro_representations_fusionB,
Peltola:Basis_for_solutions_of_BSA_PDEs_with_particular_asymptotic_properties}.

\bigskip

To motivate the results of this section, in Section~\ref{subsec:FusionCFT} we explain features of 
the operator algebra and OPE postulates in conformal field theory, which 
also lead to the so-called conformal bootstrap hypothesis (to be discussed in Section~\ref{subsec: ASPWC}): 
given certain data, the associated CFT can be completely solved.
 Like Sections~\ref{subsec:CFT} and~\ref{subsec:mgles},
this preliminary section is not intended to be mathematically precise.
In contrast, in Section~\ref{subsec: fusion dub} we state a rigorous result, Theorem~\ref{thm: Dubedat}
proved in J.~Dub\'edat's work~\cite{Dubedat:SLE_and_Virasoro_representations_fusionB},
towards understanding of how the operator algebra structure
can be formulated for the CFT correlation functions $\big\langle \Phi_{1,2}(x_1)  \cdots \Phi_{1,2}(x_{2N}) \big\rangle$
corresponding to the multiple SLE partition functions $\PartF(x_1, \ldots, x_{2N})$ from Section~\ref{sec: Multiple SLE partition functions}.
The OPE multiplication rules (fusion rules) for these specific correlation functions 
were found early in the CFT literature~\cite[Sections~\red{5},\red{6}]{BPZ:Infinite_conformal_symmetry_in_2D_QFT}.
In particular, we will see that  limits of solutions $\PartF$ of the second order PDEs~\eqref{eq: multiple SLE PDEs}
give rise to solutions of higher order PDEs 
(recall items~\ref{item: baby fusion2}--\ref{item: OPE} of Theorem~\ref{thm::purepartition_existence_forallK} as well).
Furthermore, the form of this fusion hierarchy can be made very explicit:
in Sections~\ref{subsec: fusion SCCG}--\ref{subsec:OPE for ppf}, 
we briefly discuss a systematic approach 
from the work~\cite{Kytola-Peltola:Conformally_covariant_boundary_correlation_functions_with_quantum_group,
Peltola:Basis_for_solutions_of_BSA_PDEs_with_particular_asymptotic_properties} of the author with K.~Kyt\"ol\"a, 
establishing rather general results.
We state the most important findings in Theorem~\ref{thm: big prop} and Proposition~\ref{prop: OPE}.

\subsection{\label{subsec:FusionCFT}Fusion and operator product expansion in conformal field theory}

In CFT, it is postulated that the conformal fields are operators that constitute an  
algebra with associative product, the operator product expansion, 
OPE~\cite{Wilson:Non_Lagrangian_models_of_current_algebra, 
BPZ:Infinite_conformal_symmetry_in_2D_QFT}. 
In some cases, this algebra and its OPE structure obtain a mathematically clean 
formulation using vertex operator algebras 
--- see, e.g.,~\cite{FHL:On_axiomatic_approaches_to_vertex_operator_algebras_and_modules,
Zhu:Modular_invariance_of_characters_of_vertex_operator_algebras,
Kac:Vertex_algebras_for_beginners,
Schottenloher:Mathematical_introduction_to_CFT}, and references therein.
For general background on OPEs in conformal field theory, the reader may consult,
e.g., the books~\cite{DMS:CFT, Schottenloher:Mathematical_introduction_to_CFT, Mussardo:Statistical_field_theory}.

In the physics literature, 
the formal ``operator product'' of two fields 
$\Phi_{\index_1}(z_1)$ and $\Phi_{\index_2}(z_2)$ is often written  in the form
\begin{align} \label{eq:OPE general}
\textnormal{``}
\Phi_{\index_1}(z_1) \Phi_{\index_2}(z_2) \; \sim \;
\sum_{\index} \frac{C_{\index_1, \index_2}^{\index}}{(z_1-z_2)^{\Delta_{\index_1} + \Delta_{\index_2} - \Delta_\index}} 
\, \Phi_{\index}(z_2) \textnormal{''} ,
\qquad \textnormal{as } 
|z_1 - z_2| \to 0,
\end{align}
where $\Phi_{\index}$ are (scalar) primary fields with conformal weights $\Delta_\index \in \bR$,  
and $C_{\index_1, \index_2}^{\index} \in \bC$ are called structure constants.
(We again omit the anti-holomorphic sector.)
More generally, one could write the right-hand side of~\eqref{eq:OPE general} in the form 
$\sum_{\index} C_{\index_1, \index_2}^{\index} (z_1,z_2) \Phi_{\index}(z_2)$,
for some functions $C_{\index_1, \index_2}^{\index} (z_1,z_2)$ allowing, e.g., logarithmic terms in the expansion.
Physicists speak of ``fusion rules'' that tell which fields 
$\Phi_{\index}$ are present in the OPE product~\eqref{eq:OPE general} of 
$\Phi_{\index_1}$ and $\Phi_{\index_2}$, i.e., which of the structure constants $C_{\index_1, \index_2}^{\index}$ are non-zero.

Morally, Equation~\eqref{eq:OPE general} should be understood ``inside correlations'', that is,
as an asymptotic statement relating correlation functions of type
$\big\langle \Phi_{\index_1}(z_1) \Phi_{\index_2}(z_2) \cdots \big\rangle$
to those of type $\big\langle \Phi_{\index}(z_2) \cdots \big\rangle$ when $|z_2-z_1|\to0$.
In Sections~\ref{subsec:OPE for ppf} and~\ref{subsec: ASPWC}, we shall give 
mathematically precise statements of this sort.

\bigskip

Fusion rules from the physics literature can be used to motivate the choice of 
asymptotic boundary conditions 
in order to single out specific solutions to the PDEs satisfied by correlation functions of fields with degeneracies 
(recall Section~\ref{subsec:CFT} and Appendix~\ref{app:Vir}).
To explicate this, we would like to identify the functions $\PartF \in \Sol_N$, 
discussed in Section~\ref{sec: Multiple SLE partition functions},
with correlation functions of type 
$\big\langle \Phi_{1,2}(z_1)  \cdots \Phi_{1,2}(z_{2N}) \big\rangle$.
We recall that these functions are solutions to the PDEs~\eqref{eq: multiple SLE PDEs},
and a basis for the space $\Sol_N$ can be found by imposing asymptotics properties~\eqref{eq: multiple SLE asymptotics}. 
According to~\cite[Section~\red{6}]{BPZ:Infinite_conformal_symmetry_in_2D_QFT}, the relevant fusion structure looks like
\begin{align} \label{eq: fusion rules spec}
\textnormal{``}
\Phi_{1,2}(z_1) \Phi_{1,2}(z_2) \; \sim \; 
\frac{C_{2,2}^1}{ (z_1-z_2)^{2 h_{1,2} - h_{1,1}}}  \, \Phi_{1,1}(z_2) 
\; + \; 
\frac{C_{2,2}^3}{ (z_1-z_2)^{2 h_{1,2} -h_{1,3} }} \, \Phi_{1,3}(z_2) \textnormal{''}  ,
\qquad \textnormal{as } 
|z_1 - z_2| \to 0 ,
\end{align}
where $C_{2,2}^1$ and $C_{2,2}^3$ are the structure constants,  
and in terms of the parameter $\kappa > 0$, the conformal weights read
\begin{align*}
h_{1,3}  = \frac{8-\kappa}{\kappa} , \qquad 
h_{1,2} = \frac{6-\kappa}{2\kappa} = h ,
\qquad \textnormal{and} \qquad 
h_{1,1} = 0 .
\end{align*}
More generally, for the fields $\Phi_{1,s}$ discussed in Section~\ref{subsec:CFT} and Appendix~\ref{app:Vir},  
according to~\cite[Section~\red{6}]{BPZ:Infinite_conformal_symmetry_in_2D_QFT},
we expect that
\begin{align} 
\label{eq: fusion rules General 1}
\textnormal{``}
\Phi_{1,2}(z_1) \Phi_{1,s}(z_2) \; \sim \; \; &
\frac{C_{2,s}^{s-1}}{ (z_1-z_2)^{h_{1,2} + h_{1,s} - h_{1,s-1}}} \, \Phi_{1,s-1}(z_2) \; + \; 
\frac{C_{2,s}^{s+1}}{ (z_1-z_2)^{h_{1,2} + h_{1,s} - h_{1,s+1}}} \, \Phi_{1,s+1}(z_2) \textnormal{''}  ,
\\
\label{eq: fusion rules General 2}
\textnormal{``}
\Phi_{1,s_1}(z_1) \Phi_{1,s_2}(z_2) \; \sim \;  \; &
\sum_{\projdmn \in S_{1,2}}
\frac{C_{s_1, s_2}^{\projdmn}}{(z_1 - z_2)^{h_{1,s_1} + h_{1,s_2} - h_{1,\projdmn}}} \,
\Phi_{1,\projdmn}(z_2) \textnormal{''} ,
\qquad \textnormal{as } 
|z_1 - z_2| \to 0 
\end{align}
(with the convention that $\Phi_{1,0}=0$), where
\begin{align} \label{eq: Kac weights BSA}
h_{1,s} = \frac{(s-1)(2(s+1)-\kappa)}{2\kappa} , \qquad \textnormal{for } s \in \bZpos 
\end{align}
are the Kac conformal weights~\eqref{eq: Kac weights GEN} parameterized in terms of $\kappa$,
and the index set is
\begin{align*}
S_{1,2} := \big\{ |s_{j+1} - s_{j}| + 1 , \,  |s_{j+1} - s_{j}| + 3 ,\,  \ldots , \, s_{j}+s_{j+1} - 3 , \, s_{j}+s_{j+1} - 1 \big\}  .
\end{align*}
In Theorems~\ref{thm: Dubedat} and~\ref{thm: big prop}, and Proposition~\ref{prop: OPE},
we will see how the fusion rules~\eqref{eq: fusion rules spec}--\eqref{eq: fusion rules General 2} can be phrased mathematically precisely.

\subsection{\label{subsec: fusion dub} Fusion: analytic approach}

In this section, we consider systems of  PDEs written in terms of the first order differential operators
\begin{align*}
\sL_{-k}^{(j)} =
    \sum_{ \substack{ 1 \leq i \leq n \\ i \neq j} } 
    \bigg( \frac{(k-1) h_{1,s_i}}{(z_i-z_j)^{k}} - \frac{1}{(z_i-z_j)^{k-1}} \pder{z_i} \bigg) 
     , \qquad \textnormal{for } k \in \bZpos ,
\end{align*}
labeling the conformal weights $h_{1,s_i}$ parameterized by $\kappa > 0$ as in~\eqref{eq: Kac weights BSA} 
by $\multidim = (s_1, \ldots, s_n) \in \bZpos^n$.
The PDE system of interest is
\begin{align}\label{eq: BSA differential equations}
\left[
\sum_{k=1}^{s_j}\sum_{\substack{n_{1},\ldots,n_{k}\geq1\\
n_{1}+\ldots+n_{k}=s_j}}
\frac{(-4/\kappa)^{s_j-k}\,(s_j-1)!^{2}}{\prod_{l=1}^{k-1}(\sum_{i=1}^{l}n_{i})(\sum_{i=l+1}^{k}n_{i})}\times\sL_{-n_{1}}^{(j)}\cdots
\sL_{-n_{k}}^{(j)} \right]
F(z_1,\ldots,z_n) = 0 , \qquad \textnormal{for all } j \in \{1,\ldots,n\} ,
\end{align}
for functions $F \colon \extendedChamber_n \to \bC$.
We recall from Section~\ref{subsec:CFT} that this type of PDEs are expected to hold
for correlation functions of the  conformal fields $\Phi_{1,s_j}$ with degeneracies at levels $s_j$.
In fact, the above PDEs are a special case of the ones appearing
in~\cite{BPZ:Infinite_conformal_symmetry_in_2D_QFT},
with explicit formulas~\eqref{eq: BSA differential equations} found 
by L.~Benoit and Y.~Saint-Aubin 
 in~\cite{BSA:Degenerate_CFTs_and_explicit_expressions_for_some_null_vectors}.
In general, the PDEs in~\cite{BPZ:Infinite_conformal_symmetry_in_2D_QFT}
also include conformal weights of type $h_{r,s}$ inside the Kac table 
(see~\eqref{eq: Kac weights GEN} in Appendix~\ref{app:Vir}).
We will only consider the case of~\eqref{eq: BSA differential equations}
in this article.
Appendix~\ref{app:Vir} contains examples of PDEs of type~\eqref{eq: BSA differential equations}
of orders one and two.
In particular, for translation-invariant functions, the second order PDE system where we take $\multidim = (2, 2, \ldots, 2)$ 
in~\eqref{eq: BSA differential equations}
is equivalent to PDE system~\eqref{eq: multiple SLE PDEs}, whose solution space we analyzed in Section~\ref{sec: Multiple SLE partition functions}.
The topic of this section is to consider solutions to higher order PDEs of type~\eqref{eq: BSA differential equations}.

\bigskip

J.~Dub\'edat proved in~\cite{Dubedat:SLE_and_Virasoro_representations_fusionB} 
that solutions of the second order PDEs~\eqref{eq: multiple SLE PDEs}
can be used to produce solutions to higher order PDEs of type~\eqref{eq: BSA differential equations}.
This pertains to a mathematical  formulation for the fusion structure discussed in Section~\ref{subsec:FusionCFT}.

\begin{theorem} \label{thm: Dubedat}
\textnormal{\cite[Lemma~\red{14} \& Theorem~\red{15}, simplified]{Dubedat:SLE_and_Virasoro_representations_fusionB};
see also~\cite[Lemma~\red{5.6}]{KKP:Conformal_blocks_pure_partition_functions_and_KW_binary_relation}} 
\;
Let $\kappa \in (0,8) \setminus \bQ$.
\begin{enumerate}
\itemcolor{red}
\item \label{item:Dub1}
Let $\PartF \colon \chamber_{2N} \to \bC$ be a solution of the second order
PDE system~\eqref{eq: multiple SLE PDEs}. 
Suppose that, for $\epsilon > 0$ small enough, we have
\begin{align} \label{eq: detailed boundary bound h12}
\PartF(x_1, x_2, x_3, \ldots, x_{2N}) = \OO \Big( (x_2 - x_1)^{h_{1,3} - 2h_{1,2} - \epsilon} \Big) , \qquad \textnormal{as } x_2 \searrow \, x_1.
\end{align}
Then, the limit
\begin{align*}
\hat{\PartF}(x_1, x_3, \ldots, x_{2N}) 
:= \lim_{x_2 \searrow \, x_1} 
\frac{\PartF(x_1, x_2, x_3, \ldots, x_{2N})}{(x_2 - x_1)^{h_{1,3} - 2h_{1,2}}}
\end{align*}
exists and defines a solution to the following system of $2N-1$ PDEs:
\begin{align} 
\begin{split}
\label{eq: PDEs after one fusion}
\left[ \pdder{x_j} - \frac{4}{\kappa} \sL^{(j)}_{-2} \right] 
\hat{\PartF}(x_1, x_3, \ldots, x_{2N}) = \; & 0 , \qquad \textnormal{for all } j \in \{3,\ldots,2N\} , \\
\left[ \pddder{x_1}
    - \frac{16}{\kappa} \sL^{(1)}_{-2} \pder{x_1}
    + \frac{8(8-\kappa)}{\kappa^2} \sL^{(1)}_{-3} \right] 
\hat{\PartF}(x_1, x_3, \ldots, x_{2N}) = \; & 0 ,
\end{split}
\end{align}
where
\begin{align*} 
\sL_{-2}^{(j)} = \; &
    \sum_{\substack{3 \leq i \leq 2N \\ i \neq j }}
       \bigg( \frac{h_{1,2}}{(x_i-x_j)^{2}} - \frac{1}{x_i-x_j} \pder{x_i} \bigg)
    + \bigg( \frac{h_{1,3}}{(x_1-x_j)^{2}} - \frac{1}{x_1-x_j} \pder{x_1} \bigg)
    , \qquad \textnormal{for all } j \in \{3,\ldots,2N\} , \\
\sL^{(1)}_{-3} = \; &
    \sum_{3 \leq i \leq 2N} \bigg( \frac{2h_{1,2}}{(x_i-x_1)^{3}} - \frac{1}{(x_i-x_1)^{2}} \pder{x_i} \bigg) .
\end{align*}

\item \label{item:Dub2}
Let $\PartF \colon \chamber_n \to \bC$ be a solution of 
the PDE system of type~\eqref{eq: BSA differential equations}
with $\multidim = (s, 2, s_3, \ldots, s_n)$. 
Suppose that, for $\epsilon > 0$ small enough, we have
\begin{align*}
\PartF(x_1, x_2, x_3, \ldots, x_n) = \OO \Big( (x_2 - x_1)^{h_{1,s+1} - h_{1,s} - h_{1,2} - \epsilon} \Big) , \qquad \textnormal{as } x_2 \searrow \, x_1.
\end{align*}
Then, the limit
\begin{align} \label{eq: detailed boundary bound hgen}
\hat{\PartF}(x_1, x_3, \ldots, x_n) 
:= \lim_{x_2 \searrow \, x_1} 
\frac{\PartF(x_1, x_2, x_3, \ldots, x_n)}{(x_2 - x_1)^{ h_{1,s+1} - h_{1,s} - h_{1,2}}}
\end{align}
exists and defines a solution to the system of $n-1$ PDEs comprising~\eqref{eq: BSA differential equations} 
with $\multidim = (s+1, s_3, \ldots, s_n)$ in the $n-1$ variables
$(z_1, z_2, \ldots, z_{n-1}) = (x_1, x_3, \ldots, x_n)$.
\end{enumerate}
\end{theorem}

We invite the reader to compare item~\ref{item:Dub1} of Theorem~\ref{thm: Dubedat} 
with Equation~\eqref{eq: fusion rules spec}, 
and item~\ref{item:Dub2}
with Equation~\eqref{eq: fusion rules General 1}.

\begin{remark}
When $\kappa > 0$, we have
\begin{align*}
h_{1,3} - 2h_{1,2} =  \frac{2}{\kappa} \quad > \quad \frac{\kappa-6}{\kappa} = -2h = h_{1,1} - 2h_{1,2} 
\qquad \qquad \textnormal{if and only if} \qquad \qquad \kappa \in (0,8).
\end{align*}
Therefore, the power $-2h$ always gives the leading asymptotics when $\kappa \in (0,8)$.
Item~\ref{item:Dub1} of Theorem~\ref{thm: Dubedat} concerns the subleading asymptotics, with power $h_{1,3} - 2h_{1,2} = 2/\kappa$.
Similarly, by~\eqref{eq: Kac weights BSA},  
we have $h_{1,s+1} - h_{1,s} - h_{1,2} \; > \; h_{1,s-1} - h_{1,s} - h_{1,2}$ when $\kappa \in (0,8)$ and $s \geq 2$,
so item~\ref{item:Dub2} of Theorem~\ref{thm: Dubedat} also concerns the subleading asymptotics
(when $s=1$, there is only the trivial asymptotics).
\end{remark}

The reason for the condition $\kappa \notin \bQ$ in Theorem~\ref{thm: Dubedat} 
is representation theoretic. 
Furthermore, when $\kappa$ is rational, solutions of PDE system~\eqref{eq: BSA differential equations}
could have Frobenius series with logarithmic terms, see~\cite[Theorem~\red{2}]{Flores-Kleban:Solution_space_for_system_of_null-state_PDE4}. 
In certain applications, the condition $\kappa \notin \bQ$
can be removed by a separate argument, as discussed in~\cite{Dubedat:SLE_and_Virasoro_representations_fusionB}.

In~\cite{Dubedat:SLE_and_Virasoro_representations_fusionB}, Dub\'edat only studied the case when the first two variables of $\PartF$ tend to each other.
One could iterate this to find PDEs of higher order in the other variables as well, producing solutions to general systems of type~\eqref{eq: BSA differential equations}.
However, it is not immediately clear 
whether such iterated limits depend on the order in which the limits are taken. 
Next, we discuss a  systematic method for the fusion procedure,
in which, e.g., iterated limits can be taken easily.

\subsection{\label{subsec: fusion SCCG} Fusion: systematic algebraic approach}

In this section, we consider a general collection of
functions $\PartF_\linkpatt$ that solve PDE systems of type~\eqref{eq: BSA differential equations},
when $\kappa \in (0,8) \setminus \bQ$.
They are indexed by planar \emph{valenced} link patterns $\linkpatt$,
defined in detail in~\cite[Section~\red{2}]{Peltola:Basis_for_solutions_of_BSA_PDEs_with_particular_asymptotic_properties}.
The valenced  link patterns
generalize the usual link patterns (planar pair partitions) $\alpha$, appearing in Figure~\ref{fig: nonvallp} in Section~\ref{subsec:ppfdef}. 
Roughly, a valenced link pattern is 
a collection of $\ell \in \bZnn$ links $\link{a}{b}$ in the upper half-plane, 
with endpoints $a_1,\ldots,a_\ell,b_1,\ldots,b_\ell$ on the real axis,
\begin{align*} 
\linkpatt = \{\link{a_1}{b_1},\ldots,\link{a_\ell}{b_\ell} \} ,
\end{align*} 
where for each link $\link{a}{b} \in \linkpatt$, the two endpoints $a$ and $b$ 
are distinct, $a \neq b$.
The links in $\linkpatt$ are counted by multiplicity, so $\linkpatt$ is a multiset.
We denote by $\ell_{a,b}(\linkpatt)$ the multiplicity of the link $\link{a}{b}$ in $\linkpatt$. 
See Figure~\ref{fig: vallp} for an illustration.

We denote by $\LP_\multii$ the collection of valenced link patterns $\linkpatt$ 
with given valences $\multii = (\sIndex_1, \ldots, \sIndex_n) \in \bZpos^n$, i.e., 
we have $\linkpatt \in \LP_\multii$ if and only if, for each $j \in \{1,\ldots,n\}$, 
the total number of lines in $\linkpatt$ attached to the $j$:th endpoint counted from the left equals $\sIndex_j$.
When $\multii = (1,1,\ldots,1)$ has $2N$ ones in it, we just have $\LP_\multii = \LP_N$
in our earlier notation (denoted by $\mathrm{PP}_N$ in~\cite{Peltola:Basis_for_solutions_of_BSA_PDEs_with_particular_asymptotic_properties})
--- see Figure~\ref{fig: nonvallp}.
In general, because all links in $\linkpatt$ must have a distinct pair of endpoints, 
we necessarily have
\begin{align*}
|\multii| := \sIndex_1 + \cdots + \sIndex_n \, \in \, 2 \bZnn .
\end{align*}
The parameters in the PDEs in~\eqref{eq: BSA differential equations}
are labeled by $\multidim = (s_1, \ldots, s_n)$. 
For solutions $\PartF_\linkpatt$ to~\eqref{eq: BSA differential equations}, we have
$\linkpatt \in \LP_\multii$ with 
\begin{align*}
\multii  = \multidim - 1 , \textnormal{i.e., } 
\sIndex_j = s_j - 1, \quad \textnormal{for all } j \in \{1,\ldots,n\} .
\end{align*}

\begin{figure}[h!]
\centering
\includegraphics[scale=1]{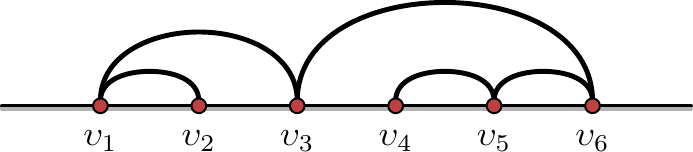}
\caption{\label{fig: vallp}
Graphical illustration of a valenced link pattern $\linkpatt \in \LP_\multii$ with valences 
$\multii = (\sIndex_1, \ldots, \sIndex_6) = (2,1,2,1,2,2)$, and $|\multii|=10$.
The function $\PartF_\linkpatt(x_1, \ldots, x_6)$ is a solution to PDE system~\eqref{eq: BSA differential equations}
with  $\multidim = \multii + 1 = (3,2,3,2,3,3)$. 
It could be thought of as a CFT correlation function of type 
$\big\langle \Phi_{1,3}(x_1) \Phi_{1,2}(x_2) \Phi_{1,3}(x_3) \Phi_{1,2}(x_4) \Phi_{1,3}(x_5) \Phi_{1,3}(x_6) \big\rangle$,
labeled by $\multidim$.}
\end{figure}

\bigskip

Next, we define a map which associates to each valenced link pattern $\linkpatt \in \LP_\multii$
a usual link pattern $\alpha = \alpha(\linkpatt) \in \LP_N$: 
\begin{align}\label{eq: link pattern to pair partition}
\LP_\multii \to \LP_N, \qquad \linkpatt \mapsto \alpha(\linkpatt) \, \in \, \LP_{(1,1,\ldots,1,1)} = \LP_N ,
\end{align}
such that 
$N = \frac{1}{2} |\multii| \in \bZnn$.
This map is defined as follows: in $\linkpatt$, for each $j \in \{1,\ldots,n\}$, 
we split the $j$:th endpoint to $\sIndex_j$ distinct points and attach the $\sIndex_j$ links 
of $\linkpatt$ ending there  to these new $\sIndex_j$ endpoints, 
so that each of them has valence one. 
This results in a link pattern in $\LP_N$, which we denote by $\alpha(\linkpatt)$.
See Figure~\ref{fig: laphamap} for an illustration.

\begin{figure}[h!]
\centering
\includegraphics[scale=.85]{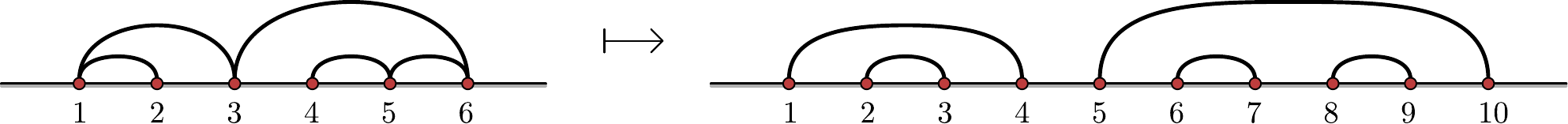}
\caption{\label{fig: laphamap}
Graphical illustration of the map $\linkpatt \mapsto \alpha(\linkpatt)$.
The function $\PartF_\linkpatt(x_1, \ldots, x_6)$ is obtained as a limit of the function
$\PartF_{\alpha(\linkpatt)}(x_1, \ldots, x_{10})$.
}
\end{figure}

Finally, the collection $\{ \PartF_\linkpatt \; | \; \linkpatt \in \LP_\multii \}$ can be defined as follows.
For $\linkpatt \in \LP_\multii$, we 
consider the function $\PartF_{\alpha(\linkpatt)}$ using the map~\eqref{eq: link pattern to pair partition}.
If $\alpha(\linkpatt) = \linkpatt \in \LP_N$, then we set $\PartF_\linkpatt := \PartF_{\alpha(\linkpatt)}$. 
Otherwise, we define $\PartF_\linkpatt$ via a limiting procedure as follows.

\begin{lem}
Let $\kappa \in (0,8) \setminus \bQ$. The following limit
determines a well-defined  smooth function of $(x_1,\ldots,x_n) \in \chamber_n$\textnormal{:}
\begin{align} \label{eq: iterated limit unnormalized}
\widetilde{\PartF}_\linkpatt (x_1,\ldots,x_n) := 
\lim_{\substack{\vspace*{1mm} \\ \hspace*{8.5mm} y_1, \ldots,y_{\sIndex_1} \to x_1 \\ y_{\sIndex_1+1}, \ldots,y_{\sIndex_1+\sIndex_2} \to x_2 \\ \vdots \\ y_{2N-\sIndex_n+1}, \ldots,y_{2N} \to x_n}}
\frac{\PartF_{\alpha(\linkpatt)} (y_1, \ldots, y_{2N})}{\bigg( \underset{1 \leq i < j \leq \sIndex_1}{\prod} (y_j-y_i)  
\underset{1 \leq i < j \leq \sIndex_2}{\prod} (y_{\sIndex_1+j}-y_{\sIndex_1+i})
\underset{1 \leq i < j \leq \sIndex_n}{\prod} (y_{2N-\sIndex_n+j}-y_{2N-\sIndex_n+i}) \bigg)^{2/\kappa}} ,
\end{align}
where $(y_1, \ldots, y_{2N}) \in \chamber_{2N}$ and $N = \frac{1}{2} |\multii|$, for $\linkpatt \in \LP_\multii$.
\end{lem}
\begin{proof}
By~\cite[Lemma~\red{5.2} \& Proposition~\red{5.6}]{Peltola:Basis_for_solutions_of_BSA_PDEs_with_particular_asymptotic_properties}, the limit 
\begin{align*}
\lim_{\substack{\vspace*{1mm} \\ \hspace*{8.5mm} y_1, \ldots,y_{\sIndex_1} \to x_1 \\ y_{\sIndex_1+1}, \ldots,y_{\sIndex_1+\sIndex_2} \to x_2 \\ \vdots \\ y_{2N-\sIndex_n+1}, \ldots,y_{2N} \to x_n}}
\frac{\BasisF_{\alpha(\linkpatt)} (y_1, \ldots, y_{2N})}{\bigg( \underset{1 \leq i < j \leq \sIndex_1}{\prod} (y_j-y_i)  
\underset{1 \leq i < j \leq \sIndex_2}{\prod} (y_{\sIndex_1+j}-y_{\sIndex_1+i})
\underset{1 \leq i < j \leq \sIndex_n}{\prod} (y_{2N-\sIndex_n+j}-y_{2N-\sIndex_n+i}) \bigg)^{2/\kappa}} 
\end{align*}
exists independently of the order of the limits taken, and equals
\begin{align*}
\left( \frac{\qnum{2}}{q - q^{-1}} \right)^{2N}
\left( \prod_{j=1}^n \frac{1}{\qfact{\sIndex_j+1}} \right) 
\times
\BasisF_\linkpatt(x_1,\ldots,x_n) ,
\end{align*}
where $\BasisF_{\alpha(\linkpatt)}$ and $\BasisF_\linkpatt$ are certain Coulomb gas integral functions discussed 
in~\cite[Section~\red{5}]{Peltola:Basis_for_solutions_of_BSA_PDEs_with_particular_asymptotic_properties}
(see also Appendix~\ref{app: Coulomb gas}),
and 
\begin{align*}
\qnum{m} := \frac{q^m - q^{-m}}{q-q^{-1}} , \qquad \qquad 
\qfact{m} := \qnum{1} \qnum{2} \cdots \qnum{m}  , \qquad \qquad 
\textnormal{for }  m \in \bZnn \; \textnormal{ and } \; q = e^{\ii \pi 4 / \kappa} ,
\end{align*}
are $q$-integers and $q$-factorials.
On the other hand, by Proposition~\ref{prop in app} in Appendix~\ref{app: Coulomb gas}
(see also~\cite[Section~\red{6}]{Peltola:Basis_for_solutions_of_BSA_PDEs_with_particular_asymptotic_properties}), we have
\begin{align*}
\BasisF_{\alpha(\linkpatt)} = (B^{2,2}_1)^{N} \PartF_{\alpha(\linkpatt)} 
, \qquad \qquad 
\textnormal{where} \qquad
B^{2,2}_1 = 
\frac{\Gamma(1-4/\kappa)^2}{\Gamma(2-8/\kappa)} ,
\end{align*}
so the limit~\eqref{eq: iterated limit unnormalized} also exists and equals
\begin{align*}
(B^{2,2}_1)^{-N}
\left( \frac{\qnum{2}}{q - q^{-1}} \right)^{2N}
\left( \prod_{i=1}^n \frac{1}{\qfact{\sIndex_i+1}} \right) 
\times \BasisF_\linkpatt(x_1,\ldots,x_n) 
=: \widetilde{\PartF}_\linkpatt (x_1,\ldots,x_n) .
\end{align*}
This proves the lemma.
\end{proof}

\begin{defn} \label{defn:Zlinkpatt}
It turns out to be natural to define $\PartF_\linkpatt \colon \chamber_n \to \bC$  via the limit in~\eqref{eq: iterated limit unnormalized}
with a different normalization: we set
\begin{align} 
\label{eq: iterated limit}
\PartF_\linkpatt (x_1,\ldots,x_n) 
:= \; & \left( \frac{q - q^{-1}}{\qnum{2}} \right)^{2N}
\left( \prod_{i=1}^n \qfact{\sIndex_i+1} \right) 
\times \widetilde{\PartF}_\linkpatt (x_1,\ldots,x_n)  \\
\label{eq: relation of F and Z}
= \; & (B^{2,2}_1)^{-N} \; \BasisF_\linkpatt(x_1,\ldots,x_n) ,
\end{align}
where again, $\BasisF_\linkpatt$ is a certain Coulomb gas integral function from~\cite[Section~\red{5}]{Peltola:Basis_for_solutions_of_BSA_PDEs_with_particular_asymptotic_properties}.
(We will not use the precise form of $\BasisF_\linkpatt$, so we omit its definition here.
For the interested reader, the detailed definition can be found 
in~\cite[Theorems~\red{3.1},~\red{5.1}, and~\red{5.3}]{Peltola:Basis_for_solutions_of_BSA_PDEs_with_particular_asymptotic_properties}.)
\end{defn}

In the next theorem,  we summarize salient properties of these functions.
We invite the reader to compare them with the properties
$\mathrm{(PDE)}$, $\mathrm{(COV)}$, and 
$\mathrm{(ASY)}$ for the multiple $\SLEk$ pure partition functions,
stated in~(\ref{eq: multiple SLE Mobius covariance},~\ref{eq: multiple SLE PDEs},~\ref{eq: multiple SLE asymptotics}).

\begin{theorem} \label{thm: big prop}
\textnormal{\cite[Theorem~\red{5.3} \& Proposition~\red{5.6}]{Peltola:Basis_for_solutions_of_BSA_PDEs_with_particular_asymptotic_properties}}
\; 
Let $\kappa \in (0,8) \setminus \bQ$. 
The collection
\begin{align*}
\{\PartF_\linkpatt \; | \; \linkpatt \in \LP_\multii, \, \multii \in \bZpos^n,  \, n \in \bZnn\}
\end{align*}
of functions 
defined via~\textnormal{(\ref{eq: iterated limit unnormalized},~\ref{eq: iterated limit})}
have the following properties: 
\begin{itemize}
\item[\red{$\mathrm{(PDE)}$}] {\bf \textit{Partial differential equations}}:  
For any $\linkpatt \in \LP_\multii$, the function
$\PartF_\linkpatt$ satisfies PDE system~\eqref{eq: BSA differential equations} with $\multidim = \multii + 1 $. 

\item[\red{$\mathrm{(COV)}$}] {\bf \textit{M\"obius covariance}}: 
The function $\PartF_\linkpatt$ is M\"obius covariant:
\begin{align}\label{eq: Mobius covariance}
\PartF_\linkpatt(x_{1},\ldots, x_n) = 
\prod_{i=1}^{n}\Mob'(x_{i})^{h_{1,s_{i}}} \times 
\PartF_\linkpatt \left(\Mob(x_{1}),\ldots,\Mob(x_n)\right) ,
\end{align}
for all M\"obius maps $\Mob \colon \bH \to \bH$ 
such that $\Mob(x_{1}) < \cdots < \Mob(x_n)$.

\item[\red{$\mathrm{(ASY)}$}] {\bf \textit{Asymptotics}}:  For any $j \in \{1,2,\ldots,n-1 \}$, 
$\; m = \frac{1}{2}\left(s_{j}+s_{j+1}-\projdmn-1\right) 
\in \{0,1,\ldots,\min(s_{j},s_{j+1}) - 1 \}$, and
$\xi \in (x_{j-1},x_{j+2})$,
the function $\PartF_\linkpatt$ has the asymptotics property
\begin{align}\label{eq: asymptotic properties}
\lim_{x_j,x_{j+1}\to\xi}
\frac{\PartF_\linkpatt(x_1,\ldots,x_n)}{(x_{j+1}-x_j)^{h_{1,\projdmn} - h_{1,s_j} - h_{1,s_{j+1}}}}
=\; & \begin{cases}
0 , \quad 
& \textnormal{if } \ell_{j,j+1}(\linkpatt) < m ,\\
\frac{\constantfromdiagram{\projdmn}{s_j}{s_{j+1}} \; B^{s_j,s_{j+1}}_{\projdmn}}{ (B^{2,2}_1)^{m}} 
\times 
\PartF_{\hat{\linkpatt}}(x_1,\ldots,x_{j-1},\xi,x_{j+2}\ldots,x_n) ,
& \textnormal{if } \ell_{j,j+1}(\linkpatt) = m ,
\end{cases}
\end{align}
where $\hat{\linkpatt} = \linkpatt\removeLink (m\times\link{j}{j+1})$ 
denotes the valenced link pattern obtained from $\linkpatt$ by removing 
$m$ links $\link{j}{j+1}$ from it (and merging the $j$:th and $(j+1)$:th endpoints if no links remain between them,
and removing endpoints if they become empty),  
and the multiplicative constants are non-zero and explicit:
\begin{align*}
B^{s_j,s_{j+1}}_{\projdmn} = \; & \frac{1}{m!} \; \prod_{u=1}^m
    \frac{\Gamma\big( 1 - \frac{4}{\kappa}(s_j-u)\big) \; \Gamma\big( 1 - \frac{4}{\kappa}(s_{j+1}-u)\big) \; \Gamma\big( 1 + \frac{4}{\kappa}u \big)}
        {\Gamma\big( 1 + \frac{4}{\kappa}\big) \; \Gamma\big( 2 - \frac{4}{\kappa}(s_j+s_{j+1}-m-u)\big)} , \\
\constantfromdiagram{\projdmn}{s_j}{s_{j+1}} = \; &
\frac{\qnum{2}^m\qfact{s_j-1}\qfact{s_{j+1}-1}\qfact{s_j+s_{j+1}-2m-1}}{\qfact{s_j-1-m}\qfact{s_{j+1}-1-m}\qfact{s_j+s_{j+1}-m-1}} ,
\qquad \qquad 
\textnormal{where}
\qquad \qquad
m = \frac{\left(s_{j}+s_{j+1}-\projdmn-1\right)}{2} .
\end{align*}

\item[\red{$\mathrm{(CAS)}$}] 
 {\bf \textit{Cascade property}}:  
For any $1 \leq j < k \leq n$ and $\xi \in (x_{j-1},x_{k+1})$, we have
\begin{align*}
\lim_{x_j, x_{j+1}, \ldots, x_k \to \xi}
\frac{\PartF_\linkpatt(x_1,\ldots,x_n)}{\PartF_\tau(x_j,\ldots,x_k)}
= \; & \PartF_{\linkpatt \removeLink \tau}(x_1,\ldots,x_{j-1},\xi,x_{k+1},\ldots,x_n) ,
\end{align*}
where $\tau$ denotes the sub-link pattern of $\linkpatt$ between the $j$:th and $k$:th endpoints,
and $\linkpatt \removeLink \tau$ denotes the link pattern obtained from $\linkpatt$ by removing the  sub-link pattern  $\tau$, 
as detailed in~\textnormal{\cite[Section~\red{5.3}]{Peltola:Basis_for_solutions_of_BSA_PDEs_with_particular_asymptotic_properties}}.
\end{itemize}
\end{theorem}
\begin{proofIdea}
The first three asserted properties follow from~\cite[Theorem~\red{5.3}]{Peltola:Basis_for_solutions_of_BSA_PDEs_with_particular_asymptotic_properties} and the relationship~\eqref{eq: relation of F and Z}
of the functions $\PartF_\linkpatt$ with the functions $\BasisF_\linkpatt$ considered in~\cite{Peltola:Basis_for_solutions_of_BSA_PDEs_with_particular_asymptotic_properties},
and the cascade property $\mathrm{(CAS)}$  then
follows from~\cite[Proposition~\red{5.6}]{Peltola:Basis_for_solutions_of_BSA_PDEs_with_particular_asymptotic_properties}.
\end{proofIdea}

We invite the reader to compare property~$\mathrm{(ASY)}$ with Equation~\eqref{eq: fusion rules General 2}
in Section~\ref{subsec:FusionCFT} (and to see Proposition~\ref{prop: OPE} in Section~\ref{subsec:OPE for ppf}).

\begin{remark} \label{rem:fusion}
Asymptotics properties~\eqref{eq: asymptotic properties} are consistent with the asymptotics discussed in Section~\ref{sec: Multiple SLE partition functions}
for the functions $\PartF_\alpha$: 
\begin{itemize}
\item If $s_j = s_{j+1} = 2$ and $\projdmn = 1$  (so $m = 1$),
then~\eqref{eq: asymptotic properties} 
agrees with~\eqref{eq: multiple SLE asymptotics} (see also~\eqref{eq: fusion rules spec}):  we have
\begin{align*}
B^{s_j,s_{j+1}}_{\projdmn} = B^{2,2}_{1} , \qquad\qquad
\constantfromdiagram{\projdmn}{s_j}{s_{j+1}} = \constantfromdiagram{1}{2}{2} = 1 , 
\qquad\qquad
h_{1,1} - 2h_{1,2} = \frac{\kappa-6}{\kappa} ,
\end{align*}
and asymptotics property $\mathrm{(ASY)}$ reads
\begin{align*}
\lim_{x_j,x_{j+1}\to\xi}
\frac{\PartF_\alpha(x_1,\ldots,x_{2N})}{(x_{j+1}-x_j)^{(\kappa-6)/\kappa}}
=\; & \begin{cases}
0, \quad 
& \textnormal{if } \ell_{j,j+1}(\alpha)  =  0 , \\
\PartF_{\hat{\alpha}}(x_1,\ldots,x_{j-1},x_{j+2}\ldots,x_{2N}) ,
& \textnormal{if } \ell_{j,j+1}(\alpha) = 1,\\
\end{cases}
\end{align*}
where $\hat{\alpha} = \alpha \removeLink \link{j}{j+1}$. 

\item If $s_j = s_{j+1} = 2$ and $\projdmn = 3$ (so $m = 0$),
then~\eqref{eq: asymptotic properties} 
agrees with~\eqref{eq: fusion limit example} (see also~\eqref{eq: fusion rules spec}): we have
\begin{align*}
B^{s_j,s_{j+1}}_{\projdmn} = B^{2,2}_{3} = 1 , \qquad\qquad
\constantfromdiagram{\projdmn}{s_j}{s_{j+1}} = \constantfromdiagram{3}{2}{2} = 1 , \qquad\qquad
h_{1,3} - 2h_{1,2} = \frac{2}{\kappa} ,
\end{align*}
and asymptotics property $\mathrm{(ASY)}$ reads
\begin{align*}
\lim_{x_j,x_{j+1}\to\xi}
\frac{\PartF_\alpha(x_1,\ldots,x_{2N})}{(x_{j+1}-x_j)^{2/\kappa}}
=\; &
\PartF_{\hat{\alpha}}(x_1,\ldots,x_{j-1},\xi,x_{j+2}\ldots,x_{2N}) ,
\qquad \textnormal{if } \ell_{j,j+1}(\alpha) = 0,
\end{align*}
where $\hat{\alpha}$ is obtained from $\alpha$ via fusion of the points $j$ and $j+1$:
e.g., $\quad \alpha = \vcenter{\hbox{\includegraphics[scale=0.35]{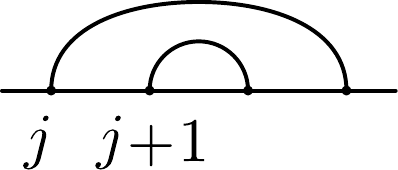}}}$
$\quad \longmapsto \quad$
$\hat{\alpha} = \hbox{\includegraphics[scale=0.35]{figures/link-fused.pdf}}$.
\end{itemize}
\end{remark}

In fact, it follows from~\cite[Theorem~\red{2}]{Flores-Kleban:Solution_space_for_system_of_null-state_PDE4} that
the functions $\PartF_\alpha$ with $\kappa \in (0,8) \setminus \bQ$ have a Frobenius series of the form
\begin{align*}
\PartF_\alpha(x_1,\ldots,x_{2N})
\; = \;  (x_{j+1}-x_j)^{(\kappa-6)/\kappa} \; F_{1,1}(x_1,\ldots,x_{2N})
\; + \;  (x_{j+1}-x_j)^{2/\kappa} \; F_{1,3}(x_1,\ldots,x_{2N}) .
\end{align*}
We invite the reader to compare this 
with the fusion rules~\eqref{eq: fusion rules spec} in Section~\ref{subsec:FusionCFT}
and the observations in Remark~\ref{rem:fusion}. 
When $\kappa \in \bQ$, the above series could contain logarithmic terms, 
see~\cite{Flores-Kleban:Solution_space_for_system_of_null-state_PDE4}.

\bigskip

For those functions $\PartF$  that are (obtained as limits of) functions in $\Sol_N$, 
the conclusions in Theorem~\ref{thm: Dubedat} follow from Theorem~\ref{thm: big prop} combined with 
Theorems~\ref{thm:Steven} and~\ref{thm::purepartition_existence_forallK}.
Indeed, Theorem~\ref{thm:Steven} says that $\dmn \Sol_N  = \Catalan_N$, and
Theorem~\ref{thm::purepartition_existence_forallK} gives a basis of cardinality $\Catalan_N$
for this space, which coincides with the collection $\{\PartF_\alpha \; | \; \alpha\in \LP_N\}$
that appears as a special case in Theorem~\ref{thm: big prop}.
Then, items $\mathrm{(PDE)}$ and $\mathrm{(ASY)}$ in 
Theorem~\ref{thm: big prop} show that all these functions (and their limits) satisfy the conclusions in Theorem~\ref{thm: Dubedat}.

\begin{quest} \label{quest: general solutions}
Are there other solutions to the second order PDE system~\eqref{eq: multiple SLE PDEs} than the ones belonging to the solution space $\Sol_N$?
\end{quest}

By definition~\eqref{eq: solution space}, all solutions of~\eqref{eq: multiple SLE PDEs} which satisfy in addition
the M\"obius covariance~\eqref{eq: multiple SLE Mobius covariance} and growth bound~\eqref{eqn::powerlawbound}
belong to $\Sol_N$. 
However, PDE system~\eqref{eq: multiple SLE PDEs} does have other solutions too ---
for instance, solutions satisfying other M\"obius covariance properties (where infinity is a special point).
This kind of solutions are also discussed in~\cite{Kytola-Peltola:Conformally_covariant_boundary_correlation_functions_with_quantum_group,
Peltola:Basis_for_solutions_of_BSA_PDEs_with_particular_asymptotic_properties},
and all of these solutions 
 satisfy the conclusions in Theorem~\ref{thm: Dubedat}.
The real problem in Question~\ref{quest: general solutions} is therefore whether there exist solutions 
to PDEs~\eqref{eq: multiple SLE PDEs} other than the Coulomb gas integral functions studied 
in~\cite{Flores-Kleban:Solution_space_for_system_of_null-state_PDE1, Flores-Kleban:Solution_space_for_system_of_null-state_PDE2,
Flores-Kleban:Solution_space_for_system_of_null-state_PDE3, Flores-Kleban:Solution_space_for_system_of_null-state_PDE4,
Kytola-Peltola:Conformally_covariant_boundary_correlation_functions_with_quantum_group,
Peltola:Basis_for_solutions_of_BSA_PDEs_with_particular_asymptotic_properties}.

\begin{quest}
Are there other solutions to the general PDE systems~\eqref{eq: BSA differential equations} than the ones obtained as limits of functions in $\Sol_N$?
\end{quest}

For the PDEs of higher order, there is no analogue of Theorem~\ref{thm:Steven}.
Indeed, to prove the upper bound in Theorem~\ref{thm:Steven}, elliptic PDE theory is used, 
which seems quite specific to the case of the second order PDE system~\eqref{eq: multiple SLE PDEs}.

\subsection{\label{subsec:OPE for ppf} Operator product expansion for multiple SLE partition functions} 

We conclude with specific fusion rules for the functions $\PartF_\linkpatt$.
From the point of view of representation theory, the following expansion is not 
surprising~\cite{Peltola:Basis_for_solutions_of_BSA_PDEs_with_particular_asymptotic_properties}, 
but analytical verification for it is challenging.
We only know a proof using the representation theory of 
the quantum group $\Uqsltwo$ and the ``spin chain~--~Coulomb gas correspondence'' 
of~\cite{Kytola-Peltola:Conformally_covariant_boundary_correlation_functions_with_quantum_group}
(with $q = e^{\ii \pi 4 / \kappa}$).

\begin{prop} \label{prop: OPE}
Let $\kappa \in (0,8) \setminus \bQ$.
The collection  $\{\PartF_\linkpatt \; | \; \linkpatt \in \LP_\multii, \, \multii \in \bZpos^n,  \, n \in \bZnn\}$
of functions 
defined via~\textnormal{(\ref{eq: iterated limit unnormalized},~\ref{eq: iterated limit})}
satisfies a closed operator product expansion in the following sense:
\begin{itemize}
\item[\red{$\mathrm{(OPE)}$}]  
For any $(x_1,\ldots,x_n) \in \chamber_n$,
and for all $j \in \{1, \ldots, n-1 \}$ and $\xi \in (x_{j-1}, x_{j+2})$, we have
\begin{align} \label{eq: OPE expansion}
\begin{split}
\PartF_\linkpatt (x_1,\ldots,x_n)
= \; & \sum_{\projdmn \in S_{j,j+1}} 
\frac{C_{s_j, s_{j+1}}^{\projdmn}}{(x_{j+1} - x_{j})^{h_{1,s_j} + h_{1,s_{j+1}} - h_{1,\projdmn}}} \;
\PartF_{\linkpatt\removeLink (m\times\link{j}{j+1})}(x_1,\ldots,x_{j-1},\xi,x_{j+2}\ldots,x_n) \\
\; & + o \left( |x_{j+1} - x_{j}|^{\Delta} \right)  ,
\qquad \textnormal{as } 
x_j, x_{j+1} \to \xi,
\end{split}
\end{align}
where we use the notation
\begin{align*}
m = \; &  \frac{s_{j}+s_{j+1}-\projdmn-1}{2}
\in \{0,1,\ldots,\min(s_{j},s_{j+1}) - 1 \} , \\
\projdmn = \; & s_{j}+s_{j+1} - 2 m - 1
\; \in \; \big\{ |s_{j+1} - s_{j}| + 1 , \, |s_{j+1} - s_{j}| + 3 , \, \ldots , \, s_{j}+s_{j+1} - 3 , \, s_{j}+s_{j+1} - 1 \big\} \; =: \; S_{j,j+1} , \\
\Delta = \; & h_{1,s_{j}+s_{j+1} - 1} - h_{1,s_j} - h_{1,s_{j+1}} ,
\end{align*}
and the structure constants are explicit:
\begin{align*}
C_{s_j, s_{j+1}}^{\projdmn} 
=
\begin{cases}
0, \quad 
& \textnormal{if } \ell_{j,j+1}(\linkpatt) < m , \\
\frac{\constantfromdiagram{\projdmn}{s_j}{s_{j+1}} \; B^{s_j,s_{j+1}}_{\projdmn}}{ (B^{2,2}_1)^{m}} ,
& \textnormal{if } \ell_{j,j+1}(\linkpatt) \geq m.\\
\end{cases}
\end{align*}
\end{itemize}
\end{prop}

\begin{proofIdea}
The rough idea is to write the function $\PartF_\linkpatt$ as a sum of terms, each of which has a prescribed asymptotics as claimed in 
the assertion~\eqref{eq: OPE expansion}. This is established using the quantum group symmetry
developed in~\cite{Kytola-Peltola:Conformally_covariant_boundary_correlation_functions_with_quantum_group, 
Peltola:Basis_for_solutions_of_BSA_PDEs_with_particular_asymptotic_properties}.  
The terms with $m \geq \ell_{j,j+1}(\linkpatt)$ are already immediate from~\eqref{eq: asymptotic properties}.
Note that $m = \ell_{j,j+1}(\linkpatt)$ gives the leading asymptotics in the series~\eqref{eq: OPE expansion},
because $\kappa \in (0,8)$.
For the other terms (subleading asymptotics), careful investigation of the Coulomb gas integral construction 
in~\cite[Theorems~\red{3.1} and~\red{5.1}]{Peltola:Basis_for_solutions_of_BSA_PDEs_with_particular_asymptotic_properties}  
gives the asserted terms in~\eqref{eq: OPE expansion}. 
We leave the details to the reader.
\end{proofIdea}

Again, we invite the reader to compare the recursive asymptotics properties~\eqref{eq: OPE expansion} with Equation~\eqref{eq: fusion rules General 2} in Section~\ref{subsec:FusionCFT}.

\section{\label{sec:Malek}From OPE structure to products of random distributions?}

In this final section, we discuss a conformal bootstrap idea, which might be useful when trying to
interpret the general multiple SLE partition functions $\PartF_\linkpatt$ as correlations of some quantum fields in a CFT.
The bootstrap method appeared in
the recent work~\cite{Abdesselam:Second-quantized_Kolmogorov-Chentsov_theorem} of A.~Abdesselam,  
pertaining to a construction of quantum fields as random distributions from some already known fields 
via multiplication. 
Such an idea was proposed in the physics 
literature initiated by K.~Wilson~\cite{Brandt:Derivation_of_renormalized_relativistic_perturbation_theory_from_finite_local_field_equations,
Wilson:Non_Lagrangian_models_of_current_algebra,
Wilson-Zimmermann:Operator_product_expansions_and_composite_field_operators_in_the_general_framework_of_quantum_field_theory, 
Polyakov:Non_Hamiltonian_approach_to_conformal_quantum_field_theory,
Witten:Perturbative_quantum_field_theory}.
In~\cite{Abdesselam:Second-quantized_Kolmogorov-Chentsov_theorem},  
Abdesselam established a mathematical result for the construction of products of random distributions 
using the OPE structure for their correlations (see Theorem~\ref{thm: Malek}).
In general, the problem of multiplication of distributions is notoriously difficult,
and attempts to accomplish this go back to Schwartz, with spectacular success 
recently in the random setup by T.~Lyons's theory of rough paths~\cite{Lyons:Differential_equations_driven_by_rough_signals} 
and M.~Hairer's regularity structures~\cite{Hairer:Regularity_structures}.
Abdesselam proposed a different approach, which seems potentially useful for the multiple SLE partition functions 
thanks to their OPE hierarchy (Proposition~\ref{prop: OPE}).
Our discussion in this section shall be brief and very restricted --- for more details and background,
we refer to~\cite{Abdesselam:Second-quantized_Kolmogorov-Chentsov_theorem} 
and references therein.

\subsection{\label{subsec: ASPWC} Conformal bootstrap}

In CFT \`a la Belavin, Polyakov \& Zamolodchikov, the complete solution of a theory  
should be possible via the ``conformal bootstrap''.
It is a recursive procedure, where the correlation functions of 
the field operators are found using their fusion rules. 
For this, one only has to know the operator content of the theory
(``spectrum'') and the structure constants appearing in~\eqref{eq:OPE general} in Section~\ref{subsec:FusionCFT}.
Using this data, one then recursively derives all correlation functions.  
In the early work~\cite{BPZ:Infinite_conformal_symmetry_in_2D_QFT, BPZ:Infinite_conformal_symmetry_of_critical_fluctuations_in_2D},
the bootstrap was successfully performed, e.g., for the CFT corresponding to the critical 2D Ising model.
In that case, there are only three primary fields: 
the identity $\one$, energy $\varepsilon$, and spin $\sigma$.
The CFT for the Ising model is an example of a minimal model, where there are only finitely many primary fields,
and which have been solved and classified, see, e.g.,~\cite[Chapter~\red{11}]{Mussardo:Statistical_field_theory}.
Recently, the Ising CFT has also been quite well understood as a scaling limit of the Ising lattice 
model~\cite{CHI:Conformal_invariance_of_spin_correlations_in_planar_Ising_model, Hongler-Smirnov:Energy_density_in_planar_Ising_model, 
CHI:inprep}.

For CFTs with infinitely many primary fields,  one encounters apparent difficulties  in the bootstrap program: 
convergence issues, problems with finding the structure constants, and trouble in classifying the primary fields. 
However, certain CFTs have further restrictive data. 
For example, if the theory consists of fields with degeneracies, as discussed in Section~\ref{subsec:CFT}, 
then the fusion rules become relatively simple~\cite[Section~\red{6}]{BPZ:Infinite_conformal_symmetry_in_2D_QFT}.
The fusion rules for the general multiple SLE partition functions $\PartF_\linkpatt$, stated in Proposition~\ref{prop: OPE},
coincide with these fusion rules.
(On the other hand, the structure constants, calculated 
by V.~Dotsenko and V.~Fateev~\cite{Dotsenko-Fateev:Conformal_algebra_and_multipoint_correlation_functions_in_2D_statistical_models, 
Dotsenko-Fateev:4pt_correlation_functions_and_operator_algebra_in_2D_conformal_invariant_theories_with_central_charge_less_than_one}, 
are still rather complicated, and differ slightly from those appearing in Proposition~\ref{prop: OPE}.)

\bigskip

In~\cite{Abdesselam:Second-quantized_Kolmogorov-Chentsov_theorem}, 
the notion of ``abstract systems of pointwise correlations'' was introduced, 
pertaining to the mathematical understanding of products of quantum fields via their OPE hierarchy.
In Sections~\ref{subsec: Random tempered distributions}--\ref{subsec: Bootstrap for SLEs}, 
we discuss results from~\cite{Abdesselam:Second-quantized_Kolmogorov-Chentsov_theorem} 
on how this could be established in practise (see in particular Theorem~\ref{thm: Malek}).

\begin{defn} \label{defn:ASPWC}
\cite{Abdesselam:Second-quantized_Kolmogorov-Chentsov_theorem} 
\;
Let $\indexSet$ be a finite index set. 
An abstract system of pointwise correlations,
\begin{align} \label{eq: abstract system of pointwise correlations}
\{ F_{\index_1, \ldots, \index_n} \; | \; \index_1, \ldots, \index_n \in \indexSet, \, n \in \bZnn \} ,
\end{align}
consists of specifying, for all $n > 0$ and for all $\index_1, \ldots, \index_n \in \indexSet$, smooth functions  
\begin{align*}
F_{\index_1, \ldots, \index_n} \colon \extendedChamber_n^d \to \bC ,
\end{align*}
defined on the configuration space
(or, in dimension $d=2$, equivalently on $\extendedChamber_n$ given in Equation~\eqref{eq: chamberComplex}) 
\begin{align*}
\extendedChamber_n^d :=\; & 
 \{ (z_{1},\ldots,z_n) \in \bR^{nd} \; | \; z_i \neq z_j \textnormal{ if } i \neq j \} ,
\end{align*}
with normalization convention $F_\emptyset \equiv 1$, for $n=0$.
This collection is required to satisfy the following properties:
\begin{itemize}
\item {\bf \textit{Permutation symmetry}}: For all permutations $\sigma \in \SymmGrp_n$ and for all $(z_1, \ldots, z_n) \in \extendedChamber_n^d$, we have 
$F_{\index_{\sigma(1)}, \ldots, \index_{\sigma(n)}} (z_{\sigma(1)}, \ldots, z_{\sigma(n)}) =  F_{\index_1, \ldots, \index_n} (z_1, \ldots, z_n)$.

\item {\bf \textit{Unit object}}: 
There exists a distinguished object
$\index_0 \in \indexSet$ such that, for all $\index_1, \ldots, \index_n \in \indexSet$ and  for all
$(z_0,z_1, \ldots, z_n) \in \extendedChamber_{n+1}$,
we have $F_{\index_0, \index_1, \ldots, \index_n} (z_0,z_1, \ldots, z_n) = F_{\index_1, \ldots, \index_n} (z_1, \ldots, z_n)$.

\item {\bf \textit{Scaling dimensions}}: 
To each $\index \in \indexSet$, we associate a real number $D_\index$,
and we set $D_{\index_0} := 0$.

\item {\bf \textit{OPE structure}}: The collection~\eqref{eq: abstract system of pointwise correlations} 
of functions satisfies a closed operator product expansion 
in the following sense: given $D \in \bR$, 
for all $\index_1, \ldots, \index_n \in \indexSet$, 
and for all $(z_1, \ldots, z_n) \in \extendedChamber_n^d$
and $\xi \in \bR^d \setminus \{z_3, \ldots, z_n\}$, we have
\begin{align} \label{eq: OPE expansion malek}
\begin{split}
F_{\index_1, \index_2, \index_3, \ldots, \index_n} (z_1, z_2, z_3, \ldots, z_n)
= \; &  \sum_{\index \in \indexSet(D)} 
\frac{C_{\index_1, \index_2}^{\index}}{|z_2 - z_1|^{D_{\index_1} + D_{\index_2} - D_\index}} \;
F_{\index, \index_3, \ldots, \index_n} (\xi, z_3,\ldots, z_n) 
\\
\; & 
\; + o \left( |z_2 - z_1|^{D - D_{\index_1} - D_{\index_2} } \right) ,
\qquad \textnormal{as } 
z_1, z_2 \to \xi ,
\end{split}
\end{align}
where $C_{\index_1, \index_2}^{\index} \in \bC$ are some constants, and
$\indexSet(D) := \{ \index \in \indexSet \; | \; D_\index \leq D \}$.
\end{itemize}
\end{defn}

In dimension $d=1$, the general multiple SLE partition functions from Definition~\ref{defn:Zlinkpatt} in Section~\ref{subsec: fusion SCCG} (with $\kappa \in (0,8) \setminus \bQ$),
\begin{align*}
\{\PartF_\linkpatt \; | \; \linkpatt \in \LP_\multii, \, \multii \in \bZnn^n,  \, n \in \bZnn\},
\end{align*}
form an abstract system of pointwise correlations in a loose sense. Namely,
these functions $\PartF_\linkpatt \colon \chamber_n \to \bC$ are a priori defined on 
the  configuration space~\eqref{eq: chamber},
where their variables are ordered. Therefore, the permutation symmetry is not meaningful.
The functions are indexed by the valences $\multii = (\sIndex_1, \ldots, \sIndex_n) \in \bZnn^n$
of the valenced link patterns $\linkpatt \in \LP_\multii$, and the valence zero, $\sIndex_0 = 0$,
can be thought of as the unit object (corresponding to the empty link pattern $\emptyset \in \LP_0$) 
omitted from $\linkpatt$ as in Definition~\ref{defn:ASPWC}.
To each valence $\sIndex_j \neq 0$, the conformal weight $h_{1,\sIndex_j+1}$
is associated as in Section~\ref{subsec: fusion SCCG}. 
Finally, Proposition~\ref{prop: OPE} gives an OPE structure for this collection of functions,
where the scaling dimensions equal the conformal weights, $D_{\sIndex_j} = h_{1,\sIndex_j+1}$.

The functions $\PartF_\linkpatt \colon \chamber_n \to \bC$ can also be analytically continued to become 
functions of $n$ complex variables on $\extendedChamber_n$.
Then, by including also the anti-holomorphic sector, i.e., by considering functions of the form 
$F_\linkpatt(z_1,\ldots,z_n) := \PartF_\linkpatt (z_1,\ldots,z_n) \PartF_\linkpatt(\bar{z}_1,\ldots,\bar{z}_n)$,
where $\bar{z}$ is the complex conjugate of $z \in \bC$, one obtains a collection
\begin{align*}
\{F_\linkpatt \; | \; \linkpatt \in \LP_\multii, \, \multii \in \bZnn^n,  \, n \in \bZnn\}
\end{align*}
of functions $F_\linkpatt \colon \extendedChamber_n \to \bC$ in dimension $d=2$,
with OPE structure again obtained from Proposition~\ref{prop: OPE}, 
but this time with scaling dimensions $D_{\sIndex_j}  = h_{1,\sIndex_j+1} + h_{1,\sIndex_j+1} = 2h_{1,\sIndex_j+1}$.
(We remark that the functions $F_\linkpatt$ are not permutation-invariant.
However, it is possible to construct a collection of functions that are permutation-invariant and single-valued
--- see~\cite[Chapter~\red{9}]{DMS:CFT} and~\cite[Chapter~\red{11}]{Mussardo:Statistical_field_theory}.
We leave the precise verification of this to future work~\cite{Flores-Peltola:Monodromy_invariant_correlation_function}.)

\subsection{\label{subsec: Random tempered distributions} Random tempered distributions}

In quantum field theory (QFT) , the ``fields'' can be viewed as operator-valued distributions, 
sending suitable test functions to operators acting on some Hilbert space, 
see, e.g.,~\cite{Glimm-Jaffe:Quantum_physics_a_functional_integral_point_of_view, Schottenloher:Mathematical_introduction_to_CFT}, 
and references therein. 
The ``vacuum expectation values'' $\big\langle \cdots \big \rangle$ of the fields
are then defined as tempered distributions \`a la Schwartz, say.
There are various axiomatic approaches to QFT, where the aforementioned
objects are required to satisfy a set of properties, e.g., the Wightman axioms. 
From the point of view of statistical physics and conformal field theory,
Euclidean QFT is relevant.
An axiomatic setting for Euclidean QFT is provided by the Osterwalder-Schrader axioms. 
In constructive field theory, the main objective is to construct fields that satisfy such axioms.

Euclidean two-dimensional QFT can be formulated 
by thinking of the quantum fields $\Phi$ as distribution-valued random variables, 
i.e., assigning a probability measure $\mathbb{P}$ to the space $S'(\bR^d)$ of tempered distributions 
acting on test functions in the Schwartz space $S(\bR^d)$ of rapidly decreasing functions.
The main problem is to find a probability measure 
such that the Osterwalder-Schrader axioms are satisfied.
(For interacting fields, this is a very difficult problem, especially in four and higher dimensions.)

Let $(\Omega, \mathcal{F}, \mathbb{P})$ be a probability space.
Suppose that for each $\index \in \indexSet$, we associate a random tempered distribution $\Phi_{\index}$,
that is, a random variable taking values in the space 
$S'(\bR^d)$, such that the map $\omega \to \Phi_{\index}(\omega)$ is $(\mathcal{F}, \textnormal{Borel}(S'(\bR^d)))$-measurable.
Suppose also that $\Phi_{\index}$ have finite moments, 
i.e.,
for all $\index \in \indexSet$, $p \geq 1$, and $f \in S(\bR^d)$, we have $\Phi_{\index}(f) \in L^p(\Omega, \mathcal{F}, \mathbb{P})$.
Then, the correlations 
\begin{align*}
\mathbb{E} \big[ \Phi_{\index_1}(f_1) \cdots \Phi_{\index_n}(f_n) \big]
= \mathbb{E} \big[ \Phi_{\index_1} \cdots \Phi_{\index_n} \big] (f_1, \ldots, f_n) 
:= \int \Phi_{\index_1}(f_1) \cdots \Phi_{\index_n}(f_n) \; \ud \mathbb{P} 
\end{align*}
are $n$-linear functionals of $f_1, \ldots, f_n \in S(\bR^d)$.

Sometimes pointwise correlations of the fields $\Phi_{\index}$ 
can also be defined by a renormalization 
procedure~\cite{Abdesselam:Second-quantized_Kolmogorov-Chentsov_theorem}:
\begin{align} \label{eq: correlation functions renormalization}
\mathbb{E} \big[ \Phi_{\index_1}(z_1) \cdots \Phi_{\index_n}(z_n) \big]
= \mathbb{E} \big[ \Phi_{\index_1} \cdots \Phi_{\index_n} \big] (z_1, \ldots, z_n) 
:= \lim_{r \searrow -\infty} \int
\Phi_{\index_1}\big(L^{2r} \rho(L^r (\cdot - z_1)) \cdots \Phi_{\index_n}(L^{2r} \rho(L^r (\cdot - z_n)\big) \; \ud \mathbb{P} ,
\end{align} 
where $L > 1$ is a fixed real number and
$\rho \colon \bR^d \to \bR$ is a smooth, compactly supported mollifier with $\int_{\bR^d} \rho(z) \; \ud z = 1$.
Using the physicists' $\langle \cdots \rangle$ notation 
(c.f. Sections~\ref{subsec:CFT}--\ref{subsec:mgles}), we could then write
\begin{align*}
\big\langle \Phi_{\index_1}(z_1) \cdots \Phi_{\index_n}(z_n) \big\rangle 
:= \mathbb{E}[ \Phi_{\index_1}(z_1) \cdots \Phi_{\index_n}(z_n) ] .
\end{align*} 
Now, assume that these are smooth, 
locally integrable functions of $(z_1, \ldots, z_n) \in \extendedChamber_n^d$: 
for all compact sets $K \subset \extendedChamber_n^d$, we have
\begin{align*}
\int_{K \cap \extendedChamber_n^d} 
\big| \big\langle \Phi_{\index_1}(z_1) \cdots \Phi_{\index_n}(z_n) \big\rangle \big|  \; 
\ud z_1 \cdots \ud z_n < \infty .
\end{align*}
Then, the joint moments can be written as integrals against the pointwise correlations:
for all $f_1, \ldots, f_n \in S(\bR^d)$, we have
\begin{align} \label{eq: correlations}
\mathbb{E} \big[ \Phi_{\index_1}(f_1) \cdots \Phi_{\index_n}(f_n)\big] 
= \int_{\extendedChamber_n^d} f_1(z_1) \cdots f_n(z_n) \; 
\big\langle \Phi_{\index_1}(z_1) \cdots \Phi_{\index_n}(z_n) \big\rangle \; \ud z_1 \cdots \ud z_n .
\end{align}

\begin{quest} \label{quest}
Given an abstract system $\{ F_{\index_1, \ldots, \index_n} \; | \; \index_1, \ldots, \index_n \in \indexSet, \, n \in \bZnn \}$
of pointwise correlations, 
can one construct a system of random distributions 
$\{ \Phi_{\index} \; | \; \index \in \indexSet \}$ 
whose correlation functions are given by 
$\big\langle \Phi_{\index_1}(\cdot) \cdots \Phi_{\index_n}(\cdot) \big\rangle = F_{\index_1, \ldots, \index_n}$?
\end{quest}

One motivation to consider this question is that from scaling limits of critical lattice models,
one should obtain abstract systems of pointwise correlations, 
as has been successfully done for the 2D Ising 
model~\cite{Hongler-Smirnov:Energy_density_in_planar_Ising_model, 
CHI:Conformal_invariance_of_spin_correlations_in_planar_Ising_model, 
CHI:inprep} (at least implicitly).
Then, one would hope that also the ``lattice local fields'' in these models 
would converge in the scaling limit 
to some random distributions, whose correlation functions are the scaling limits of the discrete correlations.
This is known for the spin (magnetization) field in the Ising model~\cite{CGN:Planar_Ising_magnetization_field1}, but not for the energy field.
(In fact, there is evidence that the energy field might not have a scaling limit as a random 
distributions~\cite{Clement-Christophe} --- hence, other approaches might be needed for this case.)

Another motivation  for Question~\ref{quest} 
comes from trying to mathematically understand CFTs in relation with SLEs.
Indeed, the general multiple SLE partition functions $\PartF_\linkpatt$
could be morally viewed as abstract systems of pointwise correlations, 
even though they are defined on the boundary $\bR = \partial \bH$, for variables ordered as $x_1 < \cdots < x_n$, 
and they are not permutation-invariant (c.f. Section~\ref{subsec: ASPWC}). 
It would be interesting to see whether one could make sense of ``SLE generating fields'' $\Phi_{1,2}$
and relate them to ``boundary condition changing operators''  \`a la J.~Cardy~\cite{Cardy:SLE_and_Dyson_circular_ensembles, Cardy:SLE_for_theoretical_physicists}.
Recall also the discussion in Section~\ref{subsec:mgles}.

\begin{quest} \label{quest2}
Does there exist a range of parameters $\kappa > 0$ so that the multiple $\SLE_\kappa$ partition functions 
give rise to an abstract system of pointwise correlations 
which are correlation functions of a system of random distributions?
\end{quest}

When $\kappa=4$, the multiple $\SLE_\kappa$ partition functions 
are related to level lines of the Gaussian free field (free boson), which is well-understood as a random 
distribution~\cite{Sheffield:GFF_for_mathematicians, Dubedat:SLE_and_free_field, 
Schramm-Sheffield:A_contour_line_of_the_continuum_GFF, 
Sheffield-Miller:Imaginary_geometry1, 
Peltola-Wu:Global_and_local_multiple_SLEs_and_connection_probabilities_for_level_lines_of_GFF}.
Namely, O.~Schramm and S.~Sheffield proved in~\cite{Schramm-Sheffield:A_contour_line_of_the_continuum_GFF} 
that the level lines, when properly defined, are multiple $\SLE_\kappa$ processes with $\kappa=4$.
In particular, Question~\ref{quest2} might be solvable in this case by considering correlations of the Gaussian free field.

On the other hand, for the case of $\kappa=3$ one might need a more general notion of ``quantum fields''.
Namely, the multiple $\SLE_\kappa$ partition function $\PartF_{\textnormal{Ising}}$ from~(\ref{eq: Ising pf general domain},~\ref{eqn::ZIsingtotal}) with $\kappa=3$
can be identified with a correlation function in the Ising model 
for the energy density, or the free fermion, on the boundary.
However, neither the energy density nor the fermion is understood as a random distribution.
(We also remark that the scaling dimension in this case equals $\frac{6-\kappa}{2\kappa} = 1/2$, 
which lies exactly at the edge of the admissible range in Abdessalam's Theorem~\ref{thm: Malek} 
 (with dimension $d=1$), stated in Section~\ref{subsec: Bootstrap for SLEs}.)

\subsection{\label{subsec: Bootstrap for SLEs} Bootstrap construction?}

Even though the moment problem in Question~\ref{quest} seems difficult,
there do exist fields that are understood as random distributions 
(e.g., the Ising magnetization field, the Gaussian free field, $\varphi_2^4$ in 2D).
The next natural but very difficult question is whether their products are distributions too.

\begin{quest}
Given $\{ \Phi_{\index} \; | \; \index \in \indexSet \}$, can one 
make sense of products of the form $\Phi_{\index_1,\index_2} = \Phi_{\index_1} \Phi_{\index_2}$ 
as random distributions?
\end{quest}

One answer to this question was given in~\cite{Abdesselam:Second-quantized_Kolmogorov-Chentsov_theorem}
using the abstract pointwise correlations defined in Section~\ref{subsec: ASPWC}. 
We present here  a simplified and slightly informal statement, referring to~\cite{Abdesselam:Second-quantized_Kolmogorov-Chentsov_theorem} 
for the precise formulation and extensions.
The main idea is to use the OPE hierarchy~\eqref{eq: OPE expansion malek} of the correlations to recursively 
construct fields from already constructed ones.

\begin{theorem} \label{thm: Malek}
 \textnormal{\cite[Theorem~\red{1}, simplified]{Abdesselam:Second-quantized_Kolmogorov-Chentsov_theorem}}
\;

Let $\mathrm{PWC} := \{ F_{\index_1, \ldots, \index_n} \; | \; \index_1, \ldots, \index_n \in \indexSet, \, n \in \bZnn \}$
be an abstract system of pointwise correlations. Assume that the following further properties hold:
\begin{itemize}
\item For all $\index \in \indexSet$, we have $D_\index \in [0,\frac{d}{2})$.

\item There exists a constant $C > 0$ such that, for all $F_{\index_1, \ldots, \index_n} \in \mathrm{PWC}$
and for all $(z_1, \ldots, z_n) \in \extendedChamber_n^d$, we have
\begin{align} \label{eq: BNNFB}
\tag{BNNFB}
\big| F_{\index_1, \ldots, \index_n} (z_1, \ldots, z_n) \big|
 \leq \; & \; C \;
\prod_{i=1}^n \big( \min_{j \neq i} |z_i-z_j| \big)^{-D_{\index_i}} . 
\end{align}

\item For a subset $\indexSet_0 \subset \indexSet$ of indices, 
we have already constructed random distributions $\{ \Phi_\index \; | \; \index \in \indexSet_0 \}$
whose correlation functions are given by
$\mathrm{PWC}_0 := \{ F_{\index_1, \ldots, \index_n} \; | \; \index_1, \ldots, \index_n \in \indexSet_0, \, n \in \bZnn \}$.
\end{itemize}
Then, for each $\index^* \in \indexSet \setminus \indexSet_0$ such that 
$C_{\index_1, \index_2}^{\index^*} \neq 0$, for some $\index_1, \index_2 \in \indexSet_0$, and
$\indexSet(D_{\index^*}) \setminus \{\index^*\} \subset \indexSet_0$, 
we can construct a random distribution $\Phi_{\index^*}$ as the renormalized product 
$\Phi_{\index_1} \Phi_{\index_2}$\textnormal{:}
\begin{align*}
\textnormal{``}
\Phi_{\index^*} (z)
:= \lim_{w \to z} \;
\frac{|w - z|^{D_{\index_1} + D_{\index_2} - D_{\index^*}}}{C_{\index_1, \index_2}^{\index^*}}
\left(
\Phi_{\index_1} (w) \Phi_{\index_2} (z) \; - \;
\sum_{\index \in \indexSet(D_{\index^*}) \setminus \{\index^*\}} 
\frac{C_{\index_1, \index_2}^{\index}}{|w - z|^{D_{\index_1} + D_{\index_2} - D_\index}} \; \Phi_{\index}(z) \right)
\textnormal{''},
\end{align*}
where the quotation marks indicate that the equation is to be understood in the sense of distributions
(as a limit in $L^p(\Omega, \mathcal{F}, \mathbb{P})$ for any $p \geq 1$, and $\mathbb{P}$-almost surely)
and in terms of an appropriate regularization procedure --- 
see~\textnormal{\cite{Abdesselam:Second-quantized_Kolmogorov-Chentsov_theorem}} for details.
\end{theorem}

\begin{conj} \label{conj:Malek}
\textnormal{\cite[Conjecture~\red{1}, simplified]{Abdesselam:Second-quantized_Kolmogorov-Chentsov_theorem}} 
\;
Any reasonable conformal field theory satisfies 
$\mathrm{OPE}$~\eqref{eq: OPE expansion malek} and~\eqref{eq: BNNFB}. 
\end{conj}

By inspection of Equation~\eqref{eq: Kac weights BSA} defining the Kac conformal weights $h_{1,s}$,
we note that when $\kappa \in (0,8)$,  we have $0 = h_{1,1} < h_{1,3}$ and $h_{1,r} < h_{1,s}$, for $2 \leq r < s$.
Therefore, the construction proposed by Theorem~\ref{thm: Malek} 
for a family of fields $\Phi_{1,s}$ could possibly give rise to iterated ``operator products'' of type
\begin{align*}
\textnormal{``}
\Phi_{1,s}(x) 
:= \lim_{y \searrow x} \;
\frac{(y - x)^{h_{1,2} + h_{1,s-1} - h_{1,s}}}{C_{2, s-1}^{s}}
\left(
\Phi_{1,2} (y) \Phi_{1,s-1} (x) \; -  \;
\sum_{r < s}
\frac{C_{2, s-1}^{r}}{(y - x)^{h_{1,2} + h_{1,s-1} - h_{1,r}}} \; \Phi_{1,r}(x) \right) 
\textnormal{''} ,
\end{align*}
provided that one first could make sense of the ``building block'' fields $\Phi_{1,2}$ as random distributions (which may or may not be possible).
For instance, we would like to interpret the multiple $\SLE_\kappa$ 
partition functions 
$\PartF_\alpha(x_1, \ldots, x_{2N})$, for $\alpha \in \LP_N$, as correlation functions of type 
$\big\langle \Phi_{1,2}(x_1) \cdots \Phi_{1,2}(x_{2N}) \big\rangle$, 
and then construct fields $\Phi_{1,s}$ by using the OPE structure from Proposition~\ref{prop: OPE}.
Indeed, we already noticed in Section~\ref{subsec: ASPWC} that these functions give rise to an abstract system of pointwise correlations 
(relaxing permutation invariance). 
When $\kappa \in (0,6]$, the partition functions $\PartF_\alpha$ also satisfy the required bound~\eqref{eq: BNNFB}:

\begin{prop} \label{prop:Maleks bound}
If $\kappa\in (0,6]$, then the collection
$\{ \PartF_\alpha \; | \; \alpha \in \LP \} $ 
of multiple $\SLEk$ pure partition functions satisfies the bound
\begin{align} \label{eq:Maleks bound}
0 < \PartF_\alpha (x_1, \ldots, x_{2N})  \, \le \, \prod_{i = 1}^{2N} \big( \min_{j \neq i} |x_i-x_j| \big)^{-h_{1,2}} 
\; = \; \prod_{i = 1}^{2N} \big( \min \big( |x_i-x_{i - 1}| , |x_i-x_{i + 1}| \big) \big)^{-h_{1,2}} .
\end{align}
\end{prop}

\begin{proof}
If $\kappa\in (0,6]$, then  $h_{1,2} \geq 0$.
Thus, bound~\eqref{eqn::partitionfunction_positive} shows that,
for all $N \geq 1$, $\alpha \in \LP_N$, and for all $(x_1,\ldots, x_{2N}) \in \chamber_{2N}$, we have
\begin{align*}
0<\PartF_\alpha(x_1, \ldots, x_{2N}) 
\le \; & \prod_{\link{a}{b} \in \alpha} |x_{b}-x_{a}|^{-2h_{1,2}} 
=  \; \bigg( \prod_{j=1}^{N} |x_{2j-1}-x_{\alpha(2j-1)}|^{-h_{1,2}} \bigg)
\bigg( \prod_{j=1}^{N} |x_{2j}-x_{\alpha(2j)}|^{-h_{1,2}} \bigg) \\
\le  \; & \bigg( \prod_{j=1}^{N} \big( \min \big( |x_{2j-1}-x_{2j-2}| , |x_{2j-1}-x_{2j}|\big) \big)^{-h_{1,2}} \bigg)
\bigg( \prod_{j=1}^{N} 
\big( \min \big( |x_{2j}-x_{2j-1}| , |x_{2j}-x_{2j+1}| \big) \big)^{-h_{1,2}} \bigg)
,
\end{align*}
where $\alpha(i)$ denotes the pair of $i$ in $\alpha$, i.e., $\link{i}{\alpha(i)} \in \alpha$.
The claimed bound~\eqref{eq:Maleks bound} now follows by collecting the terms.
\end{proof}

It seems likely that the bound~\eqref{eq: BNNFB} also holds for the functions $\PartF_\linkpatt$ 
obtained as limits of the functions $\PartF_{\alpha(\linkpatt)}$ as explained in Section~\ref{subsec: fusion SCCG}. 
However, this property is not immediate from the construction, but would require additional arguments.

\begin{problem}
Prove property~\eqref{eq: BNNFB} for all of the functions in the collection 
$\{\PartF_\linkpatt \; | \; \linkpatt \in \LP_\multii, \, \multii \in \bZpos^n,  \, n \in \bZnn\}$.
\end{problem}

\newpage

\appendixpage

\begin{appendices}
\renewcommand{\thesection}{\Alph{section}}
\renewcommand{\thesubsection}{\arabic{subsection}}
\renewcommand{\thesubsubsection}{\Alph{subsubsection}}

\section{\label{app:Vir} Representation theory of the Virasoro algebra}

In this appendix, we summarize some aspects of the representation theory of the Virasoro algebra $\Vir$, which 
plays the role of infinitesimal symmetries in conformally invariant field theories. 
See, e.g., the book~\cite{Iohara-Koga:Representation_theory_of_Virasoro} for more background.

As a Lie algebra, $\Vir$ is spanned by 
$\{ \mathrm{L}_n \; | \; n \in \bZ \}$ and a central element $\mathrm{C}$, which satisfy the commutation relations 
\begin{align}\label{eq: Virasoro commutation relations}
[\mathrm{L}_n,\mathrm{C}] = 0 \qquad \quad \textnormal{and} \quad \qquad 
[\mathrm{L}_n,\mathrm{L}_m] = (n-m) \mathrm{L}_{n+m} 
+ \frac{1}{12} n(n^2-1) \delta_{n,-m} \mathrm{C} , \qquad \textnormal{for } n,m \in \bZ.
\end{align}
We will use the same notation $\Vir$ also for the universal enveloping algebra of the Virasoro algebra,
i.e., the associative algebra obtained by taking the quotient of polynomials in the generators of $\Vir$
modulo the relation $[X,Y] = XY - YX$.
(Because there is a one-to-one correspondence between the representations of a Lie algebra 
and its universal enveloping algebra, 
we do not have to distinguish between them here.)

Important elements of the general representation theory of Lie algebras are the highest-weight modules.
We say that a $\Vir$-module $V$ is a highest-weight module if $V = \Vir \, v_0$,
where $v_0$ is a highest-weight vector of weight 
$h \in \bC$ and central charge $c \in \bC$, that is, a vector $v_0 \in V$ satisfying 
\begin{align*}
\mathrm{L}_0 v_0 = h v_0, 
\qquad 
\mathrm{L}_n v_0 = 0 , \qquad \textnormal{for } n \geq 1, 
\qquad \textnormal{and} \qquad 
 \mathrm{C} v_0 = c v_0 .
 \end{align*}
In particular, 
for any pair $(c,h)$, there exists a unique Verma module 
$\mathrm{M}_{c,h} = \Vir/\mathrm{I}_{c,h}$
(up to isomorphism), where
$\mathrm{I}_{c,h}$ is the left ideal generated by the elements 
$\mathrm{L}_0 - h1$, $\mathrm{C} - c1$, and $\mathrm{L}_n$, for $n \geq 1$.
The Verma module $\mathrm{M}_{c,h}$ is a highest-weight module generated 
by a highest-weight vector $v_{c,h}$ of weight $h$ and central charge $c$
(given by the equivalence class of the unit $1$). 
It has a Poincar{\'e}-Birkhoff-Witt type basis
$\set{\mathrm{L}_{-n_1} \cdots \mathrm{L}_{-n_k}v_{c,h} \;|\; n_1 \geq \cdots \geq n_k > 0, \, k \in \bZnn}$
given by the action of the Virasoro generators with negative index,
ordered by applying the commutation relations~\eqref{eq: Virasoro commutation relations}.
The Verma modules $\mathrm{M}_{c,h}$ are universal in the sense that 
if $V$ is any $\Vir$-module containing a highest-weight vector $v_0$ of weight 
$h$ and central charge $c$, then there exists a canonical homomorphism 
$\varphi \colon \mathrm{M}_{c,h} \to V$ such that 
$\varphi (v_{c,h}) = v_0$. 
In other words, any highest-weight $\Vir$-module is isomorphic to a quotient of some Verma module.

%

Each Verma module $\mathrm{M}_{c,h}$ has a unique maximal proper submodule,
and the quotient of $\mathrm{M}_{c,h}$ by this submodule is 
the unique irreducible highest-weight $\Vir$-module of weight $h$ and central charge $c$.
In general, submodules of Verma modules were classified by B.~Fe{\u\i}gin and D.~Fuchs
\cite{Feigin-Fuchs:Invariant_skew-symmetric_differential_operators_on_the_line_and_Verma_modules_over_Virasoro,
Feigin-Fuchs:Verma_modules_over_Virasoro_book,
Feigin-Fuchs:Representations_of_Virasoro}, who showed that
every non-trivial submodule of a Verma module $\mathrm{M}_{c,h}$ is generated by some singular vectors
--- a vector $v \in \mathrm{M}_{c,h} \setminus \set{0}$ is said to be
singular at level $\ell \in \bZpos$ if it has the properties
\begin{align}\label{eq: definition of singular vector}
\mathrm{L}_0 v  =  (h + \ell) v \qquad \textnormal{and} \qquad 
\mathrm{L}_n v  =  0 , \qquad \textnormal{for } n \geq 1.
\end{align}
Note that the $\mathrm{L}_0$-eigenvalue of a basis vector 
$v = \mathrm{L}_{-n_1} \cdots \mathrm{L}_{-n_k} v_{c,h} \in \mathrm{M}_{c,h}$
can be calculated using the commutation 
relations~\eqref{eq: Virasoro commutation relations}: we have
$\mathrm{L}_0 v = (h + \sum_{i=1}^k n_i) v = (h + \ell) v$.
The number $\ell := \sum_{i=1}^k n_i$ is called the level of the vector $v$.

In particular, Fe{\u\i}gin and Fuchs found a characterization for the existence of singular vectors
and thus for the irreducibility of $\mathrm{M}_{c,h}$.
Indeed, the Verma module $\mathrm{M}_{c,h}$ is irreducible if and only if it contains no singular vectors. 
On the other hand, $\mathrm{M}_{c,h}$ contains singular vectors precisely when the numbers $(c,h)$ belong to a special class:
\begin{theorem} \label{thm:FF}
\textnormal{\cite[Proposition~\red{1.1} \& Theorem~\red{1.2}]{Feigin-Fuchs:Verma_modules_over_Virasoro_book}}
\; 
The following are equivalent:
\begin{enumerate}
\itemcolor{red}
\item The Verma module $\mathrm{M}_{c,h}$ contains a singular vector.

\item There exist $r,s \in \bZpos$, and $t \in \bC \setminus \set{0}$ such that 
\begin{align}\label{eq: Kac weights GEN}
h = h_{r,s}(t) := 
\frac{(r^2-1)}{4}  t + \frac{(s^2-1)}{4}  t^{-1} 
+ \frac{(1-rs)}{2} 
\qquad
\qquad \textnormal{and} \qquad
\qquad 
c = c(t) = 13 - 6( t + t^{-1} ).
\end{align}
In this case, the smallest such $\ell = rs$ is the lowest level at which a singular vector occurs in $\mathrm{M}_{c,h}$.
\end{enumerate}
\end{theorem}
Fe{\u\i}gin and Fuchs also obtained a fine classification of the submodule structure for the Verma modules.
The weights $h_{r,s}$ are the roots of the Kac determinant
\cite{Kac:Contravariant_form_for_infinite-dimensional_Lie_algebras_and_superalgebras,
Kac:Highest_weight_representations_of_infinite_dimensional_Lie_algebras}, 
often called Kac conformal weights. 
For instance, one can check that $\mathrm{L}_{-1} v_{c,h}$ is a singular vector at 
level one if and only if $h = h_{1,1} = 0$. As a more involved example, let us make an ansatz 
\begin{align}\label{eq: singular vector level two}
v  =  (\mathrm{L}_{-2} + a \mathrm{L}_{-1}^2) v_{c,h} 
\end{align}
for a singular vector at level two, with some $a \in \bC$.
Definition~\eqref{eq: definition of singular vector} implies that,  
in order for $v$ to be singular,
we must have $a = - \frac{3}{2(2h + 1)}$
and $h = \frac{1}{16} \big( 5 - c \pm \sqrt{(c-1)(c-25)} \big)$,
which equals $h_{1,2}$ or $h_{2,1}$ depending on the choice of sign.

In general, explicit expressions for singular vectors are hard to find --- 
one has to construct a suitable (complicated) polynomial $P$ so that the vector
$v = P(\mathrm{L}_{-1},\mathrm{L}_{-2},\ldots) v_{c,h}$
is singular. Remarkably, in the case when either $r=1$ or $s=1$,  
L.~Benoit and Y.~Saint-Aubin
found a family of such vectors~\cite{BSA:Degenerate_CFTs_and_explicit_expressions_for_some_null_vectors}: for $r = 1$ and $s \in \bZpos$, the singular vector at level $\ell = s$ 
has the formula
\begin{align}\label{eq: BSA singular vector}
\sum_{k=1}^{s}\sum_{\substack{n_{1},\ldots,n_{k}\geq1\\
n_{1}+\ldots+n_{k} = s}
}\frac{(-t)^{k-s}\,(s-1)!^{2}}{\prod_{j=1}^{k-1}(\sum_{i=1}^{j}n_{i})(\sum_{i=j+1}^{k}n_{i})}\times
\mathrm{L}_{-n_{1}}\cdots\mathrm{L}_{-n_{k}} v_{c,h_{1,s}}.
\end{align}
The case $s = 1$ and $r \in \bZpos$ is obtained by taking $t \mapsto t^{-1}$.
Later,
M.~Bauer, P.~Di Francesco, C.~Itzykson, and J.-B.~Zuber found the general singular vectors 
via a fusion procedure~\cite{BDIZ:Covariant_differential_equations_and_singular_vectors_in_Virasoro_representations}.
The formulas for these expressions, however, are not explicit.

\bigskip

As described in Section~\ref{subsec:CFT}, singular vectors give rise to degeneracies in conformal field theory
--- null fields whose correlation functions are solutions to PDEs~\eqref{eq: PDE for correlation functions} 
obtained from the Virasoro generators.
From the singular vector at level one, one obtains the null field $\mathrm{L}_{-1} \Phi_{1,1}(z)$, whose correlation functions
$F_{\index_1, \ldots, \index_n, \index} (z_1,\ldots,z_n, z) := \big\langle \Phi_{\index_1}(z_1) \cdots \Phi_{\index_n}(z_n) \Phi_{1,1}(z) \big\rangle$
satisfy the PDE 
\begin{align*}
0 = \mathcal{L}_{-1}^{(z)} \; F_{\index_1, \ldots, \index_n, \index} (z_1,\ldots,z_n, z) 
= - \sum_{i=1}^n \pder{z_i} F_{\index_1, \ldots, \index_n, \index} (z_1,\ldots,z_n, z) .
\end{align*}
Assuming that the correlation function $F$ is translation-invariant, we can replace 
$\sum_{i=1}^n \pder{z_i}$ by the single derivative $\pder{z}$, so 
\begin{align*}
\pder{z} F_{\index_1, \ldots, \index_n, \index} (z_1,\ldots,z_n, z) = 0 ,
\end{align*}
i.e., the correlation function is constant in the variable $z$ corresponding to $\Phi_{1,1}(z)$.

More interestingly, for the level two singular vectors~\eqref{eq: singular vector level two}, the corresponding null fields are
$\big(\mathrm{L}_{-2} - \frac{3}{2(2h_{1,2} + 1)} \mathrm{L}_{-1}^2 \big) \Phi_{1,2}(z)$
and
$\big(\mathrm{L}_{-2} - \frac{3}{2(2h_{2,1} + 1)} \mathrm{L}_{-1}^2 \big) \Phi_{2,1}(z)$.
In the former case, the correlation functions
$F_{\index_1, \ldots, \index_n, \index} (z_1,\ldots,z_n, z) := \big\langle \Phi_{\index_1}(z_1) \cdots \Phi_{\index_n}(z_n) \Phi_{1,2}(z) \big\rangle$
satisfy the second order PDE
\begin{align} \label{eq: singular equation level two}
\left[
- \frac{3}{2(2h_{1,2} + 1)} \left( \sum_{i=1}^n \pder{z_i} \right)^2
\hspace*{-1mm}
- 
\;
\sum_{i=1}^n 
\left(  \frac{1}{z_i - z} \pder{z_i}  -\frac{\Delta_{\index_i}}{(z_i - z)^2}\right) 
\right]
F_{\index_1, \ldots, \index_n, \index} (z_1,\ldots,z_n, z) = 0 ,
\end{align}
where $\Delta_{\index_i}$ are the conformal weights of the fields $\Phi_{\index_i}$, for $1 \leq i  \leq n$.
Assuming again translation invariance, this PDE simplifies~to
\begin{align} \label{eq: singular equation level two simplified}
\left[
- \frac{3}{2(2h_{1,2} + 1)} \pdder{z}
\; - \; \sum_{i=1}^n 
\left(  \frac{1}{z_i - z} \pder{z_i}  - \frac{\Delta_{\index_i}}{(z_i - z)^2}\right) 
\right]
F_{\index_1, \ldots, \index_n, \index} (z_1,\ldots,z_n, z) = 0 .
\end{align}

\begin{remark}
Using the parameterization $t = \kappa / 4$, 
we have $c = \frac{(3\kappa-8)(6-\kappa)}{2\kappa}$ and $h_{1,2}  = \frac{6-\kappa}{2\kappa}$,
and if we take in addition $\Delta_{\index_i} = h_{1,2}$, for all $1 \leq i  \leq n$,
then PDE~\eqref{eq: singular equation level two simplified} is equivalent to~\eqref{eq: multiple SLE PDEs} appearing in Section~\ref{sec: Multiple SLE partition functions}.
\end{remark}

In Section~\ref{sec:OPE}, we briefly discuss higher order PDEs obtained from the higher level singular vectors~\eqref{eq: BSA singular vector}.

\section{\label{app:Hao}Probabilistic construction of the pure partition functions}

In this appendix, we discuss a probabilistic approach to construct the pure partition functions $\PartF_\alpha$ inductively 
using SLE theory. The construction is rigorous for $\kappa\in (0,6]$,  as  proved recently by H.~Wu~\cite{Wu:Convergence_of_the_critical_planar_ising_interfaces_to_hypergeometric_SLE}.
For $\kappa\in (6,8)$, the same construction should also work, but to carry 
it out, one needs certain estimates of technical nature, which seem unavailable at the moment.
We review the approach of~\cite{Wu:Convergence_of_the_critical_planar_ising_interfaces_to_hypergeometric_SLE},
pointing out where the difficulties for $\kappa > 6$ emerge.

We recall that $\alpha$ denote planar pair partitions of the integers $\{1,2,\ldots,2N\}$,
that we call link patterns,
and $\LP_N$ denotes the set of all of them for fixed $N \geq 0$.
The cardinality of this set is the $N$:th Catalan number, $\LP_N = \# \Catalan_N:= \frac{1}{N+1} \binom{2N}{N}$. Also, we set
\begin{align*}
\LP := \bigsqcup_{N\geq0} \LP_N 
\qquad \qquad 
\textnormal{and} \qquad \qquad
\LP_{<N} := \bigsqcup_{M=0}^{N-1} \LP_M .
\end{align*}
Our aim is to construct a collection of functions 
\begin{align}
\{ \PartF_\alpha \; | \; \alpha \in \LP \} 
\end{align}
inductively as follows. Set $\PartF_{\emptyset} \equiv 1$. 
Let $N \geq 1$ and suppose that all of the functions $\{ \PartF_\alpha \; | \; \alpha \in \LP_{<N} \}$
have been defined and that they satisfy properties $\mathrm{(COV)}$, $\mathrm{(PDE)}$,
asymptotics property $\mathrm{(ASY)}$ for all $\alpha \in \LP_{<N}$, as well as 
the strong bounds $\mathrm{(B)}$: 
\begin{align*} 
0<\PartF_\alpha \le \prod_{\link{a}{b} \in \alpha} |x_{b}-x_{a}|^{-2h} .
\end{align*}
Via~\eqref{eq: ppf def in polygon}, we extend the definition of these functions to polygons 
$(\Omega; x_1, \ldots, x_{2n})$ with $2n < 2N$ marked points (on sufficiently regular boundary segments).
Then, for a fixed polygon $(\Omega; x_1, \ldots, x_{2N})$ and for fixed $\alpha \in \LP_N$, we set
\begin{align} \label{eq: Haos construction}
\PartF_\alpha(\Omega; x_1, \ldots, x_{2N}) 
:= H_\Omega (x_a, x_b)^h \;
\mathbb{E}_{\Omega;a,b} \big[ \PartF_{\hat{\alpha}}(\hat{\Omega}_\eta; x_1,\ldots,\hat{x}_a,\ldots,\hat{x}_b,\ldots,x_{2N}) \big] ,
\end{align}
where $\link{a}{b} \in \alpha$ is any link in $\alpha$ with $a < b$, 
the notation $\hat{x}_a$ and $\hat{x}_b$ means that these variables are omitted, and 
\begin{itemize}
\item $H_\Omega (x_a, x_b) = H_\Omega (x_b,x_a)$ is the boundary Poisson kernel in $\Omega$ between the points $x_a, x_b \in \partial \Omega$,

\item  
$h = \frac{6-\kappa}{2\kappa}$ (note that $h \geq 0$ when $\kappa\in (0,6]$ and $h < 0$ when $\kappa > 6$),

\item $\mathbb{E}_{\Omega;a,b} = \mathbb{E}_{\Omega;b,a}$ 
is the expectation under the probability measure $\mathbb{P}_{\Omega;a,b}$
of the chordal $\SLEk$ curve $\eta$ in $(\Omega; x_a, x_b)$, 
which is symmetric in the interchange of $x_a$ and $x_b$ by the celebrated reversibility property of the 
$\SLEk$ measure~\cite{Zhan:Reversibility_of_chordal_SLE, Sheffield-Miller:Imaginary_geometry3}, 

\item $\hat{\alpha} = \alpha \removeLink \link{a}{b} \in \LP_{N-1}$
is obtained from $\alpha$ by removing the link $\link{a}{b}$,

\item $\hat{\Omega}_\eta$ is the union of those connected components $D$ of 
$\Omega \setminus \eta$ that contain some of the points $\{x_1, \ldots, x_{2N}\} \setminus \{x_a, x_b\}$ in $\cl{D}$: 
\begin{align*}
\hat{\Omega}_\eta := 
\bigsqcup_{\substack{ D \textnormal{ c.c of } \Omega \setminus \eta \\ \cl{D} \cap \{x_1, \ldots, x_{2N}\} \setminus \{x_a, x_b\} \neq \emptyset }} D ,
\end{align*}

\item and $\PartF_{\hat{\alpha}}(\hat{\Omega}_\eta; \cdots)$ 
is a generalized pure partition function defined for the (random) finite union $\hat{\Omega}_\eta$
of polygons $D$ as follows:
\begin{itemize}
\item If $\eta$ partitions $\Omega$ into components such that 
the variables $x_c$ and $x_d$ corresponding to some link $\link{c}{d} \in \hat{\alpha}$
belong to different components of $\Omega$, then we set
$\PartF_{\hat{\alpha}}(\hat{\Omega}_\eta; \cdots) := 0$.
(Note that, as $\kappa \in (0,8)$, this event has probability $<1$.)

\item Otherwise, denoting by $\hat{\alpha}_D$ the 
sub-link patterns of $\hat{\alpha}$ 
associated to the components $D \subset \hat{\Omega}_\eta$, we set
\begin{align} \label{eq: componentwise pf}
\PartF_{\hat{\alpha}}(\hat{\Omega}_\eta; \cdots)
:= 
\prod_{\substack{ D \textnormal{ c.c of } \Omega \setminus \eta \\ \cl{D} \cap \{x_1, \ldots, x_{2N}\} \setminus \{x_a, x_b\} \neq \emptyset }}
\PartF_{\hat{\alpha}_D}(D; \cdots) ,
\end{align}
where for each $D$, the ellipses ``$\cdots$'' stand for those 
variables among $\{x_{1},\ldots,x_{2N}\} \setminus \{x_a,x_b\}$ which belong to  $\partial D$.
\end{itemize}
\end{itemize}
We remark that the functions $\PartF_{\hat{\alpha}_D}(D; \cdots)$ have less than $2N$ variables and have thus been defined already.

\bigskip

The first task is to show that $\PartF_\alpha$ is well-defined  via~\eqref{eq: Haos construction}, 
i.e., that the right-hand side of~\eqref{eq: Haos construction}
does not depend on the choice of the link $\link{a}{b} \in \alpha$. 
H.~Wu  proved this 
in~\cite[Lemma~\red{6.2}]{Wu:Convergence_of_the_critical_planar_ising_interfaces_to_hypergeometric_SLE} for the case of $\kappa \in (0,6]$,
and the same proof also works for $\kappa \in (6,8)$. The crucial ingredients in this proof are properties of 
the $2$-$\SLEk$ process (``hypergeometric'' SLE in~\cite{Wu:Convergence_of_the_critical_planar_ising_interfaces_to_hypergeometric_SLE}),
a probability measure on pairs $(\gamma_1, \gamma_2)$ of curves, symmetric in the exchange of the two curves
--- see Appendix~\ref{app}.

\begin{restatable}{prop}{welldefined}
\label{prop: well-defined}
\textnormal{\cite[Lemma~\red{6.2}, extended]{Wu:Convergence_of_the_critical_planar_ising_interfaces_to_hypergeometric_SLE}} 
\;
Let $\kappa\in (0,8)$.
The function $\PartF_\alpha$ is well-defined  via Equation~\eqref{eq: Haos construction}, that is, for any  two different links 
$\link{a}{b}, \link{c}{d} \in \alpha$, we have
\begin{align} \label{eq: Haos construction independent of link choice}
\begin{split}
\PartF_\alpha(\Omega; x_1, \ldots, x_{2N}) 
:= \; & H_\Omega (x_a, x_b)^h \;
\mathbb{E}_{\Omega;a,b} \big[ \PartF_{\alpha \removeLink \link{a}{b}}(\hat{\Omega}_\eta; x_1,\ldots,\hat{x}_a,\ldots,\hat{x}_b,\ldots,x_{2N}) \big]  \\
= \; & H_\Omega (x_c, x_d)^h \;
\mathbb{E}_{\Omega;c,d} \big[ \PartF_{\alpha \removeLink \link{c}{d}}(\hat{\Omega}_\eta; x_1,\ldots,\hat{x}_c,\ldots,\hat{x}_d,\ldots,x_{2N}) \big] .
\end{split}
\end{align}
\end{restatable}

We will summarize the main steps of the proof in Appendix~\ref{app}, where we also briefly discuss the $2$-$\SLEk$. 

\bigskip 

\begin{center}
\bf Properties $\mathrm{(COV)}$, $\mathrm{(PDE)}$, $\mathrm{(ASY)}$, and $\mathrm{(B)}$ for the functions $\PartF_\alpha$.
\end{center}

For the case of $\Omega = \bH$ and $x_1 < \cdots < x_{2N}$, we have
$H_\bH (x_a, x_b) = |x_b-x_a|^{-2}$, so
\begin{align} \label{eq: Haos constructionH}
\PartF_\alpha(x_1, \ldots, x_{2N}) 
:= \PartF_\alpha(\bH;x_1, \ldots, x_{2N}) 
:=
|x_b-x_a|^{-2h} \;
\mathbb{E}_{\bH;a,b} \big[ \PartF_{\hat{\alpha}}(\hat{\bH}_\eta; x_1,\ldots,\hat{x}_a,\ldots,\hat{x}_b,\ldots,x_{2N}) \big] .
\end{align}
We aim to prove the following properties for this function:
\begin{enumerate}
\itemcolor{red}
\item
The function $\PartF_\alpha$ satisfies the M\"obius covariance~\eqref{eq: multiple SLE Mobius covariance} in property $\mathrm{(COV)}$.
[See Lemma~\ref{lem: COV}.]

\item 
The function $\PartF_\alpha \colon \chamber_{2N} \to \bRpos$ 
is smooth and it solves the PDE system~\eqref{eq: multiple SLE PDEs} in property $\mathrm{(PDE)}$.
[See Lemma~\ref{lem: PDE}.]

\item 
The collection $\{ \PartF_\alpha \; | \; \alpha \in \LP_{<N+1} \}$
satisfies the recursive asymptotics in~\eqref{eq: multiple SLE asymptotics} in property $\mathrm{(ASY)}$.
[See Lemma~\ref{lem: ASY}.]

\item 
The function $\PartF_\alpha$ satisfies the strong bound~\eqref{eqn::partitionfunction_positive} in property $\mathrm{(B)}$.
[See Lemma~\ref{lem: bound}.]

\end{enumerate}

\noindent
The first property~$\mathrm{(COV)}$ is immediate from construction:

\begin{lem} \label{lem: COV} 
\textnormal{\cite[Lemma~\red{6.5}, extended]{Wu:Convergence_of_the_critical_planar_ising_interfaces_to_hypergeometric_SLE}} 
\;
Let $\kappa \in (0,8)$. 
The function $\PartF_\alpha$ defined in~\eqref{eq: Haos constructionH} 
satisfies the M\"obius covariance~\eqref{eq: multiple SLE Mobius covariance} in property $\mathrm{(COV)}$.
\end{lem}

\begin{proof}
This follows from the construction of $\PartF_\alpha$ in~\eqref{eq: Haos constructionH}, 
conformal invariance of the chordal $\SLEk$ measure $\mathbb{P}_{\bH;a,b}$, 
and the conformal covariance property
$H_\bH (x_a, x_b) = \Mob'(x_a) \Mob'(x_b) \; H_\bH (\Mob(x_a), \Mob(x_b))$
of the boundary Poisson kernel.
\end{proof}

When $\kappa \in (0,6]$,
property $\mathrm{(ASY)}$ is also not difficult to show, 
by virtue of Proposition~\ref{prop: well-defined}, which allows us to choose the link $\link{a}{b}$
in the construction~\eqref{eq: Haos constructionH} of $\PartF_\alpha$ freely.
Before giving the proof, we list and recall some notation, to be frequently used throughout.
Fix $\link{a}{b} \in \alpha$ for the construction of $\PartF_\alpha$, with $a < b$, and denote $\hat{\alpha} = \alpha \removeLink \link{a}{b}$. 
Denote also by
\begin{itemize}
\item  $\eta \sim \mathbb{P}_{\Omega;a,b}$ the chordal $\SLEk$ in $(\bH; x_a, x_b)$,

\item $\hat{\bH}_\eta$ the union of the connected components of $\bH \setminus \eta$ 
containing some of the points $\{x_1, \ldots, x_{2N}\} \setminus \{x_a, x_b\}$ on the boundary,

\item $\mathcal{E}_\eta = \mathcal{E}_{\eta;a,b}^{\alpha}(\bH; x_1, \ldots, x_{2N})$ the event that $\eta$ does not partition $\bH$ into components
where some variables corresponding to a link in $\alpha$ would belong to different components
(note that on the complement of this event, $\PartF_{\hat{\alpha}}(\hat{\bH}_\eta; \cdots)$ is zero), and

\item on the event $\mathcal{E}_\eta$, for each link $\link{c}{d} \in \alpha$ such that $\link{c}{d} \neq \link{a}{b}$,  let
$\smash{H_{\hat{\bH}_\eta} (x_{c}, x_{d})}$ denote the boundary Poisson kernel in the connected component of 
$\smash{\hat{\bH}_\eta}$ that has $x_{c}$ and $x_{d}$ on its boundary. 
\end{itemize}

\begin{lem} \label{lem: ASY} 
\textnormal{\cite[Lemma~\red{6.6}, extended]{Wu:Convergence_of_the_critical_planar_ising_interfaces_to_hypergeometric_SLE}} 
\;
Let $\kappa \in (0,6]$. The collection $\{ \PartF_\alpha \; | \; \alpha \in \LP_{<N+1} \}$
satisfies the recursive asymptotics in~\eqref{eq: multiple SLE asymptotics} in property $\mathrm{(ASY)}$.
\end{lem}

\begin{proof}
If $N=1$, the claim is clear.
For the case of $N=2$, asserted asymptotics properties~\eqref{eq: multiple SLE asymptotics} can be checked by hand:
Equations~\eqref{eq: hg formulas for 4p fctions1}--\eqref{eq: hg formulas for 4p fctions2} state explicit formulas for 
the two functions 
$\PartF_{\vcenter{\hbox{\includegraphics[scale=0.2]{figures/link-2.pdf}}}}$
and
$\PartF_{\vcenter{\hbox{\includegraphics[scale=0.2]{figures/link-1.pdf}}}}$
in terms of a hypergeometric function. Investigation of these formulas shows~\eqref{eq: multiple SLE asymptotics}
for $\{ \PartF_\beta \; | \; \beta \in \LP_{<2} \} 
= \{\PartF_{\vcenter{\hbox{\includegraphics[scale=0.2]{figures/link-0.pdf}}}} , 
\PartF_{\vcenter{\hbox{\includegraphics[scale=0.2]{figures/link-2.pdf}}}} ,
\PartF_{\vcenter{\hbox{\includegraphics[scale=0.2]{figures/link-1.pdf}}}} \}$
(and for all $\kappa \in (0,8)$).

Hence, we assume that $N \geq 3$.
By our induction hypothesis, the collection $\{ \PartF_\beta \; | \; \beta \in \LP_{<N} \}$
satisfies the asymptotics~\eqref{eq: multiple SLE asymptotics} in property $\mathrm{(ASY)}$.
Fix $\alpha \in \LP_N$, $j \in \{1,\ldots,2N-1\}$, and $\xi \in (x_{j-1}, x_{j+2})$. 
Choose a link $\link{a}{b} \in \alpha$ (with $a < b$) such that $\{a,b\} \cap \{j,j+1\} = \emptyset$.
Then by definition~\eqref{eq: Haos constructionH}, we have
\begin{align} 
\nonumber
\; & \lim_{x_j , x_{j+1} \to \xi} 
\frac{\PartF_\alpha(x_1 , \ldots , x_{2N})}{(x_{j+1} - x_j)^{-2h}} \\
\nonumber
= \; & \lim_{x_j , x_{j+1} \to \xi} 
\left( \frac{x_b-x_a}{x_{j+1} - x_j} \right)^{-2h} \;
\mathbb{E}_{\bH;a,b} \big[ \PartF_{\alpha \removeLink \link{a}{b} }(\hat{\bH}_\eta; x_1,\ldots,\hat{x}_a,\ldots,\hat{x}_b,\ldots,x_{2N}) \big] \\
= \; & (x_b-x_a)^{-2h} \; 
 \lim_{x_j , x_{j+1} \to \xi}  \; \mathbb{E}_{\bH;a,b} \left[ \one_{\mathcal{E}_\eta}
\left( \frac{H_{\hat{\bH}_\eta} (x_{j}, x_{j+1})}{H_\bH (x_{j}, x_{j+1})} \right)^h
\frac{\PartF_{\alpha \removeLink \link{a}{b}}(\hat{\bH}_\eta; x_1,\ldots,\hat{x}_a,\ldots,\hat{x}_b,\ldots,x_{2N})}{H_{\hat{\bH}_\eta} (x_{j}, x_{j+1})^h} \right] .
\label{eq: limit of pf}
\end{align}
Now, asymptotics property~\eqref{eq: multiple SLE asymptotics GEN} 
for the already constructed functions $\PartF_{\alpha \removeLink \link{a}{b}}$ combined with
Lemma~\ref{lem: ratio poisson} from Appendix~\ref{app} implies that,
for the expression inside the expectation $\mathbb{E}_{\bH;a,b}$ in~\eqref{eq: limit of pf}, we have
\begin{align*}
 \lim_{x_j , x_{j+1} \to \xi} 
\one_{\mathcal{E}_\eta}
\left( \frac{H_{\hat{\bH}_\eta} (x_{j}, x_{j+1})}{H_\bH (x_{j}, x_{j+1})} \right)^h
\frac{\PartF_{\alpha \removeLink \link{a}{b}}(\hat{\bH}_\eta; x_1,\ldots,\hat{x}_a,\ldots,\hat{x}_b,\ldots,x_{2N})}{H_{\hat{\bH}_\eta} (x_{j}, x_{j+1})^h}
= \; & \begin{cases}
0 , \quad &
    \textnormal{if } \link{j}{j+1} \notin \alpha , \\
\PartF_{\alpha \removeLink (\link{a}{b} \cup \link{j}{j+1})} (\hat{\bH}_\eta; \cdots) , &
    \textnormal{if } \link{j}{j+1} \in \alpha ,
\end{cases} 
\end{align*}
almost surely. Noticing that by definition~\eqref{eq: Haos constructionH}, we have
\begin{align*}
(x_b-x_a)^{-2h} \; \mathbb{E}_{\bH;a,b} \left[ 
\PartF_{\alpha \removeLink (\link{a}{b} \cup \link{j}{j+1})} (\hat{\bH}_\eta; \cdots)
\right] 
= \PartF_{\alpha \removeLink \link{j}{j+1}} (x_{1},\ldots,x_{j-1},x_{j+2},\ldots,x_{2N}) ,
\end{align*}
we see that in order to prove  asserted property~\eqref{eq: multiple SLE asymptotics} for $\PartF_\alpha$,
we only need to prove that the limit and the expectation in~\eqref{eq: limit of pf} can be exchanged.
This is guaranteed if the expression inside the expectation $\mathbb{E}_{\bH;a,b}$ in~\eqref{eq: limit of pf} is uniformly integrable.
Indeed, using the strong bound~\eqref{eqn::partitionfunction_positive_in_polygon} for 
$\PartF_{\alpha \removeLink \link{a}{b}}$, 
the monotonicity property $H_{\hat{\bH}_\eta} (x_{j}, x_{j+1}) \leq H_\bH (x_{j}, x_{j+1}) = (x_{j+1} - x_{j})^{-2}$ for $\hat{\bH}_\eta \subset \bH$,
and the fact that $h \geq 0$ (which only holds when $\kappa \in (0,6]$), 
we obtain the following bound uniformly in $\eta$:
\begin{align}
\nonumber
0 \leq 
\one_{\mathcal{E}_\eta} 
\left( \frac{H_{\hat{\bH}_\eta} (x_{j}, x_{j+1})}{H_\bH (x_{j}, x_{j+1})} \right)^h
\frac{\PartF_{\alpha \removeLink \link{a}{b}}(\hat{\bH}_\eta; x_1,\ldots,\hat{x}_a,\ldots,\hat{x}_b,\ldots,x_{2N})}{H_{\hat{\bH}_\eta} (x_{j}, x_{j+1})^h} 
\leq \; &
(x_{j+1} - x_j)^{2h} \;\prod_{\substack{\link{c}{d} \in \alpha , \\ \link{c}{d} \neq \link{a}{b} }} H_{\hat{\bH}_\eta}(x_{c},x_{d})^{h} \\
\leq \; &
(x_{j+1} - x_j)^{2h} \;
\prod_{\substack{\link{c}{d} \in \alpha , \\ \link{c}{d} \neq \link{a}{b} }} (x_{d} - x_{c})^{-2h} .
\label{eq: bound for asy}
\end{align}
First, suppose that $\link{j}{j+1} \in \alpha$. Then, the right-hand side of~\eqref{eq: bound for asy} 
is independent of $x_j$ and $x_{j+1}$,
\begin{align*}
(x_{j+1} - x_j)^{2h} \;
\prod_{\substack{\link{c}{d} \in \alpha , \\ \link{c}{d} \neq \link{a}{b} }} (x_{d} - x_{c})^{-2h} =
\prod_{\substack{\link{c}{d} \in \alpha , \\ \link{c}{d} \neq \link{a}{b} \\ c,d \neq j,j+1}} (x_{d} - x_{c})^{-2h} ,
\end{align*}
so it is uniformly bounded in the limit $x_j , x_{j+1} \to \xi$.
This justifies the exchange of the limit and the expectation in~\eqref{eq: limit of pf} when $\link{j}{j+1} \in \alpha$.
Second, suppose that $\link{j}{j+1} \notin \alpha$. Then, the right-hand side of~\eqref{eq: bound for asy} equals 
\begin{align*}
\left(\frac{(x_j-x_{\alpha(j)})(x_{j+1}-x_{\alpha(j+1)})}{x_{j+1} - x_j} \right)^{-2h}  
\prod_{\substack{\link{c}{d} \in \alpha , \\ \link{c}{d} \neq \link{a}{b} \\ c,d \neq j,j+1}} (x_{d} - x_{c})^{-2h} ,
\end{align*}
where $\alpha(i)$ denotes the pair of $i$ in $\alpha$, i.e., $\link{i}{\alpha(i)} \in \alpha$, for $i=j,j+1$.
This expression tends to zero in the limit $x_j , x_{j+1} \to \xi$, justifying
the exchange of the limit and the expectation in~\eqref{eq: limit of pf} when $\link{j}{j+1} \notin \alpha$.
This concludes the proof.
\end{proof}

Concerning the case of $\kappa \in (6,8)$, we make two remarks. 
First, if $N \in \{1,2\}$, then the known explicit formulas for the pure partition functions 
immediately imply asymptotics property~\eqref{eq: multiple SLE asymptotics} in $\mathrm{(ASY)}$.
Second, if $N \geq 3$, then the proof of Lemma~\ref{lem: ASY} would carry through for $\kappa \in (6,8)$
provided that the expression inside the expectation $\mathbb{E}_{\bH;a,b}$ in~\eqref{eq: limit of pf} was uniformly integrable.
However, to prove this, additional technical work would be needed --- because we have $h < 0$ when $\kappa \in (6,8)$,
we cannot apply the bound in~\eqref{eq: bound for asy}. 
Currently, we are not aware of any proof 
of asymptotics property~\eqref{eq: multiple SLE asymptotics} in $\mathrm{(ASY)}$
for $\kappa \in (6,8)$ and $N \geq 3$.

\bigskip

Next, 
concerning property $\mathrm{(PDE)}$, thanks to Proposition~\ref{prop: well-defined}
it suffices to only verify the two PDEs with $i = a,b$ in~\eqref{eq: multiple SLE PDEs}, 
and the other PDEs then follow by symmetry. Also, reversibility of the chordal $\SLEk$
implies that it is actually enough to check only the PDE associated with $i=a$. 
This PDE can be verified using diffusion theory and the fact that the PDE is 
hypoelliptic~\cite[Theorem~\red{6}]{Dubedat:SLE_and_Virasoro_representations_localizationA}. 
The main difficulty is to show that the function $\PartF_\alpha$, defined in terms of an expectation~\eqref{eq: Haos constructionH},
is indeed twice continuously differentiable.
The function $\PartF_\alpha$ appears naturally in a certain local martingale, 
but It\^o's formula cannot be used directly because of lack of a priori regularity.
See also~\cite{Lawler-Jahangoshahi:On_smoothness_of_partition_function_for_multiple_SLEs} for the case of $\kappa \in (0,4)$.

\begin{lem} \label{lem: PDE}
\textnormal{\cite[Lemmas~\red{6.3} \& \red{6.4}, extended]{Wu:Convergence_of_the_critical_planar_ising_interfaces_to_hypergeometric_SLE}} 
\;
Let $\kappa \in (0,8)$. 
The function $\PartF_\alpha \colon \chamber_{2N} \to \bRpos$ defined in~\eqref{eq: Haos constructionH}
is smooth and it solves  the PDE system~\eqref{eq: multiple SLE PDEs} in property $\mathrm{(PDE)}$.
\end{lem}

\begin{proof}
For notational simplicity (and without losing generality), we assume that $\link{a}{b} = \link{1}{2}$.
We give a sketch of the proof.
\begin{itemize}
\item 
{\bf \textit{$\PartF_\alpha$ solves the PDE in~\eqref{eq: multiple SLE PDEs} with $i = 1$}}:

As in the construction of $\PartF_\alpha$, let $\eta$ be the chordal $\SLEk$ from $x_1$ to $x_2$, and let
$(W_t)_{t \geq 0}$ be its Loewner driving function and $(g_t)_{t \geq 0}$ the corresponding 
solution to the Loewner equation~\eqref{eq: Loewner equation}.
Then, up to the first time when $\eta$ hits the boundary $\bR = \partial \bH$,
thanks to the domain Markov property of the chordal $\SLEk$ and Equation~\eqref{eq: ppf def in polygon},
the following conditional expectation is a local martingale for $\eta$:
\begin{align*} 
M_t := \; & \mathbb{E}_{\bH;1,2} \big[ \PartF_{\hat{\alpha}}(\hat{\bH}_\eta; x_3,x_4,\ldots,x_{2N}) 
\; \big| \; \eta[0,t] \big]  \\
= \; & \prod_{j = 3}^{2N} g_t'(x_j)^h \times
\mathbb{E}_{\bH;1,2} \big[ \PartF_{\hat{\alpha}} \big(g_t(\hat{\bH}_\eta); g_t(x_3), g_t(x_4), \ldots,g_t(x_{2N})\big) 
\; \big| \; \eta[0,t] \big] \\
= \; & \prod_{j = 3}^{2N} g_t'(x_j)^h \times (g_t(x_2) - W_t)^{2h} \times 
\PartF_\alpha \big(W_t, g_t(x_2), g_t(x_3),\ldots, g_t(x_{2N})\big) 
= F(X_t) ,
\end{align*}
where $X_t = (W_t, g_t(x_2), g_t(x_3),\ldots, g_t(x_{2N}), g_t'(x_3), \ldots, g_t'(x_{2N}))$ is an It\^o process and
\begin{align*} 
F(x_1, \ldots, x_{2N}, y_3, \ldots, y_{2N}) 
:=  \prod_{j = 3}^{2N} y_j^h \times (x_2-x_1)^{2h}  \times \PartF_{\alpha}(x_1, \ldots, x_{2N}) 
\end{align*}
is a continuous function of $(x_1, \ldots, x_{2N}, y_3, \ldots, y_{2N}) \in \chamber_{2N} \times \bR^{2N-2}$.
One can check that the local martingale property of $M$ 
implies that $\PartF_\alpha$ is smooth and solves~\eqref{eq: multiple SLE PDEs} with $i = 1$,
see~\cite[proof of Lemma~\red{4.4}]{Peltola-Wu:Global_and_local_multiple_SLEs_and_connection_probabilities_for_level_lines_of_GFF} 
and~\cite[Theorem~\red{6}]{Dubedat:SLE_and_Virasoro_representations_localizationA}.  
The key fact here is that the PDE~\eqref{eq: multiple SLE PDEs} is 
hypoelliptic~\cite{Dubedat:SLE_and_Virasoro_representations_localizationA}, so any distributional solution to it is smooth.

For $\kappa \in (0,4)$, G.~Lawler and M.~Jahangoshahi~\cite{Lawler-Jahangoshahi:On_smoothness_of_partition_function_for_multiple_SLEs}
provided a proof for the smoothness of $\PartF_\alpha$ by traditional SLE techniques,
without using hypoellipticity of the PDE system. 
It then follows easily from It\^o calculus that 
$\PartF_\alpha$ solves the PDE in~\eqref{eq: multiple SLE PDEs} with $i = 1$.
Unfortunately, the current result~\cite{Lawler-Jahangoshahi:On_smoothness_of_partition_function_for_multiple_SLEs}  
is not strong enough to deal with the case of $\kappa \in [4,8)$.

\item 
{\bf \textit{$\PartF_\alpha$ solves the PDE in~\eqref{eq: multiple SLE PDEs} with $i = 2$}}:

This follows from the above argument by the reversibility of the chordal $\SLEk$ curve $\eta$.
We emphasize that, despite being natural, the reversibility is very non-trivial:
it was proved for $\kappa \in (0,4]$ by D.~Zhan in his celebrated work~\cite{Zhan:Reversibility_of_chordal_SLE} 
using a coupling  of the ``past'' and ``future'' of the $\SLEk$, 
and for $\kappa \in (4,8)$ by J.~Miller and S.~Sheffield in the even more striking work~\cite{Sheffield-Miller:Imaginary_geometry3},
which relies on the theory of ``imaginary geometry'' developed by the authors,
coupling the $\SLEk$ curve as a flow line with the Gaussian free field.
We are not aware of a proof for the PDE~\eqref{eq: multiple SLE PDEs} with $i = 2$ avoiding the reversibility.

\item 
{\bf \textit{$\PartF_\alpha$ solves the PDEs in~\eqref{eq: multiple SLE PDEs} with $i \geq 3$}}:
This follows by the symmetry of the definition~\eqref{eq: Haos constructionH} of $\PartF_\alpha$ stated in 
Equation~\eqref{eq: Haos construction independent of link choice} in Proposition~\ref{prop: well-defined}:
using the above argument for the function $\PartF_\alpha$ written in~\eqref{eq: Haos constructionH}
with some other link $\link{a}{b} \neq \link{1}{2}$, we exhaust all of the indices $i \geq 3$.
Note that in order to prove this property, one uses strong facts about the $2$-$\SLEk$ probability measure, 
as we will discuss in Appendix~\ref{app}.
\end{itemize}
\end{proof}

To finish, we prove property $\mathrm{(B)}$ for $\PartF_\alpha$.
When $\kappa \in (0,6]$, this property is easy to prove,
whereas for $\kappa \in (6,8)$ it seems to be very difficult, because we have $h < 0$ in that case.
Currently, we are not aware of any proof of~$\mathrm{(B)}$ for the case of $\kappa \in (6,8)$.

\begin{lem} \label{lem: bound} 
\textnormal{\cite[Lemma~\red{6.7}, extended]{Wu:Convergence_of_the_critical_planar_ising_interfaces_to_hypergeometric_SLE}} 
\;
Let $\kappa \in (0,6]$. 
The function $\PartF_\alpha$ defined in~\eqref{eq: Haos constructionH}
satisfies the strong bound~\eqref{eqn::partitionfunction_positive} in property $\mathrm{(B)}$.
\end{lem}

\begin{proof}
First, the positivity of $\PartF_\alpha$ follows from its construction, since 
the probability for the chordal $\SLEk$ curve in $(\bH; x_a, x_b)$ to not partition $\bH$ 
into components where some variables $x_c$, $x_d$ corresponding to a link $\link{c}{d} \in \hat{\alpha}$ 
would belong to different components is positive. 
Second, the definition~\eqref{eq: Haos constructionH} of $\PartF_\alpha$
and property $\mathrm{(B)}$ for the already constructed functions in $\PartF_{\hat{\alpha}}$ give
\begin{align*}
\PartF_\alpha(x_1, \ldots, x_{2N}) 
:= \; & |x_b-x_a|^{-2h} \;
\mathbb{E}_{\bH;a,b} \big[ \PartF_{\hat{\alpha}}(\hat{\bH}_\eta; x_1,\ldots,\hat{x}_a,\ldots,\hat{x}_b,\ldots,x_{2N}) \big] \\
\leq \; & 
\prod_{\link{c}{d} \in \alpha} |x_{d}-x_{c}|^{-2h} \;
\mathbb{E}_{\bH;a,b} \left[ \one_{\mathcal{E}_\eta}
\prod_{\substack{\link{c}{d} \in \alpha , \\ \link{c}{d} \neq \link{a}{b} }}
\left( \frac{H_{\hat{\bH}_\eta} (x_{c}, x_{d})}{H_\bH (x_{c}, x_{d})} \right)^h
 \right] .
\end{align*}
Because $\kappa \in (0,6]$, we have $h \geq 0$, so the monotonicity property 
$H_{\hat{\bH}_\eta} (x_{c}, x_{d}) \leq H_\bH (x_{c}, x_{d})$ for $\hat{\bH}_\eta \subset \bH$ implies
the asserted bound~\eqref{eqn::partitionfunction_positive} in~$\mathrm{(B)}$: 
\begin{align*}
\left( \frac{H_{\hat{\bH}_\eta} (x_{c}, x_{d})}{H_\bH (x_{c}, x_{d})} \right)^h \leq 1
\qquad \qquad \Longrightarrow \qquad \qquad
\PartF_\alpha(x_1, \ldots, x_{2N}) 
\leq \prod_{\link{c}{d} \in \alpha} |x_{d}-x_{c}|^{-2h} .
\end{align*}
\end{proof}

We note that when $\kappa \in (6,8)$, the above argument does not work, since $h < 0$.
However, if one could prove, e.g., that
\begin{align*}
\mathbb{E}_{\bH;a,b} \left[ \one_{\mathcal{E}_\eta}
\prod_{\substack{\link{c}{d} \in \alpha , \\ \link{c}{d} \neq \link{a}{b} }}
\left( \frac{H_{\hat{\bH}_\eta} (x_{c}, x_{d})}{H_\bH (x_{c}, x_{d})} \right)^h
 \right] \leq 1 ,
\end{align*}
then Lemma~\ref{lem: bound} would follow.
Arguments similar to the ones used in~\cite{Lawler:Partition_functions_loop_measure_and_versions_of_SLE,
Lawler-Jahangoshahi:On_smoothness_of_partition_function_for_multiple_SLEs}
might be helpful for this.
Similar arguments are probably needed for extending the proof of Lemma~\ref{lem: ASY}, i.e., showing uniform integrability
in~\eqref{eq: limit of pf}.

\section{\label{app}Proof of Proposition~\ref{prop: well-defined} and a technical lemma}

In this appendix, we 
summarize the main steps of the proof 
of~\cite[Lemma~\red{6.2}]{Wu:Convergence_of_the_critical_planar_ising_interfaces_to_hypergeometric_SLE},
which holds for all $\kappa \in (0,8)$.
We also prove a technical result, Lemma~\ref{lem: ratio poisson}, that was used  in the proof of Lemma~\ref{lem: ASY}.

\bigskip

It was proved in~\cite{Miller-Werner:Connection_probabilities_for_conformal_loop_ensembles} 
(see also~\cite{Sheffield-Miller:Imaginary_geometry2, Sheffield-Miller:Imaginary_geometry3, BPW:On_the_uniqueness_of_global_multiple_SLEs}) 
that for all $\kappa \in (0,8)$, 
given a polygon $(\Omega; x, y, z, w)$, there exists a unique probability measure on
pairs of curves $(\gamma_1, \gamma_2)$ such that
\begin{itemize}
\item $\gamma_1$ is a curve connecting $x$ and $y$ in $\overline{\Omega}$ and $\gamma_2$ is a curve connecting $z$ and $w$ in $\overline{\Omega}$,
and these two curves do not cross (however, they can touch when $\kappa \in (4,8)$),

\item given $\gamma_1$, the conditional law of $\gamma_2$ is that of the chordal $\SLEk$ in $(\hat{\Omega}_{z,w}; z,w)$,
that is, in the connected component $\hat{\Omega}_{z,w}$ of $\Omega \setminus \gamma_1$ having $z$ and $w$ on its boundary, and

\item given $\gamma_2$, the conditional law of $\gamma_1$ is that of the chordal $\SLEk$ in $(\hat{\Omega}_{x,y}; x,y)$,
that is, in the connected component $\hat{\Omega}_{x,y}$ of $\Omega \setminus \gamma_2$ having $x$ and $y$ on its boundary.
\end{itemize}
We call this probability measure the $2$-$\SLEk$ in $(\Omega; x, y, z, w)$. 
Importantly, it is completely symmetric in the two curves $(\gamma_1, \gamma_2)$.
The marginal laws of the curves $\gamma_1$ and $\gamma_2$ are also known, and they are given by the so-called
``hypergeometric'' SLE (``$\hSLEk$'')  --- this is nothing but a chordal $\SLEk$ variant with partition function 
$\PartF_{\{\link{1}{2},\link{3}{4}\}}= \PartF_{\vcenter{\hbox{\includegraphics[scale=0.2]{figures/link-1.pdf}}}}$
given in Equation~\eqref{eq: hg formulas for 4p fctions2}.

\welldefined*
\begin{proof}
The asserted property is trivial for 
$\PartF_{\vcenter{\hbox{\includegraphics[scale=0.2]{figures/link-0.pdf}}}}(x_1,x_2) 
= H_\Omega (x_{1}, x_{2})^h$ with $N=1$.
Also, when $N=2$, the two functions $\PartF_{\vcenter{\hbox{\includegraphics[scale=0.2]{figures/link-1.pdf}}}}$ 
and $\PartF_{\vcenter{\hbox{\includegraphics[scale=0.2]{figures/link-2.pdf}}}}$ 
are explicit and related to each other via a cyclic permutation of variables, 
see~\eqref{eq: hg formulas for 4p fctions1}--\eqref{eq: hg formulas for 4p fctions2}.
It is obvious from these formulas that the choice of link in~\eqref{eq: Haos construction} does not matter.

Now, we proceed inductively on $N \geq 3$, assuming that the claimed property~\eqref{eq: Haos construction independent of link choice} 
has already been proven for the collection $\{ \PartF_\alpha \; | \; \alpha \in \LP_{<N} \}$.
By rotational invariance, without loss of generality, we may assume that $a < b < c < d$.
To facilitate the notation, we denote $\hat{\alpha}_1 := \alpha \removeLink \link{a}{b}$
and $\hat{\alpha}_2 := \alpha \removeLink \link{c}{d}$, and we let $\eta_1$ and $\eta_2$ be independent chordal $\SLEk$ curves 
in $(\Omega; x_a, x_b)$ and $(\Omega; x_c, x_d)$, respectively. Also, we let
$\mathbb{P}_{\Omega;a,b,c,d}$ denote the $2$-$\SLEk$ probability measure on the polygon $(\Omega; x_a,x_b,x_c,x_d)$
for pairs of curves $(\gamma_1,\gamma_2)$. (We remark that if $\kappa \in (0,4]$, then the joint law of $(\eta_1,\eta_2)$ is absolutely 
continuous with respect to $(\gamma_1,\gamma_2)$, but when $\kappa \in (4,8)$, it is singular).

We also let $\mathcal{E}_1$ be the event that $\eta_1$ does not partition $\Omega$ into components
where some variables corresponding to a link in $\alpha_1$ would belong to different components.
On the event $\mathcal{E}_1$, in definition~\eqref{eq: Haos construction}--\eqref{eq: componentwise pf}, 
we let $D_{c,d}$ be the c.c of $\hat{\Omega}_{\eta_1}$
having $x_c$ and $x_d$ on its boundary. Then, by the induction hypothesis, with
$\tilde{\eta}$ denoting the chordal $\SLEk$ in $(D_{c,d}; x_c, x_d)$, we have
\begin{align*}
\PartF_{\hat{\alpha}_{D_{c,d}}}(D_{c,d}; \cdots)
= \; & H_{D_{c,d}} (x_c, x_d)^h \;
\mathbb{E}_{D_{c,d};c,d} \big[ \PartF_{\hat{\alpha}_{D_{c,d}} \removeLink \link{c}{d}} \big(\hat{D}_{c,d}(\tilde{\eta}); \cdots\big) \big] ,
\end{align*}
where we abuse notation, trusting that no confusion arises.
Here, $\hat{D}_{c,d}(\tilde{\eta})$ is the random finite union of 
those connected components of $D_{c,d} \setminus \tilde{\eta}$ that 
have some of the marked points $\{x_1, \ldots, x_{2N}\} \setminus \{x_a,x_b,x_c,x_d\}$ on their boundary.
Now, we have
\begin{align*}
 \; &
\PartF_\alpha(\Omega; x_1, \ldots, x_{2N}) \\
:= \; & H_\Omega (x_a, x_b)^h \;
\mathbb{E}_{\Omega;a,b} \big[ \PartF_{\hat{\alpha}_1}(\hat{\Omega}_{\eta_1}; \cdots) \big]  \\
= \; & H_\Omega (x_a, x_b)^h \;
\mathbb{E}_{\Omega;a,b} \left[ \PartF_{\hat{\alpha}_{D_{c,d}}}(D_{c,d}; \cdots)
\prod_{\substack{ D \textnormal{ c.c of } \Omega \setminus \eta_1 , \;  D \neq D_{c,d} \\ \cl{D} \cap \{x_1, \ldots, x_{2N}\} \setminus \{x_a, x_b\} \neq \emptyset }}
\PartF_{\hat{\alpha}_D}(D; \cdots) \right] \\
= \; & H_\Omega (x_a, x_b)^h \;
\mathbb{E}_{\Omega;a,b} \left[ 
\one_{\mathcal{E}_1} H_{D_{c,d}} (x_c, x_d)^h \;
\mathbb{E}_{D_{c,d};c,d} \left[ \PartF_{\hat{\alpha}_{D_{c,d}} \removeLink \link{c}{d}}(\hat{D}_{c,d}(\tilde{\eta}); \cdots) \right] 
\prod_{\substack{ D \textnormal{ c.c of } \Omega \setminus \eta_1 , \;  D \neq D_{c,d} \\ \cl{D} \cap \{x_1, \ldots, x_{2N}\} \setminus \{x_a, x_b\} \neq \emptyset }}
\PartF_{\hat{\alpha}_D}(D; \cdots) \right] ,
\end{align*}
where the indicator function $\one_{\mathcal{E}_1}$ just accounts for the fact that $\smash{\PartF_{\hat{\alpha}_{D_{c,d}}} = 0}$ on the complementary event $\mathcal{E}_1^c$
--- in particular, there is no problem with the seemingly troublesome situation that $H_{D_{c,d}} (x_c, x_d) = 0$ on the event $\mathcal{E}_1^c$.

Now, we can conclude after a few observations:
\begin{itemize}
\item By~\cite[Proposition~\red{3.5}]{Wu:Convergence_of_the_critical_planar_ising_interfaces_to_hypergeometric_SLE} 
(see also~\cite[Section~\red{8}]{BBK:Multiple_SLEs_and_statistical_mechanics_martingales},
\cite[Section~\red{4}]{Dubedat:Euler_integrals_for_commuting_SLEs}, and 
\cite[Section~\red{4}]{Miller-Werner:Connection_probabilities_for_conformal_loop_ensembles}), 
the law of $\eta_1$ weighted by $\one_{\mathcal{E}_1} H_{D_{c,d}} (x_c, x_d)^h$ 
is equal to the marginal law of the curve $\gamma_1$ connecting $x_a$ and $x_b$ 
in the $2$-$\SLEk$ process $(\gamma_1,\gamma_2) \sim \mathbb{P}_{\Omega;a,b,c,d}$. 
In~\cite{Wu:Convergence_of_the_critical_planar_ising_interfaces_to_hypergeometric_SLE}, this curve was called the $\hSLEk$ in $(\Omega; x_a, x_b)$ with marked points $(x_c, x_d)$.
Explicitly, 
\begin{align*}
\gamma_1 \; \sim \; \mathbb{P}_{\gamma_1}  := 
\one_{\mathcal{E}_1} H_{D_{c,d}} (x_c, x_d)^h 
\; \mathbb{E}_{\Omega;a,b} \big[ \one_{\mathcal{E}_1} H_{D_{c,d}} (x_c, x_d)^h \big]
\; \mathbb{P}_{\Omega;a,b} .
\end{align*}

\item On the other hand, conditionally on this curve $\gamma_1 \sim \mathbb{P}_{\gamma_1}$,
the curve $\gamma_2$ has the law
$\mathbb{P}_{D_{c,d}; x_c, x_d}$ of $\tilde{\eta}$, 
i.e., $\mathbb{P}_{\gamma_2 \; | \; \gamma_1} = \mathbb{P}_{D_{c,d}; c, d}$.

\item Also, by definition~\eqref{eq: Haos construction} and the already established cases $N=1$ and $N=2$, we have
\begin{align*}
\PartF_{\vcenter{\hbox{\includegraphics[scale=0.2]{figures/link-1.pdf}}}} (x_a,x_b,x_c,x_d)
= H_\Omega (x_a, x_b)^h \; \mathbb{E}_{\Omega;a,b} \big[ \one_{\mathcal{E}_1} H_{D_{c,d}} (x_c, x_d)^h \big] .
\end{align*}

\item Combining these facts, we conclude that
\begin{align*}
 \; & \PartF_\alpha(\Omega; x_1, \ldots, x_{2N}) \\
= \; & \PartF_{\vcenter{\hbox{\includegraphics[scale=0.2]{figures/link-1.pdf}}}} (x_a,x_b,x_c,x_d) \times \;
\mathbb{E}_{\gamma_1} \left[ 
\mathbb{E}_{\gamma_2 \; | \; \gamma_1} \left[ \PartF_{\hat{\alpha}_{D_{c,d}} \removeLink \link{c}{d}}(\hat{D}_{c,d}(\gamma_2); \cdots) 
\; \Big| \; \gamma_1 \right] 
\prod_{\substack{ D \textnormal{ c.c of } \Omega \setminus \gamma_1 , \;  D \neq D_{c,d} \\ \cl{D} \cap \{x_1, \ldots, x_{2N}\} \setminus \{x_a, x_b\} \neq \emptyset }}
\PartF_{\hat{\alpha}_D}(D; \cdots) \right] \\
= \; & \PartF_{\vcenter{\hbox{\includegraphics[scale=0.2]{figures/link-1.pdf}}}} (x_a,x_b,x_c,x_d) \times \;
\mathbb{E}_{\gamma_1} \left[ 
\mathbb{E}_{\gamma_2 \; | \; \gamma_1} \bigg[ 
\prod_{\substack{ D \textnormal{ c.c of } \Omega \setminus (\gamma_1 \cup \gamma_2) \\ \cl{D} \cap \{x_1, \ldots, x_{2N}\} \setminus \{x_a, x_b, x_c, x_d\} \neq \emptyset }}
\PartF_{\beta_D}(D; \cdots) \; \Big| \; \gamma_1 \bigg]  \right] \\
= \; & \PartF_{\vcenter{\hbox{\includegraphics[scale=0.2]{figures/link-1.pdf}}}} (x_a,x_b,x_c,x_d) \times \;
\mathbb{E}_{\Omega;a,b,c,d} \left[
\prod_{\substack{ D \textnormal{ c.c of } \Omega \setminus (\gamma_1 \cup \gamma_2) \\ \cl{D} \cap \{x_1, \ldots, x_{2N}\} \setminus \{x_a, x_b, x_c, x_d\} \neq \emptyset }}
\PartF_{\beta_D}(D; \cdots) \right] ,
\end{align*}
where $\beta_D$ are the sub-link patterns of $\alpha \removeLink (\link{a}{b} \cup \link{c}{d})$ 
associated to the components $D \subset \Omega \setminus (\gamma_1 \cup \gamma_2)$.
\end{itemize}
The assertion now follows because this final expression is symmetric with respect to the exchange of $\link{a}{b}$ and $\link{c}{d}$.
\end{proof}

\bigskip

Next, we prove a technical result used in the proof of Lemma~\ref{lem: ASY}.
Fix $\link{a}{b} \in \alpha$ with $a < b$, and 
recall the notations from Appendix~\ref{app:Hao} listed above Lemma~\ref{lem: ASY}.

\begin{lem} \label{lem: ratio poisson}
Let $\kappa \in (0,8)$. 
Let $j,j+1 \notin \{a,b\}$. Then, for any $\xi \in (x_{j-1} , x_{j+2} )$, 
on the event that $x_j$ and $x_{j+1}$ belong to the same connected component of $\bH \setminus \eta$,
we have
\begin{align*}
\lim_{x_j , x_{j+1} \to \xi} \frac{H_{\hat{\bH}_\eta} (x_j, x_{j+1})}{H_\bH (x_j, x_{j+1})}  = 1 , \qquad \textnormal{almost surely.}
\end{align*}
\end{lem}

\begin{proof}
The ratio of Poisson kernels of interest reads
\begin{align} \label{eq:ratio}
\frac{H_{\hat{\bH}_\eta} (x_j, x_{j+1})}{H_\bH (x_j, x_{j+1})}
= \Mob'(x_j)  \Mob'(x_{j+1}) ,
\end{align}
where $\Mob$ is a conformal map from the connected component of $\bH \setminus \eta$
containing $x_j$ and $x_{j+1}$ onto $\bH$, such that $\Mob(x_{j+1}) = x_{j+1}$ and $\Mob(x_j) = x_j$.
Another interpretation of the ratio~\eqref{eq:ratio} is the probability for a Brownian excursion 
connecting the points $x_j$ and $x_{j+1}$ in $\bH$ to stay in $\hat{\bH}_\eta$~\cite{LSW:Conformal_restriction_the_chordal_case, Virag:Brownian_beads}.
In particular, it belongs to $[0,1]$.

By topological reasons and thanks to conformal invariance of the $\SLEk$, 
we may assume that $x_a = 0$, $x_b = 1$, $x_j = R$, and $x_{j+1} = \infty$ without loss of generality.
Thus, it suffices to show that 
\begin{align*}
\lim_{R \to \infty} \Mob'(R) = 1 , \qquad \textnormal{almost surely,}
\end{align*}
where $\Mob$ is the conformal map from the unbounded component of $\bH \setminus \eta$ onto $\bH$,
such that $\Mob(R) = R$ and $\Mob'(\infty) = 1$.
Now, T.~Alberts and M.~Kozdron proved in~\cite[Corollary~\red{1.2}]{Alberts-Kozdron:Intersection_probabilities_for_chordal_SLE_path_and_semicircle} 
that if $R \geq 3$, then we have
\begin{align*}
\mathbb{P} [ \eta \cap \mathcal{C}(0,R) \neq \emptyset ] \asymp R^{1-8/\kappa} ,
\end{align*}
where $\mathcal{C}(0,R) \subset \bH$ is the semi-circle centered at $0$ with radius $R$.
Because $\kappa \in (0,8)$, we therefore have $\mathbb{P} [ \eta \cap \mathcal{C}(0,R) \neq \emptyset ] \to 0$ as $R \to \infty$.
The assertion follows from this.
\end{proof}

\section{\label{app: Coulomb gas} Coulomb gas construction of the pure partition functions}

In this appendix, we summarize an alternative construction of the pure partition functions $\PartF_\alpha$,
which works for $\kappa \in (0,8) \setminus \bQ$.
The functions are constructed in integral form (as so-called Coulomb gas integrals).
The key tool for this construction is a quantum group symmetry 
on the solution space~\eqref{eq: solution space}
of the second order PDE system~\eqref{eq: multiple SLE PDEs}~\cite{Kytola-Peltola:Conformally_covariant_boundary_correlation_functions_with_quantum_group}. 
This symmetry is very useful also for analyzing the solutions 
--- indeed, we used it, e.g., to establish Proposition~\ref{prop: OPE} in Section~\ref{subsec:OPE for ppf}.

The idea is to construct the pure partition functions $\PartF_\alpha$ in terms of Dotsenko-Fateev (Feigin-Fuchs) integrals
\cite{Dotsenko-Fateev:Conformal_algebra_and_multipoint_correlation_functions_in_2D_statistical_models},
which appear in the Coulomb gas formalism of conformal field theory. 
Then, for each $\alpha \in \LP_N$, $\PartF_\alpha$ is proportional to
\begin{align}\label{eq: ansatz for solution to PDEs gives ppf}
\BasisF_\alpha(x_{1},\ldots,x_{2N}) := 
\int_{\Gamma(\alpha)} 
\prod_{1\leq i<j\leq 2N}(x_{j}-x_{i})^{2/\kappa}
\prod_{\substack{1\leq i\leq 2N \\ 1\leq r\leq N}}
(w_{r}-x_{i})^{-4/\kappa}
\prod_{1\leq r<s\leq N}(w_{s}-w_{r})^{8/\kappa} \;
\ud w_1 \cdots \ud w_N ,
\end{align}
for $(x_{1},\ldots,x_{2N}) \in \chamber_{2N}$, 
where the branch of the integrand is chosen in a certain way (see~\cite[Section~\red{3}]{Kytola-Peltola:Conformally_covariant_boundary_correlation_functions_with_quantum_group}).
The key in the construction of $\PartF_\alpha$ is a judicious choice of the integration 
contours $\Gamma(\alpha)$, certain closed $N$-surfaces designed in such a way that 
the functions $\PartF_\alpha$ do satisfy the asymptotics properties in~\eqref{eq: multiple SLE asymptotics}.
We refer  the interested reader 
to~\cite{Kytola-Peltola:Pure_partition_functions_of_multiple_SLEs,
Kytola-Peltola:Conformally_covariant_boundary_correlation_functions_with_quantum_group, Peltola:Basis_for_solutions_of_BSA_PDEs_with_particular_asymptotic_properties}.

\begin{prop} \label{prop in app}
Let $\kappa \in (0,8) \setminus \bQ$. 
The functions appearing in Theorem~\ref{thm::purepartition_existence_forallK} can be written in the form
\begin{align} \label{eq: Ppf cadidate}
 \PartF_\alpha (x_1 , \ldots, x_{2N}) = 
\left( \frac{\Gamma(2-8/\kappa)}{\Gamma(1-4/\kappa)^2} \right)^N
 \times \BasisF_\alpha (x_1 , \ldots, x_{2N}) , \qquad \textnormal{for } \alpha \in \LP_N.
\end{align}
\end{prop}
\begin{proof}
\cite[Theorem~\red{4.1}]{Kytola-Peltola:Pure_partition_functions_of_multiple_SLEs} 
shows that for any $\alpha \in \LP_N$, the right side of~\eqref{eq: Ppf cadidate} belongs to the solution space 
$\Sol_N$ and the asymptotic properties~\eqref{eq: multiple SLE asymptotics} hold.
Uniqueness of the functions with these properties, Proposition~\ref{prop::purepartition_unique}, 
then implies that $\PartF_\alpha$ must be equal to the functions appearing in Theorem~\ref{thm::purepartition_existence_forallK}.
\end{proof}

The restriction that $\kappa$ is irrational is needed because the current form of 
the ``spin chain~--~Coulomb gas correspondence''
established in~\cite[Theorems~\red{4.16}~and~\red{4.17}]{Kytola-Peltola:Conformally_covariant_boundary_correlation_functions_with_quantum_group}
requires the representation theory of the quantum group $\Uqsltwo$ to be semisimple (here, $q = e^{\ii \pi 4 / \kappa}$).
In principle, the functions thus obtained could be analytically continued to include all $\kappa \in (0,8)$,
but the explicit continuation is not obvious, due to delicate cancellations of infinities and zeroes.
On the other hand, 
because of the non-semisimplicity of the representation theory of $\Uqsltwo$ for rational $\kappa$,
one observes interesting phenomena in these cases.

We also emphasize that smoothness of $\PartF_\alpha$ is immediate from the Coulomb gas integral construction,
and the asymptotics can also be analyzed in a powerful and systematic way. 
However, it seems very difficult to show in general that the functions $\PartF_\alpha$ obtained from~\eqref{eq: Ppf cadidate}
are positive. For $\kappa \leq 6$, this latter property is provided by the probabilistic construction of $\PartF_\alpha$ 
discussed in Appendix~\ref{app:Hao}, combined with the very strong fact from 
Proposition~\ref{prop::purepartition_unique} that both constructions indeed give the same functions.
Also the ``strong'' power law bound $\mathrm{(B)}$ given in~\eqref{eqn::partitionfunction_positive} 
is not obvious from~\eqref{eq: Ppf cadidate} at all, whereas for $\kappa \leq 6$, it is manifest in the probabilistic construction
(however difficulties do occur when $\kappa > 6$).

\end{appendices}

\bibliographystyle{annotate}

\renewcommand{\bibnumfmt}[1]{\makebox[5.3em][l]{[#1]}}

\newcommand{\etalchar}[1]{$^{#1}$}

\end{document}